\documentclass[letterpaper, 10 pt, journal, onecolumn]{IEEEtran}  %

\IEEEoverridecommandlockouts                              %
\usepackage{cite}
\usepackage{amsmath,amssymb,amsfonts,bbold}
\usepackage{graphicx}
\usepackage{textcomp}
\usepackage{xcolor}
\usepackage{algpseudocode}
\usepackage{algorithm}
\usepackage{mathtools}
\usepackage{amsthm}
\usepackage{booktabs}

\newtheorem{theorem}{\textbf{Theorem}}[section]
\newtheorem{remark}{\textbf{Remark}}[theorem]

\newtheorem{lemma}[theorem]{\textbf{Lemma}}
\newtheorem{proposition}[theorem]{\textbf{Proposition}}

\newtheorem{definition}{\textbf{Definition}}[section]
\newtheorem{assumption}{\textbf{Assumption}}[section]

\renewcommand{\Pr}{\mathop{\mathrm{Pr}}\nolimits}

\newcommand{\ie}{\emph{i.e.}}

\begin{document}

\setlength{\emergencystretch}{2em} %

\title{Universal Estimation of the Fisher Information
for Processes with Memory
\thanks{This work has been partially supported by the Swedish Research Council under contract number 2023-05170. The authors are with the Department of Decision and Control Systems, KTH Royal Institute of Technology, SE-100 44 Stockholm, Sweden. Emails: \texttt{\{yibos, crro\}@kth.se}}}
\author{Yibo Shi, and Cristian R. Rojas}

\maketitle

\begin{abstract}
The Fisher information matrix (FIM) determines the
accuracy attainable in an estimation problem. For a parametric process
with memory it has a closed form only for narrow classes of models, and
it cannot be computed when the model is available only as a simulator.
This paper estimates the Fisher information rate of such a process, the
per-observation limit of its FIM, from simulator output alone and
without knowledge of the memory length. The construction uses the
identity between the FIM and the local curvature of the
Kullback--Leibler (KL) divergence. Directional divergence rates are
estimated by context-tree weighting (CTW), and the matrix is recovered
from them by least squares. We establish strong consistency in an
iterated limit and a
finite-sample bound, uniform over every working depth at least as large
as the true memory length, which gives a mean error of order
$n^{-1/3}$. The analysis also delivers a finite-sample rate for
CTW-based estimation of the KL divergence rate. Experiments confirm the
predicted exponents, show that the accuracy does not degrade as the
working depth grows, and use the estimate to plan a sampling budget
that attains a prescribed precision. The estimator stays accurate on a
hidden-state source that has no finite memory length.
\end{abstract}

\begin{IEEEkeywords}
Fisher information, context-tree weighting,
Kullback--Leibler divergence, sample-size determination,
simulation-based inference, universal source coding.
\end{IEEEkeywords}
\section{Introduction}
\label{sec:introduction}

The Fisher information matrix (FIM) of a parametric
statistical model is one of the basic quantities of statistics and
signal processing, and
its importance has grown with the spread of such models
across fields. In statistics it gives the
Cram\'er--Rao lower bound on the variance of unbiased estimators and
the asymptotic covariance of the maximum likelihood
estimator~\cite{lehmann1998theory}, and it underlies
experimental design and sample-size determination by deciding how much
data a prescribed precision requires~\cite{adcock1997sample}. In
signal processing it is the performance benchmark for nearly every
estimation problem~\cite{kay1993}. In machine learning it defines the
natural gradient~\cite{amari1998natural}, it is the curvature that
second-order optimizers approximate~\cite{martens2015kfac}, and it
measures the importance of network parameters in methods against
catastrophic forgetting~\cite{kirkpatrick2017ewc}. In Bayesian analysis
it defines the Jeffreys prior~\cite{jeffreys1946invariant,
Clarke-Barron-94}. All of these applications require the
value of the FIM, and that value is increasingly hard to
obtain. First, the FIM has a closed form only for a
narrow class of model structures, and as models grow more complex most
of them fall outside that class. Second, models are increasingly given
as parametric black boxes, simulators
indexed by the parameter that generate data at any
chosen parameter value but whose likelihood can neither be
written down nor evaluated. This is the standard setting of
simulation-based
inference~\cite{Cranmer2020SBI,Sisson2018ABCHandbook,Wood2010NatureSynLik}.
Furthermore, dynamic models with memory are increasingly
important. For them
the FIM is the inverse asymptotic covariance of
prediction-error and maximum-likelihood
identification~\cite{Ljung:99,Soderstrom-Stoica-89} and the objective
of experiment design in system identification~\cite{valenzuela2017robust},
but the nature of these models makes it harder to obtain. For state-space models, for instance, it has to be approximated by
specially designed particle
methods~\cite{Poyiadjis2011Biometrika,Kantas2015StatSci}.

Estimators of the FIM that work under these
circumstances are therefore a real need, in particular for models that
are complex, have memory, or are available only as black boxes. Under
the usual regularity conditions the FIM is the
covariance of the \emph{score function}~\cite{lehmann1998theory}, which is the gradient of the log-likelihood with respect
to the parameter. Most existing estimators start from this identity,
along two tracks. The first track uses an explicit score function or a computable log-likelihood: the score covariance is
estimated by the sample average of its outer products, or the log-likelihood is differentiated
numerically~\cite{efron1978assessing,spall2005monte,das2007efficient}.
The second track needs to estimate the score function itself as
well, by approximating it from data and then taking its sample
covariance, as particle methods do for state-space
models~\cite{Poyiadjis2011Biometrika,Kantas2015StatSci,valenzuela2017robust}.
Both approaches share one premise, namely, that the structure of the model is
known: the first approach requires the likelihood to be written in a given
parameterization, whereas the second one requires a concrete
structure, such as a state-space form, to build the approximation,
and in system identification the asymptotic covariance is
likewise derived separately for each model
structure~\cite{Ljung:99}. Each different structure calls for a new
derivation, so, as a result, both approaches fail for a black-box model,
which offers no structure to exploit.

A different route avoids the score altogether and relies
only on samples generated by the model. It rests on a classical
fact: the FIM is the local curvature of divergence in
parameter space. For a small perturbation of the parameter, the
\emph{Kullback--Leibler (KL) divergence} between the
output distributions of the model at the reference and at the
perturbed parameter value is, to second order, a quadratic form on the FIM~\cite{cover1999elements}. The FIM can therefore
be recovered from the curvature, which is obtained by
differencing divergences between neighbouring parameter values.
Estimating the FIM thus becomes estimating divergences.
Divergences are defined through the likelihood, as the
FIM is, and cannot be computed either. However, the KL
divergence enjoys an interesting property: it is the
expected log-likelihood ratio, the average gap between how well a
sample is explained by its own distribution and by the other one. It
can therefore be estimated from samples alone, by fitting a predictor
to a sample from one distribution and testing it on a sample from the
other. Estimating divergences from samples is for that reason a
problem with established tools, such as nearest-neighbour~\cite{wang2009knn}
and variational~\cite{nguyen2010estimating} estimators for independent
samples and estimators built on universal compression for
sequences~\cite{ziv1993relative,cai2006universal}. No
comparable tools exist for the FIM. The whole route
therefore uses only samples, no likelihood, no score, and no
assumption on the structure of the model. It has been
taken for independent and identically distributed
(i.i.d.) samples: the
divergences are estimated nonparametrically and the FIM is
recovered by least squares~\cite{berisha2014empirical}. That
estimator, however, requires i.i.d.~samples, and the result
is consistency only. For processes with memory, which is the class for which
the FIM is hardest to obtain, no such estimator
has been proposed yet.

This paper follows this route to construct a universal,
data-driven estimator of the FIM for processes with
memory. Universality means that the estimator works for a
whole class of processes without knowing which one. In particular, it does not need to know the \emph{memory length} of the process, that is, the number of
past observations on which the next observation depends, which is
usually unknown in practice. The estimator also comes with theoretical
guarantees.

This route needs two components. The recovery of the FIM
is direct, since the directional curvatures are linear in its entries
and least squares over enough directions recovers the
matrix~\cite{berisha2014empirical}. The remaining difficulty is the
divergence estimator: it must work from data, handle dependent
observations, and not require the memory length of the process to be
known in advance.

Three kinds of divergence estimators are available. Plug-in estimators assume a fixed
memory length, estimate the conditional probabilities over
many past observations by empirical frequencies, and substitute
them into the expression for the divergence~\cite{gao2008estimating}. They are the simplest methods, and work well when the assumed length is
correct and the sample is large, but the assumed length is a design
variable, the true memory length is a property of the model, and a
good estimate requires the two lengths to match, while the true length is
unknown. Variational and neural
estimators~\cite{nguyen2010estimating,belghazi2018mutual} need
no assumed memory length, but for processes
with memory they come without consistency or error
guarantees. The third kind builds on universal
compressors. A universal compressor approaches the optimal code length
for every process in a given class without being told which process
generated the data. The code lengths lead to estimates of the entropies of the data, and their
difference correspond to an estimate of the KL divergence. This family of methods has
two branches, the Lempel--Ziv methods~\cite{ziv1993relative} and
context-tree weighting (CTW)~\cite{willems2002contextB,cai2006universal}.
For this paper, we select CTW, for three properties: It is a sequential predictor,
observation by observation, so a predictor trained on one sample can
directly predict another one, which is exactly what the divergence
estimation above requires. It does not commit to one memory length but
averages automatically over all memory lengths up to a user-chosen
depth $L$, so the user only has to choose $L$ at least as large as the true memory
length, without knowing its exact value. Its guarantees are explicit and valid for finite samples, and they stop deteriorating with $L$ once $L$ exceeds the
true memory length~\cite{Kon24}. In empirical comparisons it is also currently
the most accurate universal-compressor~\cite{gao2008estimating}. CTW has been
used to build consistent estimators of the KL divergence
between sequences in \cite{cai2006universal}, but that work stops at the
divergence, gives no rate, and is not connected to the FIM. Recent developments of the CTW family, including
extensions to real-valued time series and to processes
with infinite memory~\cite{papageorgiou2023context,papageorgiou2025bctssm}, give
it the potential to serve as a universal divergence estimator in a
general sense. This paper starts from classical CTW, where the theory
is most complete, and accordingly states its theory for stationary
processes with finite-valued observations and finite memory; the
infinite-memory case is probed only in experiments. Finally, the
theoretical guarantees for the resulting estimator, namely, consistency and
convergence rate, are the main technical contributions of the paper.

The contributions of this paper are as follows:
\begin{enumerate}
\item A CTW-based estimator of the FIM that requires
only simulator access, for stationary processes with finite-valued
observations and finite memory.
\item Strong consistency of the estimator in an iterated
limit.
\item A finite-sample error rate for the estimator, together with the
step size that attains it. The bound is numerically tight.
As an extra contribution, the analysis includes a finite-sample error
bound for CTW-based estimators of the KL divergence.
\item A numerical evaluation of the estimator, which verifies its performance even outside the regime covered by the theory, and shows its value in a
sample-size determination task.
\end{enumerate}

Section~\ref{sec:problem} formulates the problem and states the
assumptions, Section~\ref{sec:method} constructs the estimator,
Section~\ref{sec:theory} establishes consistency and the rates, and
Section~\ref{sec:experiments} reports the numerical experiments,
and Section~\ref{sec:conclusion} concludes.
Proofs, the details of the CTW algorithm, and additional experiments
are given in the appendices.

\textit{Notation.} All logarithms are natural. The
Euclidean norm of a vector is $\|\cdot\|_2$ and the Frobenius norm of a
matrix is $\|\cdot\|_F$. For positive sequences $a_n$ and $b_n$,
$a_n=O(b_n)$ means that the ratio $a_n/b_n$ is bounded above, and
$a_n\asymp b_n$ that it is bounded above and below by positive
constants.

\section{Problem Formulation}
\label{sec:problem}

This section defines the Fisher information rate, the
per-observation version of the FIM that is the relevant object for a
process with memory. It states the standing assumptions and formulates
the estimation problem.

\subsection{The Fisher Information Rate}
\label{subsec:setup}

Let $Y$ be a random variable taking values in an arbitrary finite set
$\mathcal A$ (called an \emph{alphabet}). For a parameter
$\theta\in\Theta\subset\mathbb R^d$, write its
\emph{probability mass function} (pmf) as
$p_\theta(y):=p(y;\theta):=\Pr_\theta\{Y=y\}$.
On a finite alphabet, the pmf uniquely determines a \emph{probability law} $P_\theta$:
the probability of any subset of $\mathcal A$ is the sum of
$p_\theta(y)$ over the elements in that subset.

\begin{definition}[Fisher information matrix]
\label{def:fim}
The \emph{Fisher information matrix} (FIM) at $\theta$, denoted by
$J_\theta$, is
\begin{equation}
J_\theta
\;:=\;
\mathbb{E}_\theta\!\left[
\nabla_\theta \log p_\theta(Y)\,
\nabla_\theta^{\top} \log p_\theta(Y)
\right],
\label{eq:fisher-info-def}
\end{equation}
where $\nabla_\theta$ denotes the gradient with respect to $\theta$
and $\mathbb E_\theta[f(Y)]
:=\sum_{y\in\mathcal A}f(y)p_\theta(y)$.
\end{definition}

A \emph{stochastic process} over the alphabet $\mathcal A$ is a collection
$\{Y_t\}_{t\in\mathbb Z}$ of $\mathcal A$-valued random variables. We
consider a stationary process, as formally imposed in
Assumption~\ref{ass:source}, and let $P_\theta$ denote its probability
law on $\mathcal A^{\mathbb Z}$. We use $\mathbb E_\theta$ for
expectation under $P_\theta$ or the relevant marginal thereof. For
integers $i\le j$, define the random
block $Y_{i:j}:=(Y_i,Y_{i+1},\ldots,Y_j)$,
with realization
$y_{i:j}:=(y_i,y_{i+1},\ldots,y_j)\in\mathcal A^{\,j-i+1}$.
By stationarity, the distribution of a block depends on its length but
not on its starting time; in particular,
\[
Y_{i:i+n-1}\overset{d}{=}Y_{1:n},
\qquad i\in\mathbb Z.
\]
For $n\in\mathbb N$, let $P_\theta^{(n)}$ denote the common
distribution of every length-$n$ block. Since the alphabet is finite,
this distribution has the pmf
\[
p_\theta^{(n)}(y_{1:n})
:=P_\theta^{(n)}(\{y_{1:n}\})
=\Pr\nolimits_\theta\{Y_{1:n}=y_{1:n}\},
\qquad y_{1:n}\in\mathcal A^n.
\]

\begin{definition}[Fisher information rate]
\label{def:fim-rate}
At parameter value $\theta$, the \emph{length-$n$ FIM} of the process is
\begin{equation}
J_\theta^{(n)}
\;:=\;
\mathbb{E}_{P_\theta^{(n)}}\!\left[
\nabla_\theta\log p_\theta^{(n)}(Y_{1:n})
\nabla_\theta^{\!\top}\log p_\theta^{(n)}(Y_{1:n})
\right].
\label{eq:fim-n-pn}
\end{equation}
When observations are i.i.d., $J_\theta^{(n)}=nJ_\theta$. For general, non-i.i.d.\ observation sequences, the relevant quantity is the per-symbol \emph{Fisher information rate}
\begin{equation}
I(\theta)
\;:=\;
\lim_{n\to\infty}\frac{1}{n}\,J_\theta^{(n)},
\label{eq:fim-rate}
\end{equation}
provided the limit exists. 
Under the smoothness condition of Assumption~\ref{ass:reg}(ii) below,
the FIM equals the negative of the expected Hessian of the
log-likelihood
\begin{equation}
    J_\theta^{(n)}
=-\mathbb E_{P_\theta^{(n)}}\bigl[\nabla_\theta^2\log
p_\theta^{(n)}(Y_{1:n})\bigr].
\end{equation}
Thus $n^{-1}J_\theta^{(n)}$ is the expected curvature (negative Hessian)
of the \emph{average log-likelihood} at $\theta$, and $I(\theta)$ is its
limiting per-symbol value. In the i.i.d.\ case, $I(\theta)$ reduces to
the ordinary FIM $J_\theta$.
\end{definition}

Our goal is to estimate the Fisher information rate $I(\theta)$ for non-i.i.d.~sources, in a setting where the likelihood and its
derivatives are inaccessible and only simulated trajectories of the
process are available. We make this precise through the following
standing assumptions.

\subsection{Assumptions}
\label{subsec:assumptions}

\begin{assumption}[Simulator access]\label{ass:sim}
For any finite collection of parameter values
$\theta_1,\ldots,\theta_R\in\Theta$ and any $n\in\mathbb{N}$, the
simulator can generate mutually independent trajectories
$Y_{1:n}^{(r)}\sim P_{\theta_r}^{(n)}$, $r=1,\ldots,R$.
The observations within each trajectory may be temporally dependent.
Direct evaluation of the marginal pmfs $p_{\theta_r}^{(n)}$ and their
scores $\nabla_{\theta_r}\log p_{\theta_r}^{(n)}$ is unavailable or
impractical.
\end{assumption}
\begin{assumption}[Finite alphabet, finite memory, and stationarity]
\label{ass:source}
Under $P_\theta$, the process $(Y_t)_{t\in\mathbb Z}$ is stationary,
takes values in a finite alphabet $\mathcal{A}$, and is a Markov source
of order at most $L^\star<\infty$; that is,
\begin{equation*}
\Pr_\theta\bigl(Y_t=a\mid Y_{-\infty:t-1}\bigr)
=
\Pr_\theta\bigl(Y_t=a\mid Y_{t-L^\star:t-1}\bigr)
\qquad\text{almost surely}
\end{equation*}
for every $a\in\mathcal A$.
The state process
$S_t:=Y_{t-L^\star:t-1}$, obtained by sliding a window of length
$L^\star$ along the sequence, is a first-order Markov chain on
$\mathcal{A}^{L^\star}$, which is assumed to be irreducible and
aperiodic.
\end{assumption}

\begin{assumption}[Regularity and local identifiability]
\label{ass:reg}
Fix a parameter value $\theta_0\in\operatorname{int}(\Theta)$ at which
$I(\theta_0)$ is to be estimated. There exist an open neighborhood
$\mathcal N(\theta_0)\subset\Theta$ and a constant
$\delta\in(0,\tfrac12)$ such that the following hold for every
$\theta\in\mathcal N(\theta_0)$:
\begin{enumerate}
\item[\textnormal{(i)}]
\emph{(Uniform regularity.)}
For every state $s\in\mathcal A^{L^\star}$ and every
$a\in\mathcal A$,
\[
\Pr_\theta(Y_t=a\mid Y_{t-L^\star:t-1}=s)
\in[\delta,1-\delta].
\]
\item[\textnormal{(ii)}]
\emph{(Smoothness.)}
For every $n\in\mathbb N$ and every $y_{1:n}\in\mathcal A^n$, the map
\[
\theta\mapsto
\log p_\theta^{(n)}(y_{1:n})
\]
is three times continuously differentiable on
$\mathcal N(\theta_0)$.
Since the alphabet is finite, expectations are finite
sums, and differentiation may always be interchanged with them;
consequently, the score identity
\[
\mathbb E_\theta\!\left[
\nabla_\theta\log p_\theta^{(n)}(Y_{1:n})\right]=0
\]
and the Fisher information identity
\[
-\mathbb E_\theta\!\left[
\nabla_\theta^2\log p_\theta^{(n)}(Y_{1:n})\right]
=J_\theta^{(n)}
\]
hold for every $n\in\mathbb N$ and
$\theta\in\mathcal N(\theta_0)$.
\item[\textnormal{(iii)}]
\emph{(Local identifiability.)}
If $P_\theta=P_{\theta'}$ for
$\theta,\theta'\in\mathcal N(\theta_0)$, then
$\theta=\theta'$.
\end{enumerate}
\end{assumption}

The next few remarks clarify the role of each assumption.

Assumption~\ref{ass:sim} formalizes the \emph{black-box} setting considered throughout: the data-generating model is accessible only through a simulator, while closed-form evaluation of the likelihood and its parameter derivatives is unavailable. This is the regime in which classical, score-based computation of the FIM~\eqref{eq:fisher-info-def}-\eqref{eq:fim-n-pn} cannot be applied.

Assumption~\ref{ass:source} specifies the source class considered in this
paper: stationary finite-alphabet processes with bounded memory. This is
the standard setting for CTW-based coding and divergence estimation, and
the i.i.d.\ case is recovered by $L^\star=0$. Irreducibility and
aperiodicity ensure that the finite-state chain $S_t$ is ergodic. The
bounded-memory restriction is made for compatibility with the depth-$L$
CTW estimator introduced in Section~\ref{subsec:kl-est}; extensions
beyond finite memory are left outside the scope of this paper.

Assumption~\ref{ass:reg} collects the local regularity conditions used
near the fixed target value $\theta_0$. Part (i) keeps the transition
probabilities nondegenerate, part (ii) supplies the differentiability and
Fisher identities needed for the local KL--Fisher expansion, and
part (iii) rules out local non-identifiability. These
conditions are standard for establishing the well-posedness and
consistency of Fisher-information estimators~\cite{berisha2014empirical}.
They are in force throughout the paper: the consistency
results of Section~\ref{sec:theory} use them qualitatively, and the
finite-sample analysis of Section~\ref{subsec:rate} uses them again
quantitatively (in particular the lower bound $\delta$), adding the
structural conditions it alone needs at their point of use.

\subsection{Problem Statement}
\label{subsec:problem-statement}

Under Assumptions~\ref{ass:sim}--\ref{ass:reg}, the problem addressed in this paper is the following: \emph{given only simulator access to the process laws $P_\theta$, estimate the Fisher information rate $I(\theta_0)$ at the fixed parameter value $\theta_0$, without evaluating the likelihood $p_{\theta_0}^{(n)}$ or the score $\nabla_\theta\log p_\theta^{(n)}|_{\theta=\theta_0}$.} For finite $n$, the normalized length-$n$ FIM $J_{\theta_0}^{(n)}/n$ serves as a finite-sample proxy for $I(\theta_0)$, and is recovered in the i.i.d.\ case as the ordinary FIM $J_{\theta_0}$. The estimator we propose, developed in Section~\ref{sec:method}, addresses this problem; its consistency and finite-sample guarantees are established in Section~\ref{sec:theory}.

\section{CTW-Based Estimation of the Fisher Information}
\label{sec:method}

We now describe the proposed estimator of the Fisher information rate
$I(\theta_0)$. The estimator is \emph{score-free} and
\emph{likelihood-free}: it accesses the data-generating process only
through the simulator of Assumption~\ref{ass:sim}, and does not
evaluate the marginal pmfs $p_{\theta_0}^{(n)}$ or their parameter
derivatives.

Under Assumption~\ref{ass:sim}, the score and the likelihood are unavailable, so the Fisher information cannot be obtained directly from its definition~\eqref{eq:fisher-info-def}. The estimator instead exploits a second-order relationship between the Fisher information and the KL divergence: $I(\theta_0)$ is the local curvature of the KL divergence between process laws near $P_{\theta_0}$ (Section~\ref{subsec:kl-fim-equivalence}). The KL divergence itself depends on the likelihood, but it can be written as the difference between a cross-entropy rate and an entropy rate. We estimate these two rates following the CTW-based universal divergence-estimation construction of Cai et al.~\cite{cai2006universal}. The reference trajectory, simulated at $\theta_0$, is encoded by a sequential CTW code, giving the entropy rate. CTW scorers frozen on trajectories simulated at perturbed parameter values then score that same trajectory, giving the cross-entropy rate (Section~\ref{subsec:kl-est}). Directional KL estimates obtained along a set of probing directions are then mapped to $I(\theta_0)$ through a least-squares recovery step (Section~\ref{subsec:ls-recovery}). The complete procedure is summarized in Algorithms~\ref{alg:ctw-kl} and~\ref{alg:fim-recovery}.

\subsection{Relation Between the KL Divergence and the Fisher Information Matrix}
\label{subsec:kl-fim-equivalence}

The first reduction replaces the Fisher information by KL divergences.

\begin{definition}[Block information measures]
\label{def:block-information}
For two probability laws $P^{(n)}$ and $Q^{(n)}$ on
$\mathcal{A}^n$, with pmfs $p^{(n)}$ and $q^{(n)}$, respectively, the
\emph{(block) cross-entropy} is
\begin{equation}
H\bigl(P^{(n)},Q^{(n)}\bigr)
:= -\,\mathbb{E}_{P^{(n)}}\!
\left[\log q^{(n)}(Y_{1:n})\right],
\label{eq:ce-def}
\end{equation}
and the \emph{entropy} of $P^{(n)}$ is its diagonal value,
$H(P^{(n)}):=H(P^{(n)},P^{(n)})$. The
\emph{KL divergence (relative entropy)} of $P^{(n)}$ from $Q^{(n)}$
is the difference
\begin{equation}
D\bigl(P^{(n)}\,\|\,Q^{(n)}\bigr)
:=H\bigl(P^{(n)},Q^{(n)}\bigr)-H\bigl(P^{(n)}\bigr).
\label{eq:kl-def}
\end{equation}
For fixed block length $n$, define the \emph{normalized block
cross-entropy}, \emph{entropy}, and \emph{KL divergence}, respectively, as
\begin{equation}
\bar{H}_n\bigl(P^{(n)},Q^{(n)}\bigr)
:= \tfrac{1}{n}H\bigl(P^{(n)},Q^{(n)}\bigr),
\quad
\bar{H}_n\bigl(P^{(n)}\bigr)
:= \tfrac{1}{n}H\bigl(P^{(n)}\bigr),
\quad
\bar{D}_n\bigl(P^{(n)}\,\|\,Q^{(n)}\bigr)
:= \tfrac{1}{n}D\bigl(P^{(n)}\,\|\,Q^{(n)}\bigr).
\label{eq:per-sample}
\end{equation}
\end{definition}

\begin{definition}[Information rates]
\label{def:information-rates}
For the sequence models considered here
(Assumption~\ref{ass:source}), the \emph{entropy rate},
\emph{cross-entropy rate}, and \emph{KL divergence rate} between
process laws at two parameter values $\theta$ and $\theta'$ are
\begin{align}
\bar H(P_\theta)
&:=\lim_{n\to\infty}\bar H_n(P_\theta^{(n)}), \nonumber\\
\bar H(P_\theta,P_{\theta'})
&:=\lim_{n\to\infty}
\bar H_n(P_\theta^{(n)},P_{\theta'}^{(n)}), \nonumber\\
\bar D(P_\theta\,\|\,P_{\theta'})
&:=
\lim_{n\to\infty}
\bar D_n(P_\theta^{(n)}\,\|\,P_{\theta'}^{(n)}),
\label{eq:kl-rate-def}
\end{align}
whenever these limits exist. Equivalently,
\begin{equation}
\bar D(P_\theta\,\|\,P_{\theta'})
\;=\;
\bar H(P_\theta,P_{\theta'})-\bar H(P_\theta).
\label{eq:kl-rate-entropy}
\end{equation}
\end{definition}

We now relate these information measures to the Fisher information at
the fixed target parameter $\theta_0$. For a sufficiently small
perturbation $\Delta\in\mathbb R^d$, a second-order Taylor expansion of
the length-$n$ log-pmf
$\log p_\theta^{(n)}(y_{1:n})$ around $\theta_0$ gives
\begin{equation}
D\bigl(P_{\theta_0}^{(n)}\,\|\,P_{\theta_0+\Delta}^{(n)}\bigr)
=
\tfrac12\,\Delta^\top J_{\theta_0}^{(n)}\Delta
+O\!\bigl(n\|\Delta\|^3\bigr).
\label{eq:kl-fim-block}
\end{equation}
Thus, both the pmfs $p_{\theta_0}^{(n)}$ and
$p_{\theta_0+\Delta}^{(n)}$ underlying the KL divergence and the
corresponding FIM $J_{\theta_0}^{(n)}$ depend on the block length $n$.
Dividing~\eqref{eq:kl-fim-block} by $n$ yields the normalized block
relationship
\begin{equation}
\bar D_n\bigl(P_{\theta_0}^{(n)}\,\|\,P_{\theta_0+\Delta}^{(n)}\bigr)
=
\tfrac12\,\Delta^\top
\left(\tfrac1n J_{\theta_0}^{(n)}\right)\Delta
+O\!\bigl(\|\Delta\|^3\bigr).
\label{eq:kl-fim-normalized-block}
\end{equation}
Finally, taking $n\to\infty$ in~\eqref{eq:kl-fim-normalized-block}
gives the corresponding rate-level relationship
\begin{equation}
\bar{D}\bigl(P_{\theta_0}\,\|\,P_{\theta_0+\Delta}\bigr)
= \tfrac{1}{2}\,\Delta^\top I(\theta_0)\,\Delta
+ O\!\bigl(\|\Delta\|^3\bigr),
\label{eq:kl-fim-rate}
\end{equation}
where $I(\theta_0)$ is the Fisher information
rate~\eqref{eq:fim-rate}. Equations~\eqref{eq:kl-fim-block}-\eqref{eq:kl-fim-rate} exhibit the progression from block laws and
their length-$n$ FIM to the corresponding per-symbol limits. The
rate-level quadratic relationship is the foundation of the proposed
estimator: it shows that $I(\theta_0)$ can be recovered from the KL
divergence rates between $\theta_0$ and nearby parameter values, with
no access to scores or likelihoods. The following
lemma makes this precise, including a uniform bound on the third-order
remainder; the proof is given in Appendix~\ref{app:aux}.

\begin{lemma}[Local KL--FIM expansion]
\label{lem:kl-fim}
Under Assumptions~\ref{ass:source}--\ref{ass:reg}, there
exist constants $\varepsilon_0>0$ and $C_3<\infty$, depending only on
the local model family in a neighborhood of $\theta_0$, such that for
every $\Delta\in\mathbb R^d$ with $0<\|\Delta\|_2\le\varepsilon_0$,
\begin{equation}
\bar{D}\bigl(P_{\theta_0}\,\|\,P_{\theta_0+\Delta}\bigr)
= \tfrac12\,\Delta^\top I(\theta_0)\,\Delta + r(\Delta),
\qquad
|r(\Delta)|\le C_3\,\|\Delta\|_2^3.
\label{eq:kl-fim-expansion}
\end{equation}
\end{lemma}

It remains to estimate the KL divergence rate
$\bar D(P_{\theta_0}\,\|\,P_{\theta_0+\Delta})$ for small perturbations
$\Delta$ from simulator output alone.
This is the subject of the next subsection.

\subsection{CTW-Based Universal Estimation of the KL Divergence Rate}
\label{subsec:kl-est}

The expansion~\eqref{eq:kl-fim-rate} reduces the estimation of
$I(\theta_0)$ to that of KL divergence rates
$\bar{D}(P_{\theta_0}\,\|\,P_{\theta_0+\Delta})$. These rates still
involve the marginal pmfs $p_{\theta_0}^{(n)}$ and
$p_{\theta_0+\Delta}^{(n)}$, which are inaccessible under
Assumption~\ref{ass:sim}. Since
\begin{equation}
\bar{D}(P_{\theta_0}\,\|\,P_{\theta_0+\Delta})
=
\bar{H}(P_{\theta_0},P_{\theta_0+\Delta})
-
\bar{H}(P_{\theta_0}),
\label{eq:kl-split}
\end{equation}
it is enough to estimate an entropy rate and a cross-entropy rate from
simulator trajectories. We do this with the CTW-based universal divergence
estimator of Cai et al.~\cite{cai2006universal}. We first recall the
CTW coding mechanism and then describe the frozen scoring protocol used
for KL estimation.

\subsubsection{Context-tree weighting}

We use the \emph{context-tree weighting (CTW)} algorithm of Willems,
Shtarkov, and Tjalkens~\cite{willems2002contextB} as the universal
coding component of our estimator. Its underlying data structure is the
\emph{context tree} of working depth $L$: the complete
$|\mathcal A|$-ary tree whose nodes are the strings
$s\in\mathcal A^{\ell}$, $0\le \ell\le L$. A node $s$ at depth $\ell$
stands for the $\ell$ most recent symbols of the past and is called a
\emph{context}. In particular, the state
$S_t=Y_{t-L^\star:t-1}$ of Assumption~\ref{ass:source} is the context of
depth $L^\star$ at time $t$. The root is the empty string, the children of a context
$s$ are its one-symbol extensions $as$, $a\in\mathcal A$, and the leaves
are the contexts of full length $L$. The working depth is chosen so
that $L\ge L^\star$. Since the conditional law of the next symbol
depends on the past only through its $L^\star$ most recent symbols,
the working tree
contains every context the source actually uses.

While processing a string, each node $s$ stores the counts of the
symbols that followed the context $s$. These counts define, at each
node, a local \emph{Krichevsky--Trofimov (KT)}
predictor~\cite{krichevsky1981performance}: a smoothed empirical
estimate of the conditional distribution of the next symbol given the
context $s$, assigning strictly positive probability to every symbol,
including those not yet observed. Its exact form and the Bayesian
argument behind it are given in Appendix~\ref{app:ctw}. CTW then
combines the KT predictors of different context lengths, from the
leaves up to the root, into a single probability
$Q^{\mathrm{CTW}}_{L}(y_{1:i})$ assigned to the processed string
$y_{1:i}$, from which the next symbol is predicted sequentially by
\[
Q^{\mathrm{CTW}}_{L}(y_i\mid y_{1:i-1})
=
\frac{Q^{\mathrm{CTW}}_{L}(y_{1:i})}{Q^{\mathrm{CTW}}_{L}(y_{1:i-1})}.
\]

Three features make CTW suitable for the KL-estimation step below:

\begin{itemize}
    \item \textbf{Linear-time implementation.} The CTW recursion processes one symbol at a time and evaluates the required block probabilities at total cost linear in the block length $n$.
    \item \textbf{No prior knowledge of the memory length.} The CTW recursion combines the context-tree models of depth at most $L$. Consequently, choosing $L\ge L^\star$ allows the true bounded-memory source to be represented without requiring prior knowledge of its exact memory length. Although the original algorithm was formulated for binary sources, it has been extended to larger alphabets~\cite{tjalkens1993sequential}, which is the form used here; the construction further extends to unbounded tree depth~\cite{willems2002contextE}, and, through quantization, can serve as a building block for sources with continuous observation spaces~\cite{papageorgiou2023context}.
    \item \textbf{Positive probabilities.} For every history
$y_{1:t}$ and every next symbol $a\in\mathcal A$, the sequential CTW
predictor satisfies
\begin{equation}
Q^{\mathrm{CTW}}_{L}\bigl(a\mid y_{1:t}\bigr)
\;\ge\;\frac{1}{2t+|\mathcal{A}|}.
\label{eq:ctw-positivity-method}
\end{equation}
Consequently, every CTW block probability assigned to a finite
string is strictly positive, so the negative log scores used in
\eqref{eq:entropy-est} and \eqref{eq:cross-entropy-est} are finite on
every sample path.
\end{itemize}

The CTW properties above are established in
Appendix~\ref{app:ctw}.

For entropy estimation, CTW can be used as it is: the per-symbol code length
$-\tfrac{1}{n}\log Q^{\mathrm{CTW}}_{L}(z_{1:n})$ of the observed block
directly estimates the entropy rate.
Let $Z_{1:n}\sim P_{\theta_0}^{(n)}$ be a reference path. The entropy
rate of $P_{\theta_0}$ is estimated by the sequential CTW code length of
this block,
\begin{equation}
\widehat{\bar{H}}_n(P_{\theta_0})
:= -\tfrac{1}{n}\log
Q^{\mathrm{CTW}}_{L}(Z_{1:n}).
\label{eq:entropy-est}
\end{equation}

This settles the second term of~\eqref{eq:kl-split}. The
cross-entropy term $\bar H(P_{\theta_0},P_{\theta_0+\Delta})$, by
contrast, requires evaluating a CTW model built from one source on data
generated by another, which the sequential code does not provide; this
is the purpose of the frozen scoring protocol introduced next.

\subsubsection{Frozen CTW scoring for KL estimation}

Following Cai et
al.~\cite{cai2006universal}, a CTW
scorer is built from a block generated by one source, then frozen and
used to score a block generated by the other source. This frozen
scoring step is what turns CTW coding into a divergence estimator.
Building a depth-$L$ frozen CTW scorer from a string
$x_{1:n}$ means running the standard CTW recursion on $x_{1:n}$ and
then keeping the resulting node counts, KT predictors, and CTW weighted
probabilities fixed when scoring another string; see
Appendix~\ref{app:ctw}.

\begin{definition}[Frozen CTW scorer]
\label{def:frozen-ctw-scorer}

For a string $x_{1:n}\in\mathcal A^n$ used to build the scorer
and a working depth $L$, let
\[
\widehat P^{\mathrm{CTW}}_{x_{1:n},L}(z_{1:r}),
\qquad r\ge 1,
\]
denote the probability assigned to a scoring string
$z_{1:r}\in\mathcal A^r$ by the frozen CTW scorer built from
$x_{1:n}$ and then kept fixed while scoring $z_{1:r}$.
The node counts are accumulated over the positions
$t=L+1,\dots,n$, so that every count is taken at full depth.

\end{definition}

\begin{definition}[Frozen CTW one-step predictor]
\label{def:frozen-ctw-onestep}

For a frozen scorer built from $x_{1:n}$ at working depth $L$ and for a
full working-depth context $s\in\mathcal A^L$, write
\[
\widehat T^{\mathrm{CTW}}_{x_{1:n},L}(a\mid s),
\qquad a\in\mathcal A,
\]
for the one-step probability assigned to symbol $a$ when the current
scoring context is $s$. For $n\ge L$, the associated block score admits the
factorization
\begin{equation}
\widehat P^{\mathrm{CTW}}_{x_{1:n},L}(z_{1:n})
=
C_{x,L}(z_{1:L})
\prod_{t=L+1}^{n}
\widehat T^{\mathrm{CTW}}_{x_{1:n},L}
\bigl(z_t\mid z_{t-L:t-1}\bigr),
\label{eq:frozen-ctw-factorization}
\end{equation}
where $C_{x,L}(z_{1:L})$ denotes the warm-up probability assigned to the
first $L$ scoring symbols,
\[
C_{x,L}(z_{1:L})
:=
\widehat P^{\mathrm{CTW}}_{x_{1:n},L}(z_{1:L}).
\]

\end{definition}

The frozen scorer differs from the sequential code: the frozen one-step
predictor $\widehat T^{\mathrm{CTW}}_{x_{1:n},L}(a\mid s)$ uses the
counts of the \emph{completed} construction pass over $x_{1:n}$, whereas the
sequential conditional $Q^{\mathrm{CTW}}_{L}(z_t\mid z_{1:t-1})$ uses
only the counts accumulated up to time $t-1$.
$Q^{\mathrm{CTW}}_{L}$ is a pmf on
$\mathcal A^n$, and so is $\widehat P^{\mathrm{CTW}}_{x_{1:n},L}$ for each
fixed construction string $x_{1:n}$, since in either case the one-step
probabilities sum to one over $\mathcal A$. They differ in what they depend on:
$Q^{\mathrm{CTW}}_{L}$ is a universal code that learns from the very
block it encodes, whereas
$\widehat P^{\mathrm{CTW}}_{x_{1:n},L}$ is determined by the
construction path alone and is therefore independent of the block it
scores whenever that path is drawn independently. The two objects
therefore receive separate treatments in the analysis of
Section~\ref{sec:theory}: the first is controlled by its coding
redundancy and the second by conditioning on the construction path.

The cross-entropy rate is estimated by drawing an independent auxiliary
path $X^{(\Delta)}_{1:n}\sim P_{\theta_0+\Delta}^{(n)}$,
building a frozen CTW scorer from
$X^{(\Delta)}_{1:n}$, and then
scoring the reference path $Z_{1:n}$ of~\eqref{eq:entropy-est}:
\begin{equation}
\widehat{\bar{H}}_n\bigl(P_{\theta_0},P_{\theta_0+\Delta}\bigr)
:= -\tfrac{1}{n}\log
\widehat P^{\mathrm{CTW}}_{X^{(\Delta)}_{1:n},L}(Z_{1:n}).
\label{eq:cross-entropy-est}
\end{equation}
Taking the difference of the two estimates yields the plug-in
per-symbol KL estimator
\begin{equation}
\widehat{\bar{D}}_n\bigl(P_{\theta_0}\,\|\,P_{\theta_0+\Delta}\bigr)
:= \widehat{\bar{H}}_n\bigl(P_{\theta_0},P_{\theta_0+\Delta}\bigr)
- \widehat{\bar{H}}_n\bigl(P_{\theta_0}\bigr).
\label{eq:kl-est}
\end{equation}
Equations~\eqref{eq:entropy-est}-\eqref{eq:kl-est} are
a variant of the universal divergence estimator of
\cite{cai2006universal}, adapted here to local perturbations of
$\theta_0$, with the entropy rate estimated by the
sequential CTW code rather than by a frozen model built from the
reference path itself. Its strong consistency is established in
Lemma~\ref{lem:ctw-kl}, and its finite-sample accuracy
in Proposition~\ref{prop:kl-rate}. The complete CTW-based KL estimator is stated in Algorithm~\ref{alg:ctw-kl}.

\begin{algorithm}[t]
\caption{CTW-based plug-in estimator of
$\bar{D}(P_{\theta_0}\,\|\,P_{\theta_0+\Delta})$}
\label{alg:ctw-kl}
\begin{algorithmic}[1]
\Require Target parameter $\theta_0$; perturbation $\Delta$;
length $n$;
CTW working depth $L$.
\Ensure $\widehat{\bar{D}}_n(P_{\theta_0}\,\|\,P_{\theta_0+\Delta})$.
\State Draw a reference path
$z_{1:n}\sim P_{\theta_0}^{(n)}$.
\State Draw an independent auxiliary path
$x_{1:n}\sim P_{\theta_0+\Delta}^{(n)}$.
\State Run the sequential CTW code on $z_{1:n}$ to obtain
$Q^{\mathrm{CTW}}_{L}(z_{1:n})$.
\State Build a depth-$L$ frozen CTW scorer
$\widehat P^{\mathrm{CTW}}_{x_{1:n},L}$ from $x_{1:n}$.
\State $\widehat{H}_n(P_{\theta_0})
\leftarrow -n^{-1}\log
Q^{\mathrm{CTW}}_{L}(z_{1:n})$.
\State $\widehat{H}_n(P_{\theta_0},P_{\theta_0+\Delta})
\leftarrow -n^{-1}\log
\widehat P^{\mathrm{CTW}}_{x_{1:n},L}(z_{1:n})$.
\State \textbf{return}
$\widehat{\bar{D}}_n(P_{\theta_0}\,\|\,P_{\theta_0+\Delta})
\leftarrow \widehat{H}_n(P_{\theta_0},P_{\theta_0+\Delta})
- \widehat{H}_n(P_{\theta_0})$.
\end{algorithmic}
\end{algorithm}

\subsection{Least-Squares Recovery of the Fisher Information}
\label{subsec:ls-recovery}

The final step inverts the quadratic relation~\eqref{eq:kl-fim-rate}
to recover $I(\theta_0)$ from a set of directional KL estimates. The
recovery follows the least-squares scheme of Berisha and
Hero~\cite{berisha2014empirical}, developed there for i.i.d.\ data with
a graph-based divergence estimator; here it is driven instead by the
CTW-based KL estimator of Section~\ref{subsec:kl-est}, which extends
the construction to non-i.i.d.\ finite-alphabet processes, and is
accompanied by the consistency and
finite-sample guarantees of Section~\ref{sec:theory}.

To recover the full matrix $I(\theta_0)$, the local KL
curvature must be probed in multiple directions, since each KL estimate
provides only one scalar quadratic-form measurement. Fix a step size
$\varepsilon>0$ and select unit vectors
$\{u_j\}_{j=1}^{m}\subset\mathbb{R}^d$. We therefore
instantiate the generic perturbation in Section~\ref{subsec:kl-est} by
the directional perturbations
\begin{equation}
\Delta_j:=\varepsilon u_j,\qquad \|u_j\|_2=1.
\label{eq:directional-perturbation}
\end{equation}
For each $j$, compute
$\widehat{\bar{D}}_j
:=\widehat{\bar{D}}_n(P_{\theta_0}\,\|\,P_{\theta_0+\Delta_j})
=\widehat{\bar{D}}_n(P_{\theta_0}\,\|\,P_{\theta_0+\varepsilon u_j})$ via
Algorithm~\ref{alg:ctw-kl}. Rescaling to remove the quadratic factor
in~\eqref{eq:kl-fim-rate} gives
\begin{equation}
v_j := \tfrac{2}{\varepsilon^2}\,\widehat{\bar{D}}_j
\;\approx\; u_j^\top I(\theta_0)\,u_j,
\qquad j=1,\dots,m,
\label{eq:rescale}
\end{equation}
where the approximation incurs an $O(\varepsilon)$ error from the third-order remainder in~\eqref{eq:kl-fim-rate}; this error is quantified in Lemma~\ref{lem:kl-fim}.

Each measurement~\eqref{eq:rescale} is a quadratic form in
$I(\theta_0)$, and the collection of them determines the FIM. Since
$I(\theta_0)$ is symmetric, expanding the quadratic form gives
\begin{equation*}
u_j^\top I(\theta_0)\,u_j
= \sum_{k=1}^{d} I_{kk}\,u_{jk}^2
+ 2\sum_{1\le k<\ell\le d} I_{k\ell}\,u_{jk}u_{j\ell}.
\end{equation*}
We collect the distinct entries of $I(\theta_0)$ into the vector
\begin{equation}
f
:=\operatorname{vec}_\triangle\!\bigl(I(\theta_0)\bigr)
=\bigl(I_{11},\dots,I_{dd},\,I_{12},I_{13},\dots,I_{(d-1)d}\bigr)^\top
\in\mathbb{R}^{q},
\qquad q := \tfrac{d(d+1)}{2},
\label{eq:f-vector}
\end{equation}
where $\operatorname{vec}_\triangle$ denotes the map that stacks the
diagonal entries followed by the strictly upper-triangular entries of a
symmetric matrix. We write $\operatorname{unvec}_\triangle$ for the
inverse map, which places the off-diagonal coordinates symmetrically in
the upper- and lower-triangular positions.
For each direction, define the feature row
\begin{equation}
U(u_j)^\top := \bigl(u_{j1}^2,\dots,u_{jd}^2,\,
2u_{j1}u_{j2},\dots,2u_{j(d-1)}u_{jd}\bigr),
\label{eq:feature-row}
\end{equation}
so that $u_j^\top I(\theta_0)u_j=U(u_j)^\top f$. The $m$ relations
\eqref{eq:rescale} can then be stacked into the linear regression model
\begin{equation}
v \;\approx\; U f,
\qquad
v=(v_1,\dots,v_m)^\top\in\mathbb{R}^m,\quad U\in\mathbb{R}^{m\times q}.
\label{eq:linear-system}
\end{equation}

Recovering $I(\theta_0)$ thus reduces to solving the linear
system~\eqref{eq:linear-system} in the $q$ unknowns $f$; the system is
overdetermined when $m>q$. When
$\operatorname{rank}(U)=q$, the ordinary least-squares (OLS) solution is
\begin{equation}
\widehat{f}_{\mathrm{ols}}
= \arg\min_{f}\tfrac{1}{2}\|Uf - v\|_2^2
= (U^\top U)^{-1}U^\top v.
\label{eq:ols}
\end{equation}

The matrix $U$ depends only on the chosen directions. It has full
column rank precisely when the rank-one symmetric matrices
$\{u_ju_j^\top\}_{j=1}^m$ span the space of symmetric
$d\times d$ matrices; well-separated directions generally improve its
conditioning. When the directional estimates have heterogeneous variances or are correlated, or when $U^\top U$ is poorly conditioned, generalized or regularized variants of~\eqref{eq:ols} can improve finite-sample efficiency; these are discussed in Section~\ref{subsec:practical}.

It remains to reshape $\widehat{f}$ into a matrix. Although
Assumption~\ref{ass:reg} ensures that the true rate $I(\theta_0)$ is
positive semidefinite (PSD), the unconstrained solution~\eqref{eq:ols}
may fail to be PSD. We therefore project the reshaped estimate onto
the PSD cone:
\begin{equation}
\widehat{I}(\theta_0)
:= \Pi_{\succeq 0}\bigl(\operatorname{unvec}_\triangle(\widehat{f})\bigr)
= \arg\min_{X\succeq 0}\bigl\|X
- \operatorname{unvec}_\triangle(\widehat{f})\bigr\|_F^2,
\label{eq:psd-proj}
\end{equation}
which, if $\operatorname{unvec}_\triangle(\widehat{f})$ has eigenvalue decomposition $V\Lambda V^\top$, can be computed in closed form by eigenvalue clipping as
\begin{equation}
\Pi_{\succeq 0}\bigl(V\Lambda V^\top\bigr)
= V\max(\Lambda,0)\,V^\top,
\label{eq:eig-clip}
\end{equation}
at a cost of $O(d^3)$ operations~\cite{higham1988computing}. Since $I(\theta_0)$
itself lies in the PSD cone, this projection is
non-expansive in the Frobenius norm
(Lemma~\ref{lem:psd-proj}): it never increases the estimation error,
and reduces to the identity whenever the reshaped estimate is already
PSD. As an alternative, the PSD constraint can be imposed
during the fit through a convex semidefinite program, and
coordinate-direction anchoring of the diagonal
entries~\cite{berisha2014empirical} can further stabilize
off-diagonal recovery; these options are deferred to
Section~\ref{subsec:practical}.

The complete procedure (directional probing, the CTW--KL subroutine, least-squares recovery, and PSD projection) is summarized in Algorithm~\ref{alg:fim-recovery}. Its consistency is established in Theorem~\ref{thm:consistency}, and its convergence rate in Theorem~\ref{thm:fim-rate}.

\begin{algorithm}[t]
\caption{Reconstruction of $I(\theta_0)$ from CTW-based
directional KL estimates}
\label{alg:fim-recovery}
\begin{algorithmic}[1]
\Require Target parameter $\theta_0\in\mathbb{R}^d$; step size
$\varepsilon>0$; unit directions $\{u_j\}_{j=1}^{m}$; length $n$;
CTW working depth $L$.
\Ensure Estimated Fisher information rate $\widehat{I}(\theta_0)$.
\State Draw a reference path $z_{1:n}\sim P_{\theta_0}^{(n)}$.
\State Run the sequential CTW code on $z_{1:n}$ to obtain
$Q^{\mathrm{CTW}}_{L}(z_{1:n})$.
\State $\widehat{\bar H}_n(P_{\theta_0})\leftarrow
-n^{-1}\log
Q^{\mathrm{CTW}}_{L}\bigl(z_{1:n}\bigr)$.
\Comment{reference entropy, shared across all directions}
\For{$j=1$ \textbf{to} $m$}
  \State $\Delta_j\leftarrow \varepsilon u_j$.
  \State Draw an independent auxiliary path
  $x^{(j)}_{1:n}\sim P_{\theta_0+\Delta_j}^{(n)}
  =P_{\theta_0+\varepsilon u_j}^{(n)}$.
  \State Build a depth-$L$ frozen CTW scorer
  $\widehat P^{\mathrm{CTW}}_{x^{(j)}_{1:n},L}$ from $x^{(j)}_{1:n}$.
  \State $\widehat{\bar H}_n(P_{\theta_0},P_{\theta_0+\Delta_j})\leftarrow
  -n^{-1}\log
  \widehat P^{\mathrm{CTW}}_{x^{(j)}_{1:n},L}\bigl(z_{1:n}\bigr)$.
  \State $\widehat{\bar D}_j \leftarrow
  \widehat{\bar H}_n(P_{\theta_0},P_{\theta_0+\Delta_j)}
  - \widehat{\bar H}_n(P_{\theta_0})$.
  \State $v_j \leftarrow \tfrac{2}{\varepsilon^2}\,\widehat{\bar D}_j$.
  \State Form the feature row $U(u_j)^\top$ as in~\eqref{eq:feature-row}.
\EndFor
\State Stack $v\leftarrow(v_1,\dots,v_m)^\top$ and
$U\leftarrow[U(u_1)^\top;\dots;U(u_m)^\top]$.
\State Solve for $\widehat{f}$ by least squares,
\eqref{eq:ols}.
\State Reshape $\widehat{I}^{\mathrm{LS}}
\leftarrow\operatorname{unvec}_\triangle(\widehat{f})$.
\State Enforce PSD: $\widehat{I}(\theta_0)\leftarrow
\Pi_{\succeq 0}(\widehat{I}^{\mathrm{LS}})$ as in~\eqref{eq:eig-clip}.
\State \textbf{return} $\widehat{I}(\theta_0)$.
\end{algorithmic}
\end{algorithm}

The arrangement of Algorithm~\ref{alg:fim-recovery} reflects a
structural property of the directional estimates. The reference path
$z_{1:n}\sim P_{\theta_0}^{(n)}$ and its sequential CTW code
$Q^{\mathrm{CTW}}_{L}(z_{1:n})$, and hence the entropy estimate
$\widehat{\bar H}_n(P_{\theta_0})$, do not
depend on the probing direction: every estimate $\widehat{\bar D}_j$
subtracts the same $\widehat{\bar H}_n(P_{\theta_0})$ and encodes the same
reference path $z_{1:n}$ under a direction-specific frozen CTW scorer
$\widehat P^{\mathrm{CTW}}_{x^{(j)}_{1:n},L}$. The reference quantities are therefore
computed once, outside the loop, and reused across all $m$ directions.
The total cost is dominated by the construction of $1+m$ CTW
scorers, each linear in $n$, rather than the $2m$ scorers of a
naive per-direction implementation.

The measurements~\eqref{eq:rescale} probe the local KL curvature from
one side only: each direction $u_j$ is explored at
$\theta_0+\varepsilon u_j$ alone, so the systematic
error of $v_j$ inherits the odd part of the local
expansion~\eqref{eq:kl-fim-rate}. A symmetric variant
removes that odd part at the cost of one additional
auxiliary path per direction: run Algorithm~\ref{alg:ctw-kl} at both
perturbations $\pm\varepsilon u_j$, reusing the same reference path
$z_{1:n}$ and the same entropy estimate, and form the
averaged measurement
\begin{equation}
v^{(2)}_j:=\frac{1}{\varepsilon^2}
\Bigl[\widehat{\bar D}_n(P_{\theta_0}\,\|\,P_{\theta_0+\varepsilon u_j})
+\widehat{\bar D}_n(P_{\theta_0}\,\|\,P_{\theta_0-\varepsilon u_j})\Bigr].
\label{eq:central-meas}
\end{equation}
The remainder of Algorithm~\ref{alg:fim-recovery} being unchanged;
write $\widehat I^{(2)}_{n,\varepsilon}(\theta_0)$ for the resulting
estimator. Averaging over $\pm\varepsilon u_j$ cancels the odd-order
terms of the expansion, so the systematic error of $v^{(2)}_j$ is one
order smaller in $\varepsilon$ than that of $v_j$, while consistency is
unaffected. The cost rises from $1+m$ to $1+2m$ CTW scorers.
The two schemes therefore trade simulation cost against
accuracy, and both are retained in the sequel.
Section~\ref{subsec:rate} gives a convergence rate for each and
quantifies the resulting acceleration.

\begin{remark}[Higher-order differencing]
\label{rem:higher-order}
In principle, bias terms of successively higher order can be cancelled
by combining KL estimates taken at several step sizes, as in Richardson
extrapolation~\cite{stoer2002}. Each additional order, however, requires simulator runs
at further step sizes and stronger smoothness of the
transition probabilities in $\theta$. The weights of the extrapolated combination
also grow with the order and amplify the noise of the individual KL
estimates. Meanwhile, each further order improves the convergence rate
by a diminishing margin. Differencing reduces only the bias, and the
statistical fluctuation of the KL estimates eventually dominates the
error. Since the central
difference is the only refinement obtained by symmetry alone, we
restrict attention to the one-sided and central schemes.
\end{remark}

\begin{remark}[Choice of the direction set]
\label{rem:directions}
The constant $\kappa_U:=1/\sigma_{\min}(U)$ governs how
measurement errors propagate through the least-squares step
(Lemma~\ref{lem:ls-continuity}). It depends only on the probing
directions $\{u_j\}_{j=1}^m$, and not on $n$ or $\varepsilon$, so the
conditioning of the recovery is entirely under the user's control. A
minimal full-rank direction design ($m=q$) is
\[
\{e_k\}_{k=1}^d
\cup
\{(e_k+e_\ell)/\sqrt{2}\colon 1\le k<\ell\le d\},
\]
and well-separated directions keep $\sigma_{\min}(U)$
away from zero, hence $\kappa_U$ is small.
\end{remark}

\subsection{Practical Refinements (move to appendix?)}
\label{subsec:practical}

The estimator analyzed in Section~\ref{sec:theory} uses ordinary least squares followed by post-fit projection, which is the simplest instance of the recovery step. Several refinements can improve finite-sample performance under specific conditions. We summarize them here as practical guidance, while the consistency proof below is stated for the canonical OLS-projection version.

\begin{enumerate}
\item \emph{Generalized least squares.}
The directional estimates $v_j$ may have heterogeneous variances, for instance when some directions $u_j$ probe flatter regions of the log-likelihood and are therefore estimated less reliably; they are also correlated by construction, because all directions share the reference path and subtract the same entropy estimate. In this case, generalized least squares improves efficiency. With $W=\widehat\Sigma^{-1}$, where $\widehat\Sigma$ is the sample covariance of the measurement vector $v$ over $R$ independent replications of the full procedure, the solution
\begin{equation}
\widehat{f}_{\mathrm{gls}}
= (U^\top W U)^{-1}U^\top W\,v,
\qquad W = \widehat\Sigma^{-1},
\label{eq:wls}
\end{equation}
can downweight the noisier measurements and exploit their correlation; when $R$ is small relative to $m$, a diagonal $\widehat\Sigma$ recovers weighted least squares with $w_j\propto1/\widehat{\operatorname{Var}}(v_j)$.

\item \emph{Tikhonov regularization.}
Because $U$ is built from the user-chosen directions, ill-conditioning of $U^\top U$ can usually be avoided by selecting well-separated directions. If this ill-conditioning nonetheless arises, as with nearly collinear $u_j$, a Tikhonov term can stabilize the solution~\cite{tikhonov1977}:
\begin{equation}
\widehat{f}_{\lambda}
= (U^\top W U + \lambda I)^{-1}U^\top W\,v,
\qquad \lambda > 0.
\label{eq:ridge}
\end{equation}
A small default regularization term such as $\lambda = 10^{-3}\,\operatorname{tr}(U^\top U)/q$ provides robustness with negligible bias.

\item \emph{PSD-constrained fitting.}
The post-fit projection~\eqref{eq:psd-proj} restores positive semidefiniteness after the unconstrained fit. The constraint can instead be imposed \emph{during} the fit by solving the convex program
\begin{equation}
\min_{f}\ \frac{1}{2}\bigl\|U f - v\bigr\|_2^2
\quad\text{s.t.}\quad
\operatorname{unvec}_\triangle(f)\succeq 0,
\label{eq:sdp}
\end{equation}
a quadratic objective under a semidefinite constraint~\cite{boyd2004,vandenberghe1996semidefinite}. This formulation avoids the eigenvalue clipping of~\eqref{eq:eig-clip} and can reduce small-sample variance when the unconstrained estimate lies far from the PSD cone, at the cost of solving a semidefinite program. The generalized or regularized objective may replace the least-squares term without affecting convexity. For moderate $d$, up to a few tens, the program can be solved in seconds by off-the-shelf solvers, whereas for larger $d$ the post-fit projection is computationally lighter.

\item \emph{Diagonal anchoring.}
Berisha and Hero~\cite{berisha2014empirical} exploit coordinate directions to anchor the diagonal of the FIM. Including the coordinate directions $\{e_k\}_{k=1}^{d}$ among the $\{u_j\}$ yields direct estimates $\widehat{I}_{kk} := v_j$ (for the indices $j$ with $u_j=e_k$) of the diagonal entries, which are then held fixed while the remaining entries are recovered by a PSD-constrained fit,
\begin{equation}
\min_{f}\ \frac{1}{2}\bigl\|U f - v\bigr\|_2^2
\quad\text{s.t.}\quad
\operatorname{unvec}_\triangle(f)\succeq 0,\quad
f^{(k)} = \widehat{I}_{kk},\ k=1,\dots,d.
\label{eq:anchor-hard}
\end{equation}
The equality constraints inject direct information about the diagonal, improving conditioning and stabilizing off-diagonal recovery. When the diagonal estimates are themselves noisy, a soft-constrained variant preserves feasibility, i.e.,
\begin{equation}
\min_{\operatorname{unvec}_\triangle(f)\succeq 0}\
\frac{1}{2}\bigl\|U f - v\bigr\|_2^2
+ \frac{\mu}{2}\sum_{k=1}^{d}\bigl(f^{(k)} - \widehat{I}_{kk}\bigr)^2,
\qquad \mu > 0,
\label{eq:anchor-soft}
\end{equation}
which remains convex and often performs better when the $\widehat{I}_{kk}$ are uncertain.
\end{enumerate}
These refinements operate on the same directional KL measurements as the canonical estimator; a full asymptotic treatment of the generalized least-squares and regularized variants is deferred to future work.

\section{Consistency and Finite-Sample Analysis}
\label{sec:theory}

This section establishes the consistency of the proposed estimator and its finite-sample convergence rates. The analysis is organized around the pipeline of Algorithms~\ref{alg:ctw-kl} and~\ref{alg:fim-recovery}: Section~\ref{subsec:consistency} proves that the directional KL estimates converge almost surely and propagates the limits through the recovery (Theorem~\ref{thm:consistency}); Section~\ref{subsec:rate} quantifies the accuracy jointly in the trajectory length and the step size, for both differencing schemes (Theorem~\ref{thm:fim-rate}). As anticipated in Section~\ref{subsec:practical}, we analyze the canonical estimator that uses ordinary least squares followed by post-fit eigenvalue projection.

Throughout, $\theta_0$ denotes the fixed target parameter from
Assumption~\ref{ass:reg}. Fix a set of unit directions
$\{u_j\}_{j=1}^{m}\subset\mathbb{R}^d$, and let
$U\in\mathbb{R}^{m\times q}$, $q=d(d+1)/2$, be the feature matrix whose
rows are the vectors $U(u_j)^\top$ of~\eqref{eq:feature-row}. We write
$\widehat{\bar{D}}_n(P_{\theta_0}\,\|\,P_{\theta_0+\Delta})$
for the CTW-based plug-in KL estimator at a generic perturbation
$\Delta$ and
$\widehat{I}_{n,\varepsilon}(\theta_0)$ for the OLS-projection
estimator, and $\widehat I^{(2)}_{n,\varepsilon}(\theta_0)$
for its central-difference variant of
Section~\ref{subsec:ls-recovery}. For each fixed finite or countable collection of parameter
values used in the analysis, we place the corresponding independent
infinite simulator trajectories on their product probability space
and construct the length-$n$ estimators from their first $n$ symbols.
Thus, all estimators indexed by $n$ are defined on a common
probability space, as required for the almost-sure statements below.

\subsection{Consistency}
\label{subsec:consistency}

We first establish that, for a fixed step size $\varepsilon$, the
estimator converges almost surely as $n\to\infty$, and then show that
these limits approach the true Fisher information rate
as $\varepsilon\downarrow0$. The argument rests on the almost-sure
consistency of the directional KL estimates (Lemma~\ref{lem:ctw-kl}),
combined with the deterministic lemmas of Appendix~\ref{app:aux}.

\subsubsection{CTW-based KL estimator}
The CTW-based KL estimate is strongly consistent for each fixed
perturbation $\Delta$.

\begin{lemma}[Consistency of the CTW-based KL estimator]
\label{lem:ctw-kl}
Let $\Delta\in\mathbb R^d$ be such that
$\theta_0+\Delta\in\mathcal N(\theta_0)$ and
$\bar{D}(P_{\theta_0}\,\|\,P_{\theta_0+\Delta})<\infty$. Let the CTW
working depth satisfy $L\ge L^\star$. Then
\begin{equation}
\widehat{\bar{D}}_n\bigl(P_{\theta_0}\,\|\,P_{\theta_0+\Delta}\bigr)
\;\xrightarrow[n\to\infty]{\textnormal{a.s.}}\;
\bar{D}\bigl(P_{\theta_0}\,\|\,P_{\theta_0+\Delta}\bigr).
\label{eq:ctw-kl-as}
\end{equation}
\end{lemma}

The proof verifies the conditions of
\cite[Thm.~2]{cai2006universal} for the cross-entropy
estimate and establishes the entropy limit from the
universality of the sequential CTW code; it is given in
Appendix~\ref{app:cons-proof}.

\begin{remark}
\label{rem:assumption-match}
The estimator of \cite{cai2006universal} requires only
$\Pr_\vartheta(Y_t=a\mid Y_{t-L^\star:t-1}=s)<1$ for every relevant
parameter value $\vartheta$, context $s$, and symbol $a$, whereas
Assumption~\ref{ass:reg}(i) imposes the two-sided bound
$\Pr_\vartheta(Y_t=a\mid Y_{t-L^\star:t-1}=s)
\in[\delta,1-\delta]$. The stronger lower bound is not needed for
Lemma~\ref{lem:ctw-kl}; it keeps the transition
probabilities uniformly bounded away from zero over
$\mathcal N(\theta_0)$, which the finite-sample analysis of
Section~\ref{subsec:rate} uses throughout. The Taylor expansion in
Lemma~\ref{lem:kl-fim} is controlled primarily by the smoothness and
bounded-derivative requirements in Assumption~\ref{ass:reg}(ii).
\end{remark}

\subsubsection{Main consistency result}
Lemma~\ref{lem:ctw-kl}, combined with the KL-FIM
expansion (Lemma~\ref{lem:kl-fim}) and the deterministic recovery
lemmas of Appendix~\ref{app:aux} (Lemmas~\ref{lem:ls-continuity}
and~\ref{lem:psd-proj}), yields strong consistency in an iterated
double-limit sense.

\begin{theorem}[Iterated strong consistency]
\label{thm:consistency}
Suppose Assumptions~\ref{ass:sim}--\ref{ass:reg} hold, let the
CTW working depth satisfy $L\ge L^\star$, and let
$\{u_j\}_{j=1}^m\subset\mathbb R^d$ be unit directions such that $U$
has full column rank $q$. Let $\varepsilon_0$ be as in
Lemma~\ref{lem:kl-fim}. Then:

\textnormal{(i)} for every fixed $\varepsilon\in(0,\varepsilon_0]$,
\begin{equation}
\widehat{I}_{n,\varepsilon}(\theta_0)
\;\xrightarrow[n\to\infty]{\textnormal{a.s.}}\;
\Pi_{\succeq0}\bigl(I(\theta_0)+E_U(\varepsilon)\bigr),
\label{eq:thm-step1}
\end{equation}
where $E_U(\varepsilon)\in\mathbb{R}^{d\times d}_{\mathrm{sym}}$ is
deterministic and
$\|E_U(\varepsilon)\|_F\le C\,\varepsilon$ for a
constant $C<\infty$ independent of $\varepsilon$;

\textnormal{(ii)} consequently, with the inner limit
understood in the almost-sure sense of \textnormal{(i)},
\begin{equation}
\lim_{\varepsilon\downarrow0}\;\lim_{n\to\infty}\;
\widehat{I}_{n,\varepsilon}(\theta_0) = I(\theta_0)
.
\label{eq:thm-step2}
\end{equation}
\end{theorem}

The proof is given in
Appendix~\ref{app:cons-proof}. For fixed $\varepsilon$, the $m$
directional estimates converge almost surely (Lemma~\ref{lem:ctw-kl});
the limits pass through the least-squares map and the PSD projection
by the Lipschitz bounds of Lemmas~\ref{lem:ls-continuity}
and~\ref{lem:psd-proj}, and the residual matrix $E_U(\varepsilon)$
collects the Taylor bias of Lemma~\ref{lem:kl-fim}.
The outer limit is deterministic, since
$E_U(\varepsilon)$ is also deterministic. The same argument covers the
central-difference variant $\widehat I^{(2)}_{n,\varepsilon}(\theta_0)$:
each of its two directional KL estimates converges almost surely
(Lemma~\ref{lem:ctw-kl} applied at $\theta_0\pm\varepsilon u_j$), and
the symmetrized Taylor bias is again $O(\varepsilon)$, so
Theorem~\ref{thm:consistency} holds unchanged.

\begin{remark}[The iterated limit and the $\varepsilon$--$n$ coupling]
\label{rem:eps-n-coupling}
Theorem~\ref{thm:consistency} is an iterated limit: $n\to\infty$ at
each fixed $\varepsilon$ (using the almost-sure
consistency of Lemma~\ref{lem:ctw-kl}), followed by
$\varepsilon\downarrow0$ to remove the Taylor
bias. The order matters: the inner limit is
almost sure at each fixed $\varepsilon$, whereas the outer one is
deterministic, since the residual matrix $E_U(\varepsilon)$ of part
\textnormal{(i)} is deterministic.
\end{remark}

\subsection{Finite-Sample Rate}
\label{subsec:rate}

Theorem~\ref{thm:consistency} concerns two limits taken separately; in practice the step size must shrink with the sample size, $\varepsilon=\varepsilon_n\downarrow0$, and the accuracy of the recovered matrix is governed jointly by $n$ and $\varepsilon$. The coupling is delicate because the recovery step~\eqref{eq:rescale} divides each KL estimate by $\varepsilon^2$: an estimation error that does not shrink with $\varepsilon$ is amplified without control as $\varepsilon_n\downarrow0$, so a useful rate requires the error itself to \emph{scale with} $\varepsilon$. This subsection establishes that scaling (Proposition~\ref{prop:kl-rate}) and propagates it through the recovery, both for the one-sided measurements of Algorithm~\ref{alg:fim-recovery} and for the central-difference refinement of Section~\ref{subsec:ls-recovery}. The resulting mean-error rates, $O(n^{-1/4})$ for the former and $O(n^{-1/3})$ for the latter, are collected in Theorem~\ref{thm:fim-rate}.
Throughout this subsection and the corresponding
appendices, the working depth satisfies $L\ge1$ and the trajectory
length satisfies $n\ge2L$.

\subsubsection{Structural conditions on the source}
The rate is stated under three structural conditions on the source,
beyond the standing Assumptions~\ref{ass:source}--\ref{ass:reg}: a
stable tree topology, non-degenerate context occupancy, and a strictly
positive splitting margin. They are used only in the
finite-sample analysis; the consistency results of
Section~\ref{subsec:consistency} do not require them.
They are properties of the
model family alone, not of the estimator, and each is checkable by
inspecting the transition law.

The conditions are phrased in terms of the context
tree of Section~\ref{subsec:kl-est}, whose working depth $L$ guarantees
that it contains every context the source actually uses; typically the
source needs only part of it. Call a context $s$ \emph{sufficient} if
the conditional law of the next symbol, given any past ending in $s$,
depends on that past only through $s$: once $s$ is read, reading
further back changes nothing. By Assumption~\ref{ass:source} every
context of length $L^\star$ is sufficient; shorter contexts may or may
not be. A sufficient context is \emph{minimal} if removing its oldest
symbol leaves a context that is no longer sufficient. The
\emph{minimal context tree} $T^\star(\theta)$ is the subtree of the
working tree whose \emph{leaves} are the minimal sufficient contexts
and whose \emph{internal nodes} $\mathcal I(T^\star)$ are the
insufficient contexts: at a leaf the conditional law is fully
determined, while at an internal node some longer
suffix of the past still changes it. It is the exact memory
structure of $P_\theta$, of depth at most $L^\star$. Write
$T^\star:=T^\star(\theta_0)$.

\begin{assumption}[Fixed local topology]
\label{ass:topology}
The minimal context tree $T^\star$ is the same for every
$\theta\in\mathcal N(\theta_0)$.
\end{assumption}

\begin{assumption}[Uniform occupancy]
\label{ass:occupancy}
For a context $s$, let
$\pi_\theta(s):=\Pr_\theta(Y_{t-|s|:t-1}=s)$ be its stationary
\emph{occupancy}: the probability that the $|s|$ most recent symbols
equal $s$. The \emph{minimum occupancy} over full-depth contexts is
uniformly positive:
\[
\pi_{\min}
:=\inf_{\theta\in\mathcal N(\theta_0)}\min_{w\in\mathcal A^{L}}\pi_\theta(w)>0.
\]
\end{assumption}

\begin{assumption}[Uniform separation gap]
\label{ass:gap}
For an internal node $s\in\mathcal I(T^\star)$ of length $\ell=|s|$, let
$C_t:=Y_{t-\ell-1}$ be the past symbol that selects which depth-$(\ell+1)$
child of $s$ is active at time $t$, and let
$I_\theta(Y_t;C_t\mid S_t^{(\ell)}=s)$ be the mutual information between the
next symbol $Y_t$ and $C_t$ under the stationary law $P_\theta$,
conditioned on the event $\{S_t^{(\ell)}=s\}$, where
$S_t^{(\ell)}:=Y_{t-\ell:t-1}$ is the length-$\ell$ context at time
$t$. If $T^\star$ has no internal node, as it happens
for memoryless sources, the condition below is vacuous under the
convention $\inf\varnothing=+\infty$. The \emph{separation gap} is
uniformly positive:
\[
\kappa
:=\inf_{\theta\in\mathcal N(\theta_0)}\inf_{s\in\mathcal I(T^\star)}
\pi_\theta(s)\,I_\theta(Y_t;C_t\mid S_t^{(\ell)}=s)>0 .
\]
\end{assumption}

Each condition secures one ingredient of the analysis.
Assumption~\ref{ass:topology} fixes the qualitative memory structure:
\emph{which} parts of the past matter does not change over
$\mathcal N(\theta_0)$, so the perturbations probed by the estimator
move the transition probabilities but never the shape of the tree.
Assumption~\ref{ass:occupancy} guarantees data everywhere: every
length-$L$ context is visited at a linear rate, so empirical counts
accumulate in every state. Assumption~\ref{ass:gap} makes the
minimality of $T^\star$ quantitative and, at the same
time, strengthens it. At an internal node the deeper past does matter,
but this does not need to be visible one symbol at a time: the conditional
mutual information carried by $C_t$ alone can vanish while a longer
history remains informative. The gap excludes that case and requires,
in addition, that the information carried by $C_t$ and weighted by the
occupancy $\pi_\theta(s)$ to be definite, not merely nonzero, uniformly
over $\mathcal I(T^\star)$ and $\mathcal N(\theta_0)$.
Without this margin, \cite{cai2006universal} still give almost-sure
consistency, but no uniform finite-sample rate can be expected.

\begin{remark}[Interpretation and local enforceability
of the structural conditions]
\label{rem:mildness}
The conditions are local non-degeneracy requirements,
and two of them can be enforced by shrinking the neighborhood.
Assumption~\ref{ass:reg}(i) already forces $\pi_{\min}\ge\delta^{L}$,
so Assumption~\ref{ass:occupancy} holds automatically.
Assumption~\ref{ass:topology} is automatic when
$T^\star$ is the full tree $\mathcal A^{L^\star}$, since there are then
no equalities among conditional laws to be broken. When $T^\star$ is a
proper subtree, each of its leaves merges several deeper contexts whose
conditional laws coincide, and an arbitrary perturbation of the
parameters may separate them; fixed topology then asks the
parameterization itself to preserve these equalities, as it does when
the model is specified by a tree with free parameters at its leaves.
Given such a parameterization, shrinking $\mathcal N(\theta_0)$
suffices, provided that no split degenerates at $\theta_0$ itself. Given Assumption~\ref{ass:topology}, the content
of Assumption~\ref{ass:gap} is pointwise: if
$I_{\theta_0}(Y_t;C_t\mid s)>0$ for every internal node $s$, then,
since $\theta\mapsto\pi_\theta(s)\,I_\theta(Y_t;C_t\mid s)$ is
continuous by Assumption~\ref{ass:reg}(ii), a uniform $\kappa>0$ holds
by shrinking $\mathcal N(\theta_0)$ if necessary. Occupancy and gap
are nevertheless stated as separate assumptions because $\pi_{\min}$
and $\kappa$ are the quantities through which the constants of the rate
depend on the model.
\end{remark}

\subsubsection{Local KL rate}
Fix a unit direction $u$ and a step $\varepsilon>0$,
and run Algorithm~\ref{alg:ctw-kl} with the directional perturbation
$\Delta=\varepsilon u$.
The first result is the finite-sample analogue of Lemma~\ref{lem:ctw-kl};
it is the source of every subsequent rate, and its message is that the
KL estimation error \emph{scales with} $\varepsilon$.

\begin{proposition}[Local frozen-CTW KL rate]
\label{prop:kl-rate}
Suppose Assumptions~\ref{ass:source}--\ref{ass:reg} and
\ref{ass:topology}--\ref{ass:gap} hold with $L\ge L^\star$, and let
$\varepsilon_0$ be small enough that
$\theta_0+\varepsilon u\in\mathcal N(\theta_0)$ for all $\|u\|_2=1$,
$\varepsilon\in(0,\varepsilon_0]$. Then there is a constant
$C_{\mathrm{KL}}<\infty$ such that, uniformly over $\|u\|_2=1$ and
$0<\varepsilon\le\varepsilon_0$,
\begin{equation}
\mathbb E\bigl|\widehat{\bar D}_{n}(P_{\theta_0}\,\|\,P_{\theta_0+\varepsilon u})
-\bar D(P_{\theta_0}\,\|\,P_{\theta_0+\varepsilon u})\bigr|
\le C_{\mathrm{KL}}\!\left(\frac{\varepsilon}{\sqrt n}+\frac{\log n}{n}\right).
\label{eq:kl-rate}
\end{equation}
\end{proposition}

The content of the bound lies in the first term: the
statistical error carries the factor $\varepsilon$ inherited from the
closeness of the two laws, so it shrinks together with the perturbation
instead of remaining at the generic $1/\sqrt n$ level. The second term
$\log n/n$ is the learning overhead of the CTW machinery; it does not
scale with $\varepsilon$, but after the $\varepsilon^2$ division it is
of lower order at the step sizes chosen below. The proof is given in
Appendix~\ref{app:rate}, where the structural
Assumptions~\ref{ass:topology}--\ref{ass:gap} enter through a single
quantitative property of the frozen CTW predictor, the
prediction-moment lemma (Lemma~\ref{lem:pmc}).

\subsubsection{Main finite-sample theorem}
The main result propagates the KL rate of
Proposition~\ref{prop:kl-rate} through the recovery step, for both the
one-sided and the central-difference scheme. The bounds are mean
($L^1$) error rates; the proof is given in
Appendix~\ref{app:rate-fim}.

\begin{theorem}[Mean-error finite-sample FIM rate]
\label{thm:fim-rate}
Suppose Assumptions~\ref{ass:sim}--\ref{ass:reg} and the
structural
Assumptions~\ref{ass:topology}--\ref{ass:gap} hold, let the CTW depth
satisfy $L\ge L^\star$, and let $\{u_j\}_{j=1}^m$ be unit directions
such that $U$ has full column rank $q$. Then there is a constant
$C<\infty$ such that, for every
$\varepsilon\in(0,\varepsilon_0]$:
\begin{enumerate}
\item[(a)] the one-sided estimator
$\widehat I_{n,\varepsilon}(\theta_0)$ of
Algorithm~\ref{alg:fim-recovery} satisfies
\begin{equation}
\mathbb E\bigl\|\widehat I_{n,\varepsilon}(\theta_0)-I(\theta_0)\bigr\|_F
\le C\!\left(\frac{1}{\varepsilon\sqrt n}+\frac{\log n}{n\varepsilon^2}+\varepsilon\right),
\label{eq:fim-matrix-onesided}
\end{equation}
and the step size $\varepsilon_n\asymp n^{-1/4}$ gives
$\mathbb E\|\widehat I_{n,\varepsilon_n}(\theta_0)-I(\theta_0)\|_F
=O(n^{-1/4})$;
\item[(b)] if, in addition, the maps
$\theta\mapsto\Pr_\theta(Y_t=a\mid S_t=s)$ are four times
continuously differentiable on $\mathcal N(\theta_0)$ with derivatives
bounded uniformly in $(s,a)$, the central-difference estimator
$\widehat I^{(2)}_{n,\varepsilon}(\theta_0)$ of
Section~\ref{subsec:ls-recovery} satisfies
\begin{equation}
\mathbb E\bigl\|\widehat I^{(2)}_{n,\varepsilon}(\theta_0)-I(\theta_0)\bigr\|_F
\le C\!\left(\frac{1}{\varepsilon\sqrt n}+\frac{\log n}{n\varepsilon^2}+\varepsilon^2\right),
\label{eq:fim-matrix-rate}
\end{equation}
and the step size $\varepsilon_n\asymp n^{-1/6}$
gives
\begin{equation}
\mathbb E\bigl\|\widehat I^{(2)}_{n,\varepsilon_n}(\theta_0)-I(\theta_0)\bigr\|_F
=O\!\bigl(n^{-1/3}\bigr).
\label{eq:fim-rate-final}
\end{equation}
\end{enumerate}
\end{theorem}

The theorem prescribes the operating point of the method. With the step size set to $\varepsilon_n\asymp n^{-1/6}$, the central-difference scheme attains the matrix error $O(n^{-1/3})$ at the cost of $1+2m$ CTW scorers; the one-sided scheme, at its own optimum $\varepsilon_n\asymp n^{-1/4}$, attains $O(n^{-1/4})$ with roughly half the auxiliary simulations. The difference between the two rates stems from the Taylor bias of the measurements, which the symmetrization improves from $\varepsilon$ to $\varepsilon^2$. In particular, as the trajectory length grows, the step size should shrink at the prescribed rate.

\section{Numerical Illustration}
\label{sec:experiments}

This section reports five experiments. The first two verify that the
estimator recovers the Fisher information rate on sources for which it
is available in closed form. The first experiment uses a binary Markov chain with
a single logistic parameter, and checks the recovered value and the
convergence rate of Theorem~\ref{thm:fim-rate}. The second one uses a
multi-parameter Markov source, whose off-diagonal entries are recovered
from measurements of mixed curvature. The third and fourth experiments measure the cost of different working depths, on a source whose memory length is finite
and known to us and then on a hidden-state source that has no finite
memory length and lies outside Assumption~\ref{ass:source}. The fifth setup
uses the estimate to plan a sampling budget and compares the realized
precision with the promised one.

All five experiments run Algorithm~\ref{alg:fim-recovery} end to end. Unless stated
otherwise, the directional estimates use the central-difference
measurements~\eqref{eq:central-meas}, and every reported figure is an
average over $R$ independent replications of the whole pipeline, each
with freshly simulated trajectories.

\subsection{Binary Markov source: value and rate}
\label{subsec:exp1}

This experiment considers a binary Markov chain with a single logistic
parameter. With $p:=\sigma(\theta)$ and $\sigma(x)=1/(1+e^{-x})$, the
transition matrix is symmetric and equal to
\begin{equation}
T(\theta)=\begin{bmatrix} p & 1-p\\ 1-p & p\end{bmatrix},
\label{eq:exp1-model}
\end{equation}
so the stationary distribution of the chain is uniform and the Fisher information rate is
the scalar
\begin{equation}
I(\theta)=p(1-p)=\sigma(\theta)\bigl(1-\sigma(\theta)\bigr).
\label{eq:exp1-truth}
\end{equation}
The source has memory length $L^\star=1$, and the working depth is set
to $L=1$ throughout. The experiment first recovers the information as a
function of the parameter, and then measures the convergence rate
against the exponent stated in Theorem~\ref{thm:fim-rate}.

\subsubsection{Recovering the information landscape}
The estimator is run at 15 parameter values evenly spaced over
$\theta\in[-3,3]$, so that what is compared with the closed
form~\eqref{eq:exp1-truth} is a curve rather than a single number.
Three trajectory lengths are used, $n=10^{3}$, $3\times10^{3}$ and
$10^{4}$, each with $R=50$ repetitions and step size
$\varepsilon_n=6(\log n/n)^{1/3}$.

\begin{figure}[t]
    \centering
    \includegraphics[width=0.62\linewidth]{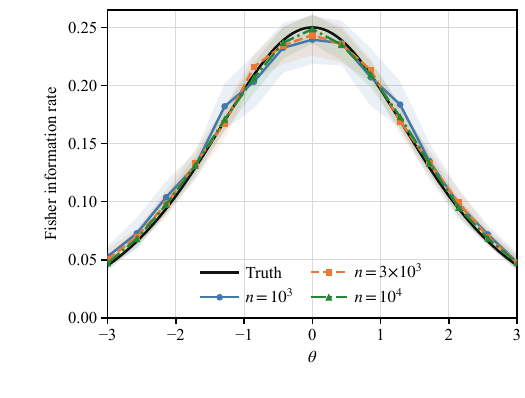}
    \caption{Recovery of the Fisher information rate on the
    source~\eqref{eq:exp1-model}. The solid black curve is the closed
    form~\eqref{eq:exp1-truth}, the coloured curves are the Monte Carlo
    means at the three trajectory lengths, and the shaded bands are one
    sample standard deviation.}
    \label{fig:exp1-landscape}
\end{figure}

Figure~\ref{fig:exp1-landscape} shows the outcome. The estimator
reproduces the whole ground truth reference, from the maximum at $\theta=0$, where
the chain is closest to a fair coin and the information is largest, to
the decay on both sides, where the transition probabilities approach
zero and one and each observation carries less information about the
parameter. The reference is never seen by the estimator, which
observes only simulated trajectories at the probed parameter values.
The three trajectory lengths differ in the way expected. At $n=10^3$
the Monte Carlo mean falls below the reference near the maximum, and
the one-standard-deviation band is at its widest there. Both the gap
and the band shrink monotonically as $n$ grows, and at $n=10^4$ the
mean is close to the reference over the whole range. 

\subsubsection{Convergence and the rate}
Theorem~\ref{thm:fim-rate} also gives a convergence rate, and a
different one for each differencing scheme, together with the step size
at which it is attained: $O(n^{-1/4})$ for the one-sided scheme at
$\varepsilon_n\asymp n^{-1/4}$, and $O(n^{-1/3})$ for the central
scheme at $\varepsilon_n\asymp n^{-1/6}$. Both schemes are run here on
the same trajectories, each with its own step-size schedule,
$\varepsilon_n=6\,n^{-1/4}$ and $\varepsilon_n=4.5\,n^{-1/6}$, and each
scheme is repeated $R=100$ times. The two fitted slopes are then
compared with the two exponents.

\begin{figure}[t]
    \centering
    \includegraphics[width=0.62\linewidth]{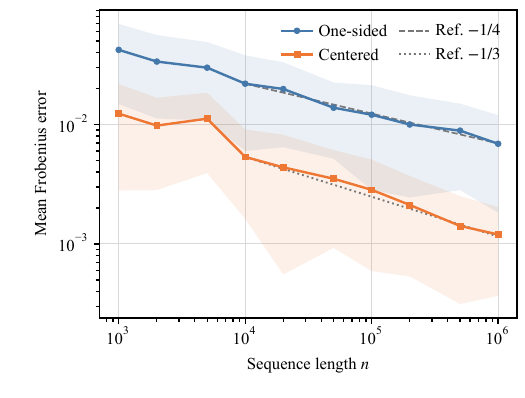}
    \caption{Empirical convergence of the one-sided and
    central-difference estimators on the source~\eqref{eq:exp1-model}.
    Each curve is the mean of
    $\|\widehat I-I(\theta_0)\|_F$ over repetitions, the unnormalized
    quantity bounded in Theorem~\ref{thm:fim-rate}, and the shaded band
    is one sample standard deviation. The grey lines show the exponents
    $-1/4$ and $-1/3$ of Theorem~\ref{thm:fim-rate}.}
    \label{fig:exp1-rate}
\end{figure}

Figure~\ref{fig:exp1-rate} reports the outcome. The fitted exponents
are $-0.250$ and $-0.335$, against the values $-1/4$ and $-1/3$ provided by
Theorem~\ref{thm:fim-rate}, respectively. The bound is an upper bound, and on this
source it is attained.
Below $n=10^{4}$ the decay is irregular and the fit is taken above that
point, since all three terms of~\eqref{eq:fim-matrix-rate} are still of
comparable size there and no single exponent should be expected. The
step size itself is confirmed independently in
Appendix~\ref{app:eps-sweep}, where the empirically optimal
$\varepsilon$ decays at the exponent the theorem prescribes
and the schedule used here sits on that optimum.

The central-difference scheme does not cancel statistical fluctuation,
since its two perturbed trajectories are independent. What it reduces
is the deterministic bias, and the smaller bias is what permits the
larger step size.

\subsection{Multi-parameter source: recovering the full
matrix}
\label{subsec:exp2}

This experiment estimates a matrix rather than a scalar, which puts
three steps of Algorithm~\ref{alg:fim-recovery} to the test: the
directional measurements must carry the mixed curvature, the
least-squares step must return the off-diagonal entries and not only
the diagonal, and the projection onto the semidefinite cone must act as
a safeguard rather than as a repair. A single sample size is used, so
the convergence rate is not tested again here.

The source is a four-state Markov chain on the states
$\{0,1,2,3\}$. From
state $i$ the chain moves to state $j$ with probability
\begin{equation}
\begin{aligned}
P_{ij}(\theta)&=\frac{e^{\theta^\top a_{ij}}}
{\sum_{k=0}^{3}e^{\theta^\top a_{ik}}},\\[2pt]
a_{ij}&=\bigl(\mathbf 1\{j=i\},\;c_{ij},\;
\mathbf 1\{j\le1\}\bigr)^\top,
\end{aligned}
\label{eq:exp2-model}
\end{equation}
where $c_{ij}$ equals
$+1$ if $j=i+1$, $-1$ if $j=i-1$, and $0$
otherwise, with indices taken modulo $4$. Each entry of
$\theta=(\theta_1,\theta_2,\theta_3)^\top\in[-1,1]^3$ weights one
feature. The parameter $\theta_1$ controls the tendency to stay put,
$\theta_2$ the preference for $j=i+1$ over $j=i-1$, and $\theta_3$ the
preference for the states $0$ and $1$. Writing
$s_{ij}(\theta)=a_{ij}-\sum_{k}P_{ik}(\theta)\,a_{ik}$ for the score of
a transition, the Fisher information rate is
\begin{equation}
I(\theta)=\sum_{i=0}^{3}\pi_\theta(i)\sum_{j=0}^{3}
P_{ij}(\theta)\,s_{ij}(\theta)\,s_{ij}(\theta)^\top ,
\label{eq:exp2-truth}
\end{equation}
which is the ground truth for this experiment.

Twelve parameter vectors are used, the eight vertices of $[-1,1]^3$ and
four interior points. Everywhere on the grid the transition
probabilities stay bounded away from zero and one, as
Assumption~\ref{ass:reg}(i) requires. The off-diagonal entries take
both signs there, so the recovery has to resolve the sign of the mixed
curvature and not only its magnitude.

Being symmetric, $I(\theta)$ has six free entries, and six directions
would determine it. Nine directions are used instead.
They combine the minimal design of
Remark~\ref{rem:directions} with the three mixed directions
of the opposite sign, which halves the conditioning constant of
Lemma~\ref{lem:ls-continuity}.

For each of the $R=100$ replications, we draw one reference trajectory of
length $n=10^{5}$ and one trajectory for each of the $18$ perturbed values
$\theta\pm\varepsilon u_j$, $19$ in total, where each trajectory is started from its own
stationary law. The working depth is $L=1$, and the measurements are
central with the step size $\varepsilon_n=2.4\,n^{-1/6}$ that
Theorem~\ref{thm:fim-rate}(b) prescribes for that scheme.

\begin{figure}[t]
    \centering
    \includegraphics[width=0.72\linewidth]{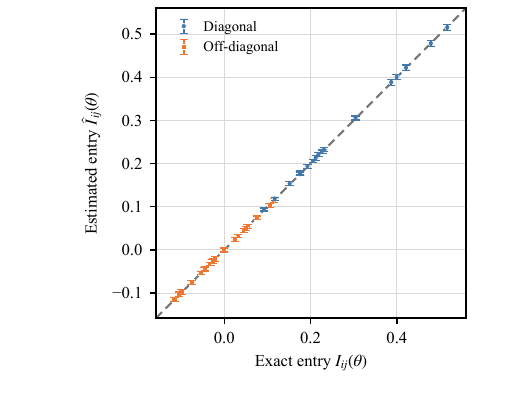}
    \caption{Recovery of the matrix-valued Fisher information rate for
    the source~\eqref{eq:exp2-model}. Each point compares one entry of
    the repetition-averaged estimate with its exact
    value~\eqref{eq:exp2-truth}. Circles are diagonal entries and
    squares are off-diagonal ones. Error bars are one sample standard
    deviation across repetitions, and not the standard error of the
    plotted mean. The dashed line is $y=x$.}
    \label{fig:exp2-entries}
\end{figure}

Figure~\ref{fig:exp2-entries} compares every entry with its exact
value. The points follow the identity line over the whole range, and
the off-diagonal group crosses the origin, so entries of both signs are
recovered. A single run recovers the matrix to within three per cent
in relative Frobenius error,
most of which is variation between runs rather than a systematic
offset. The off-diagonal block is about three times less accurate than
the matrix as a whole, partly because several off-diagonal entries pass
through zero on this grid, where any relative measure inflates. What
does not appear is a contraction towards zero or a sign error, which are the two
failure modes a directional design could plausibly produce.

The projection onto the semidefinite cone was never active: in all
$1200$ reconstructions the unconstrained least-squares estimate was
already positive semidefinite. The accuracy above is therefore that of
the raw fit. The experiment shows that the projection is not needed on
this source, and says nothing about how well it would repair an
indefinite fit.

\subsection{Choosing the working depth}
\label{subsec:exp-depth}

The working depth $L$ is the one choice the estimator leaves to its
user, and the theory asks only that it satisfy $L\ge L^\star$. This
experiment verifies that any depth meeting the condition performs as
well as $L=L^\star$ itself, and measures what violating it costs. The
memory length of the source is known to us but unknown to both
estimators.

The source is a binary chain of memory length $L^\star=4$. Writing
$z_k(s)=2s_k-1$ for the sign of the $k$th symbol of a context
$s\in\{0,1\}^{4}$, and $\phi(s)=z_1(s)z_2(s)z_3(s)z_4(s)$ for the
product of the four signs, the next symbol is drawn from
\begin{equation}
P_\theta(Y_t=1\mid s)=\sigma\bigl(\beta_s+\theta\,\phi(s)\bigr),
\qquad
\beta_s=0.2+0.25z_{4}(s)-0.15z_{3}(s)+0.35\phi(s).
\label{eq:depth-model}
\end{equation}
The parameter multiplies $\phi(s)$, which flips whenever any one of the
four symbols flips. No model of order three or less can represent it,
and such a model therefore reports no dependence on $\theta$ at all.
The state process of Assumption~\ref{ass:source} takes $16$ values
here, and the Fisher information rate is
\begin{equation}
I(\theta)=\sum_s\pi_\theta(s)\,p_\theta(s)\bigl(1-p_\theta(s)\bigr)\phi(s)^2 ,
\label{eq:depth-truth}
\end{equation}
which is the ground truth here. We take $\theta_0=-1.7$, where the
transition probabilities stay bounded away from zero and one as
Assumption~\ref{ass:reg}(i) requires.

A second estimator is needed for comparison, and it has to work under
the same constraints. According to Assumption~\ref{ass:sim}, only simulated
trajectories are available, so the likelihood cannot be evaluated and its parameter
gradient cannot be formed. Estimators built from the score are out of
reach for that reason, including the Monte Carlo constructions
of~\cite{spall2005monte}, which still need the model itself and not
only its output. A fair comparison is obtained instead by keeping
Algorithm~\ref{alg:fim-recovery} intact and replacing the CTW code and
scorers by fixed-order counting models. For a chosen depth $L$ this
baseline treats every string $s\in\mathcal A^{L}$ as a separate context
and estimates the distribution of the next symbol as
\begin{equation}
\widehat q_{L}(a\mid s)=\frac{N(s,a)+1/2}{N(s)+1}.
\label{eq:baseline-count}
\end{equation}
One such model is fitted to the reference trajectory and one to each
perturbed trajectory, and the reference trajectory is then scored with
all of them, exactly as for the frozen CTW scorers. The proposed
estimator and the baseline therefore run on the same data with the same
perturbations and the same recovery, and differ only in how the next
symbol is predicted.

Two trajectory lengths are used, $n=3\times10^{4}$ and $10^{5}$, each
with $R=100$ repetitions, and the working depth ranges over
$L=1,\dots,18$. Both estimators see exactly the same data at every
depth, namely, one reference trajectory and one at each of the two perturbed
parameter values, drawn once per repetition. The central-difference
step follows the prescribed schedule $\varepsilon_n=7.5\,n^{-1/6}$,
whose constant is fixed on an independent pilot run at $L=L^\star$ that
minimizes the error of the two jointly, so the step size favours
neither.

Accuracy is reported as the mean relative error
$\mathbb E\bigl[|\widehat I-I(\theta_0)|/I(\theta_0)\bigr]$. It is $0$
for an exact estimate and $1$ for an estimate that carries no
information about $\theta$, so smaller is better and $1$ is the natural
reference level.

\begin{figure}[t]
    \centering
    \includegraphics[width=0.62\linewidth]{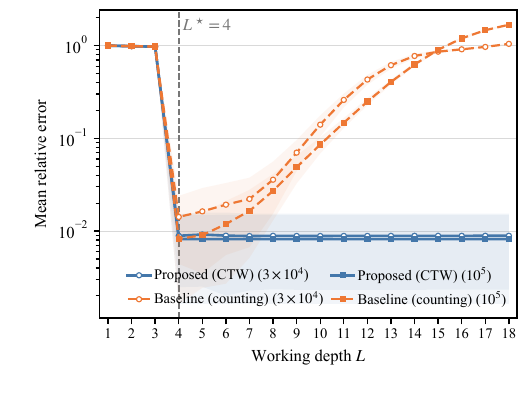}
    \caption{Robustness to the choice of memory depth on the
    source~\eqref{eq:depth-model}. The curves are the mean relative
    error against the depth $L$ chosen by the user, and the shaded bands
    are one sample standard deviation. In the legend, the
    number in parentheses is the trajectory length $n$, either
    $3\times10^{4}$ or $10^{5}$. The dashed line marks $L^\star$.}
    \label{fig:exp-depth}
\end{figure}

Figure~\ref{fig:exp-depth} falls into three ranges of the working
depth: below $L^\star$, at $L^\star$, and above it. Below $L^\star$ both estimators sit at $1$. The depth is too
small for the model class to see $\theta$ at all, and no amount of data
changes that. At $L=L^\star$ the two estimators coincide, the baseline being then
correctly specified. Above $L^\star$ they separate. The proposed
estimator is flat from $L=4$ to $L=18$, and its error at the largest
depth matches its error at $L^\star$. This is the behaviour
the CTW weighting is built for. Raising $L$ adds deeper candidate
trees, but the weight a tree receives depends on its own size and not
on $L$ once it exceeds the depth of the tree, so the shallow trees that fit the
data keep their weight and the estimate is unchanged; the extra depth
costs computation, but not accuracy. The baseline deteriorates steadily over the same range,
by about two orders of magnitude, and beyond $L\approx14$ its error
exceeds $1$.

The proposed estimator is flat above $L^\star$ at both trajectory
lengths, so the robustness is a property of the
estimator and not of the sample size. In
practice the condition $L\ge L^\star$ is therefore easy to satisfy. A
user who does not know the memory length can choose the depth
generously large, and pays for it only in computation,
which grows linearly in $L$. Below $L^\star$ neither
estimator can recover, and the cliff does not move with $n$, so a depth
that is too large is the only safe mistake. The baseline gives the user
no such freedom, and its loss above $L^\star$ can at least be reduced
by more data, which pushes the degradation to the right.

An order-selection rule would let the baseline recover
$L^\star$ from enough data. This is no objection to the comparison.
The CTW weighting already performs such a selection internally, and
Theorem~\ref{thm:fim-rate} holds uniformly over $L\ge L^\star$, so the
proposed estimator needs no selection at all, whereas for the two-stage
procedure that first selects an order and then counts at it, no
guarantee comparable to Theorem~\ref{thm:fim-rate} is known to us.

\subsection{Hidden-state source}
\label{subsec:exp3}

The previous experiments stayed within Assumption~\ref{ass:source}.
This experiment takes the estimator to a source that has no finite
memory length, so that assumption fails and
Theorem~\ref{thm:fim-rate} no longer applies, and it shows that the two
properties established so far, the recovery of the Fisher information
and the tolerance to a generous working depth, both survive. The source
is a two-state hidden Markov model. A latent chain switches state with
probability $q=0.1$ at each step, and the observed binary symbol is
emitted with a probability that depends on the current latent state
through
\begin{equation}
e(\theta)=0.05+0.4\,\sigma(3\theta),
\label{eq:exp3-model}
\end{equation}
which carries the scalar parameter. The estimation point is
$\theta_0=0$. The observation process is not a Markov chain of any
finite order, so every model of fixed order is misspecified by
construction. The reference value has no closed form here and is itself
estimated, from $8$ trajectories of length $10^{6}$, with a standard
error far below the errors compared.

Each of the $R=100$ repetitions uses $10^{6}$ symbols per scorer,
pooled from $N=20$ trajectories of length $5\times10^{4}$, and the
working depth ranges over $L=1,\dots,12$. Both estimators see exactly
the same data at every depth, namely, one reference
sample and one at each of the two perturbed parameter values, drawn once
per repetition. The
central-difference step is $\varepsilon=0.18$, fixed on an independent
pilot run of the proposed estimator at $L=8$.

\begin{figure}[t]
    \centering
    \includegraphics[width=0.62\linewidth]{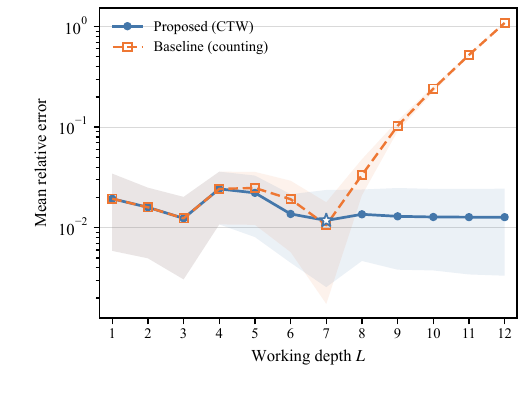}
    \caption{Robustness to the working depth on the hidden-state
    source~\eqref{eq:exp3-model}. Curves are the mean relative error
    against the depth chosen by the user, shaded regions are the
    $16$th to $84$th percentiles of the per-repetition error, and open
    squares remain visible where the baseline coincides with the
    proposed estimator. Note the logarithmic vertical scale.}
    \label{fig:exp3-depth}
\end{figure}

Figure~\ref{fig:exp3-depth} presents the results of this experiment, which fall into three ranges of the working
depth. Up to $L=4$ the two estimators coincide, since every context is
then visited often enough for the CTW weights to concentrate on the full
tree, which is the model the baseline commits to. The small rise at
$L=4$ appears in both curves, so it comes from the model class at that
depth and not from either estimator. The two curves begin to separate
at $L=5$, and the gap opens widely from $L=8$ on. The proposed
estimator holds a relative error near $1\%$ from $L=6$ to $L=12$, and
its error at $L=12$ matches its error at the best depth. The baseline
degrades monotonically over the same range, and at $L=12$ its error
exceeds $1$.

At $L=7$, where the baseline attains its own minimum, the two
estimators are statistically indistinguishable, the paired difference
over the $100$ repetitions being $0.12$ percentage points with a
standard error of $0.07$. That minimum is narrow, in the sense that one step in either
direction roughly doubles or triples the baseline's error, and nothing
in the data identifies $L=7$ in advance. The proposed estimator asks
for an upper bound on the depth instead of a value, and pays nothing
across the whole range in which that bound holds.

No rate is claimed from this experiment, which uses a single sample
size on a source the theory does not cover. What it shows is that the
estimator still returns the Fisher information of a source whose score
is available only through a forward recursion, and still absorbs an
over-generous working depth that the baseline does not.

\subsection{Planning a sampling budget under unknown
memory length}
\label{subsec:exp4}

This experiment uses the estimate for an application of sample-size
determination~\cite{adcock1997sample}. Given a model of the process
and a target precision, the task is to find the shortest trajectory
from which a quantity of the model can be estimated to that precision.
The Fisher information is estimated from simulator output, which
Assumption~\ref{ass:sim} declares available; the budget it yields is
then spent on data from the process itself, and the experiment reports
whether the prescribed precision is attained. The cost of the two kinds
of data is not compared.

The source is a binary Markov chain of order $3$. The probability
that the next symbol is $1$ depends on the past only through which of
four patterns it ends in,
\[
0,\qquad 01,\qquad 011,\qquad 111,
\]
written with the most recent symbol last, and equals
$\sigma(\theta_k)$ for pattern $k$, so that
$\theta=(\theta_1,\dots,\theta_4)^\top$ and
$\sigma(\theta_0)=(0.50,0.30,0.85,0.15)$. The quantity to be estimated
is $g^\star=\Pr(Y_t=1\mid\text{the past ends in }011)=\sigma(\theta_3)
=0.85$, which is a function of the third parameter. That pattern is visited with stationary
probability $1/12$, and $I_{33}=\pi(011)\,g^\star(1-g^\star)$ in closed
form.

For a tolerance $\tau$ on $g^\star$, the Cram\'er--Rao bound gives
$\operatorname{var}(\widehat\theta_3)\ge1/(nI_{33})$, and a first-order
expansion of $g=\sigma(\theta_3)$, the delta
method~\cite{van2000asymptotic}, carries it to
$\operatorname{var}(\widehat g)\approx
\bigl(g^\star(1-g^\star)\bigr)^2/(nI_{33})$, so that
\begin{equation}
n^\star
=\Bigl\lceil
\bigl(\widehat g(1-\widehat g)\bigr)^2\big/
\bigl(\tau^2\,\widehat I_{33}\bigr)
\Bigr\rceil
\label{eq:budget}
\end{equation}
is the smallest length for which the tolerance $\tau$ is attainable.
The budget $n^\star$ is computed at the true parameter value
$\theta_0$, the usual local assumption of experimental design, and this
affects every estimate of $I_{33}$ below equally. Throughout,
$g^\star$ is estimated from a trajectory by the empirical conditional
frequency $N(011,1)/N(011)$, i.e., the number of times the pattern $011$ is
followed by a $1$ divided by the number of times it occurs, in the
count notation of~\eqref{eq:baseline-count}. The pipeline has five
steps.
\begin{enumerate}
\item $I_{33}$ is estimated from simulator samples.
\item The budget $n^\star$ is formed from~\eqref{eq:budget}, with
$\widehat g$ the empirical conditional frequency on the simulator
sample at $\theta_0$.
\item A trajectory of length $n^\star$, the validation path, is
collected from the process.
\item $g^\star$ is estimated from the validation path by the empirical
conditional frequency, as in step~2.
\item The realized error is compared with the tolerance $\tau$.
\end{enumerate}

We compare five configurations, each a complete run of the pipeline
above, differing only in how $I_{33}$ is obtained in the first step.
They follow what a user who does not know $L^\star$ can do. With the proposed estimator the guidance of
Sections~\ref{subsec:exp-depth} and~\ref{subsec:exp3} is to choose the
depth generously, so it is run at one depth, $L=5$. With counting there
is no safe direction and the user must guess, so the baseline is run at
$L=2$, $L=3$ and $L=5$, which are respectively one
order too small, the correct order, and the
same generous depth as the proposed estimator. Below $L^\star$ the
proposed estimator fails as counting does, as
Figure~\ref{fig:exp-depth} showed, and that case is not repeated. The
fifth configuration is the analytic reference, which uses the
exact value of $I_{33}$. It shows how close the budget rule~\eqref{eq:budget} itself comes to
the tolerance $\tau$ when $I_{33}$ is known exactly, and the other four
configurations are compared with it.

Each of $K=1000$ repetitions draws $10$ simulator trajectories of
length $10^{4}$ per scorer, with central measurements at
$\varepsilon=0.5$, and one long enough validation path from the process, started from
the stationary law.
The five configurations use the same simulator samples and, at each
tolerance, the first $n^\star$ symbols of that validation path. The
tolerances are
$\tau\in\{0.10,0.05,0.02,0.01\}$, for which the analytic budgets are
$153$, $612$, $3825$ and $15300$ symbols.

\begin{figure*}[t]
    \centering
    \includegraphics[width=\linewidth]{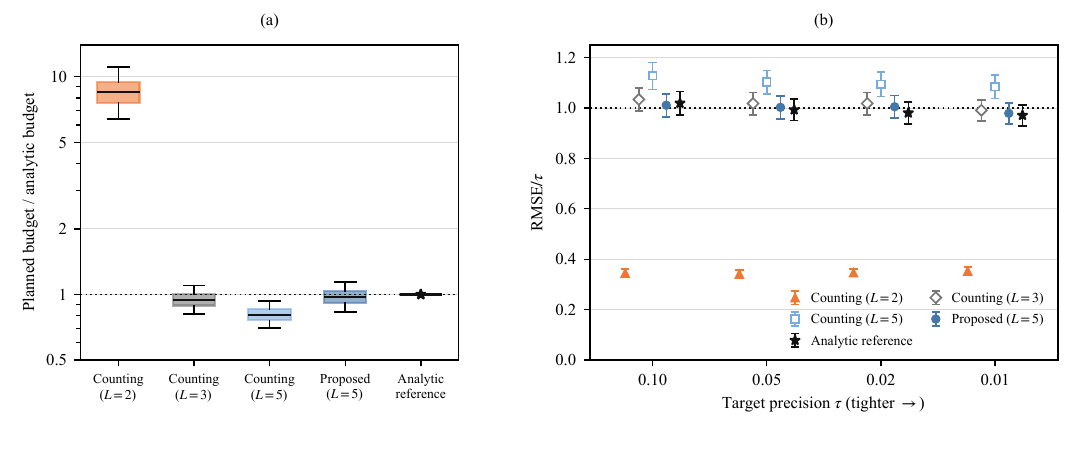}
    \caption{Sample-size determination from an estimated Fisher
    information rate. (a)~Planned budget relative to the analytic budget, over the
    repetitions; boxes span the quartiles, the line is the median, the
    whiskers are the $5$th and $95$th percentiles, outliers are not
    drawn, and the vertical axis is logarithmic. (b)~Realized root
    mean-square error of $\widehat g$ divided by the tolerance $\tau$,
    for each tolerance; error bars are $95\%$ bootstrap intervals. The
    dotted lines mark the analytic budget in~(a) and the tolerance
    in~(b).}
    \label{fig:exp4-planning}
\end{figure*}

In Figure~\ref{fig:exp4-planning}, the budget a
configuration plans in~(a) fixes the precision it attains in~(b),
so the two panels need to be read together. The analytic reference calibrates the rule itself.
With $I_{33}$ known exactly, the realized error stays within about
$3\%$ of the tolerance $\tau$ at every $\tau$, which is the accuracy of
the Cram\'er--Rao bound and the delta method at between $13$ and
$1300$ visits to the target context, and the other configurations are
judged at that resolution.

Counting at $L=2$ asks for more than $8$ times the analytic budget. At
order $2$ the memory is too short to separate the pattern $011$ from
$111$, most of the information about $\theta_3$ is lost, and the rule
compensates for information it cannot see. The error it attains is
about a third of the tolerance, but at a price far beyond what the task
requires. Counting at $L=3$, the correct order, plans about $6\%$ below
the analytic budget and meets the tolerance within the resolution of
the reference. Counting at $L=5$ plans about $20\%$ below the analytic
budget, because the over-specified model over-states $I_{33}$, and its
realized error runs about $10\%$ above the tolerance at every $\tau$,
with all four intervals above one. The proposed estimator, run at the
same depth $L=5$ on the same trajectories with the same step and
recovery, plans within $3\%$ of the analytic budget, closer to it than
counting at the correct order, and meets the tolerance at every $\tau$
within the resolution of the reference. The small amount by which its budget still falls short is a mild form
of the over-specification effect and shrinks with the simulator
sample.

The experiment thus shows that the proposed estimator, given a generous
depth rather than the memory length, plans the budget accurately,
within $3\%$ of the analytic one, and attains the prescribed precision
at every tolerance. A fixed-order estimator comes close to this only at
the correct order; at any other order it either asks for an order of
magnitude more data than needed or misses the tolerance, so it requires
the knowledge of the memory length that the proposed estimator does
without.

\section{Conclusion}
\label{sec:conclusion}

We estimated the Fisher information rate of a stationary process with
memory from simulator output alone. The estimator measures the local
curvature of the KL divergence with context-tree weighting and recovers
the matrix by least squares. It needs neither the likelihood nor the
score, and it needs no knowledge of the memory length. The user
supplies a working depth that is large enough, and the guarantees hold
uniformly above that threshold.

The estimator is strongly consistent in an iterated
limit, and with the step size
$\varepsilon_n\asymp n^{-1/6}$ its mean error is of order $n^{-1/3}$.
The same analysis gives a finite-sample rate for CTW-based estimation
of the KL divergence rate. Both exponents are confirmed numerically,
and the accuracy does not degrade as the working depth grows well
beyond the memory length. The sampling budgets planned from the
estimate deliver the requested precision, and the estimate stays
accurate on a hidden-state source that lies outside the assumptions.

The analysis is confined to finite-valued observations and to a finite
memory length. Continuous-valued sources are reached only through
quantization, and processes with infinite memory are supported by the
experiments alone. The weighted and regularized variants of the
recovery step in Section~\ref{subsec:practical} improve the
finite-sample behaviour but are not covered by the present rate
analysis. Each of these is a natural direction for further work.

\appendices
\section{Additional Experimental Results}
\label{app:extra-experiments}

\subsection{Step-size trade-off}
\label{app:eps-sweep}

This experiment separates the terms of the bound in
Theorem~\ref{thm:fim-rate} by fixing $n$ and sweeping $\varepsilon$,
and compares the step-size schedule used in Figure~\ref{fig:exp1-rate}
with the empirically optimal step. The bound has two terms that grow
as $\varepsilon$ shrinks, entering as $\varepsilon^{-1}$ and
$\varepsilon^{-2}$ and decaying with $n$, and one differencing bias of
order $\varepsilon^{2}$ that does not involve $n$. At fixed $n$ the
error is therefore $U$-shaped in $\varepsilon$, its right branch is
common to all $n$, and its minimizer decays as
$\varepsilon^\star_n\asymp n^{-1/6}$.

The source is the two-state chain~\eqref{eq:exp1-model} of
Section~\ref{subsec:exp1} at $\theta_0\in\{1,2\}$, with working depth
$L=1$. For each $(\theta_0,n,\varepsilon)$ one reference trajectory and
one trajectory at each of $\theta_0\pm\varepsilon$ are drawn, the
central measurement~\eqref{eq:central-meas} is formed and its negative
part is clipped, the scalar case of the projection onto the
semidefinite cone, and every combination is repeated $R=100$ times.
The reference trajectory and the random numbers of each perturbation
sign are shared across $\varepsilon$, so that the curves are paired
along the horizontal axis. The parameter is scalar, so the Frobenius
error is the absolute error.

\begin{figure}[t]
    \centering
    \includegraphics[width=\linewidth]{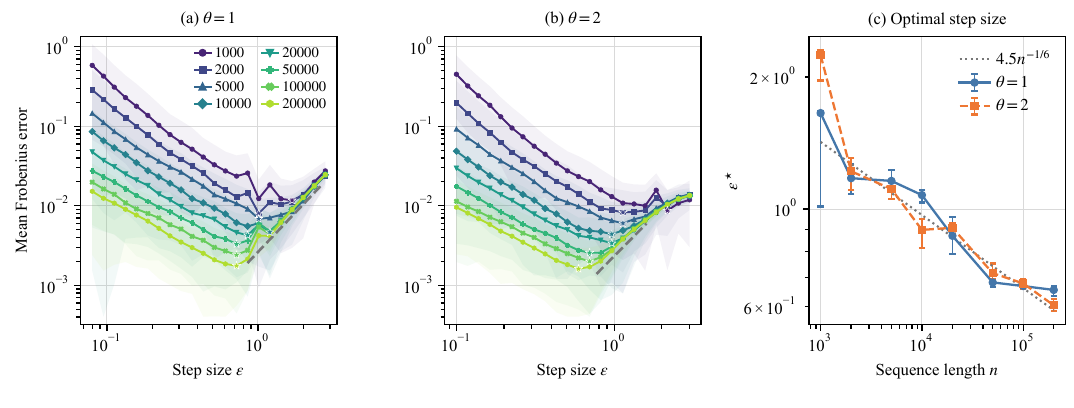}
    \caption{Mean Frobenius error of the central-difference estimator
    against the step size, on the source~\eqref{eq:exp1-model}, at
    (a)~$\theta_0=1$ and (b)~$\theta_0=2$; colours and markers give the
    trajectory lengths listed in~(a), and shaded bands are one sample
    standard deviation. Stars mark the grid minima, and the grey dashed
    segments have slope $+2$, which is the exponent of the central-difference
    bias. (c)~Interpolated empirical minimizers
    $\widehat\varepsilon^\star_n$ with paired-bootstrap $16$th to
    $84$th percentile intervals; the dotted line is the schedule
    $4.5\,n^{-1/6}$ of Figure~\ref{fig:exp1-rate}.}
    \label{fig:app-eps}
\end{figure}

Figures~\ref{fig:app-eps}(a) and \ref{fig:app-eps}(b) show the $U$-shape at both
parameter values, and the right branches of the different lengths fall
onto a common envelope, as they must if that branch is the
$n$-independent bias. The slope fitted to that branch on the large-$n$
curves lies between $1.99$ and $2.05$ at $\theta_0=1$. At $\theta_0=2$
it is close to $2$ for $\varepsilon$ up to about $1.6$ and departs
beyond, where the step leaves the range of the local expansion, and the
reference segment is drawn over that range only. On the left branch
the two $n$-dependent terms act together and no single slope is
expected.

Figure~\ref{fig:app-eps}(c) shows the minimizers, obtained by fitting
a parabola to $\log\mathrm{Err}$ against $\log\varepsilon$ on the grid
minimum and its two neighbours on each side, with uncertainty from
$2000$ paired bootstrap replicates. Fitting
$\log\widehat\varepsilon^\star_n=a+\gamma\log n$ over $n\ge10^{4}$
gives $\gamma=-0.168$ at $\theta_0=1$ and $-0.144$ at $\theta_0=2$,
with bootstrap $95\%$ intervals $[-0.202,-0.132]$ and
$[-0.178,-0.080]$; both contain the exponent $-1/6$ of
Theorem~\ref{thm:fim-rate}. The dotted line is not fitted to these
points. It is the schedule $\varepsilon_n=4.5\,n^{-1/6}$ of
Figure~\ref{fig:exp1-rate}, whose constant was fixed on a separate
pilot run, and it follows the empirical minimizers over the whole
range at both parameter values. On this source the step size the
theorem prescribes is therefore, to the resolution of the experiment,
also the best choice.

\section{The Context-Tree Weighting Algorithm}
\label{app:ctw}

This appendix details the CTW mechanism used to form the
frozen scorers in Section~\ref{subsec:kl-est}: the
Krichevsky--Trofimov leaf estimator, the weighted recursion over the
context tree, the sequential implementation of the CTW block
probability, and the positivity bound~\eqref{eq:ctw-positivity-method}
together with its counterpart for the frozen scorer.

Let the finite alphabet be $\mathcal{A}=\{0,\dots,|\mathcal{A}|-1\}$.
For a context-tree node $s$, let $T(s)$ collect the time indices at
which the past ends in $s$, so that node $s$ lies on the
CTW update path, let
$y^{(s)}=\{y_t\colon t\in T(s)\}$ be the subsequence emitted from that
context, $n_s=|T(s)|$ its length, and $b_{s,a}$ the number of times
symbol $a$ occurs in $y^{(s)}$.

\subsection{The Krichevsky-Trofimov Leaf Estimator}
\label{app:kt}

Each leaf of the context tree models the symbols emitted from its
context as an i.i.d.\ sequence, and assigns them a probability by
averaging over the unknown i.i.d.\ parameter. %
Placing a $\mathrm{Dir}(\tfrac12,\dots,\tfrac12)$ prior on that
parameter and integrating it out yields the Krichevsky-Trofimov (KT)
estimate~\cite{krichevsky1981performance}, which for the leaf
subsequence $y^{(s)}$ takes the closed form
\begin{equation}
P_e^s
= \frac{\Gamma\!\left(\tfrac{|\mathcal{A}|}{2}\right)}
       {\Gamma\!\left(n_s+\tfrac{|\mathcal{A}|}{2}\right)}
  \prod_{a\in\mathcal{A}}
  \frac{\Gamma\!\left(b_{s,a}+\tfrac12\right)}{\Gamma\!\left(\tfrac12\right)}.
\label{eq:kt-marginal}
\end{equation}%

Rather than evaluating~\eqref{eq:kt-marginal} in closed form, CTW
updates the KT estimate one symbol at a time. By the
Dirichlet-multinomial conjugacy~\cite{Bernardo-Smith-09}, observing a new symbol leaves the
posterior within the Dirichlet family, and the probability assigned to
the next symbol is the posterior mean
\begin{equation}
P_e^s\bigl(Y_{t+1}=a\mid y_{1:t}\bigr)
= \frac{b_{s,a}(t)+\tfrac12}{n_s(t)+\tfrac{|\mathcal{A}|}{2}},
\label{eq:kt-predictor}
\end{equation}%
where $b_{s,a}(t)$ and $n_s(t)$ denote the counts accumulated at node
$s$ up to time $t$. The leaf probability~\eqref{eq:kt-marginal} is
recovered as the product of the predictors~\eqref{eq:kt-predictor} over
the symbols emitted from $s$. The form of~\eqref{eq:kt-predictor}
exhibits the ``pseudocount'' of $\tfrac12$ assigned to every symbol,
ensuring that no symbol is ever predicted with zero probability.

\subsection{The Weighted Recursion}
\label{app:recursion}

CTW does not commit to a single context depth. Instead it forms, at
each node, a convex combination of two models: the KT estimate that
treats the node itself as a leaf, and the product of the probabilities
returned by the node's children, which corresponds to splitting the
context one symbol deeper. %
Writing $P_w^s$ for the weighted probability at node $s$, the recursion
sets $P_w^s=P_e^s$ at every leaf, and at every internal node
\begin{equation}
P_w^s
= \tfrac12\,P_e^s
+ \tfrac12\prod_{a\in\mathcal{A}}P_w^{as},
\label{eq:ctw-recursion}
\end{equation}%
the weight $\tfrac12$ being the one induced by the CTW prior over
context-tree topologies. Carrying~\eqref{eq:ctw-recursion} from the
leaves up to the root produces $P_w^{\mathrm{root}}$, which is the
CTW mixture probability assigned to the processed block. In the
notation of Definition~\ref{def:frozen-ctw-scorer}, the same CTW
recursion is used to build a frozen scorer from one block and then
evaluate its probability on a scoring block.
For such a scorer the node counts, the KT predictors
and the weights are computed once from the construction path; scoring
a new block then reads off the resulting predictive distribution at
each scoring context, without updating the counts.
Equation~\eqref{eq:ctw-recursion} is the computational counterpart of
the Bayesian mixture over all bounded-depth tree sources: the convex
weight at each node performs, in linear time, the averaging over tree
topologies that a direct enumeration could not afford.

\subsection{Sequential Form and Positivity}
\label{app:sequential}

The recursion~\eqref{eq:ctw-recursion} also yields a sequential
implementation of the CTW block probability; the root
mixture probability is the sequential code of the main text,
$Q^{\mathrm{CTW}}_{L}(y_{1:t})=P_w^{\mathrm{root}}(y_{1:t})$.
The one-step predictor at
the root is the ratio of consecutive mixture probabilities,
\begin{equation}
Q^{\mathrm{CTW}}_{L}\bigl(y_{t+1}\mid y_{1:t}\bigr)
= \frac{P_w^{\mathrm{root}}(y_{1:t+1})}{P_w^{\mathrm{root}}(y_{1:t})}.
\label{eq:ctw-ratio}
\end{equation}%
To express this ratio nodewise, associate with each node $s$ the
quantity
\begin{equation}
\beta_t^s := \frac{\tfrac12\,P_e^s(y_{1:t})}{P_w^s(y_{1:t})}\in(0,1),
\label{eq:beta}
\end{equation}%
which measures how much weight the recursion at $s$ places on the
node's own KT estimate relative to the contribution of its children.
Only the nodes lying on the path of the current context are affected
by a new symbol; along that path, the one-step predictor at $s$ is the
convex combination
\begin{equation}
P_w^s\bigl(Y_{t+1}=a\mid y_{1:t}\bigr)
= \beta_t^s\,P_e^s\bigl(Y_{t+1}=a\mid y_{1:t}\bigr)
+ \bigl(1-\beta_t^s\bigr)\,
  P_w^{c}\bigl(Y_{t+1}=a\mid y_{1:t}\bigr),
\label{eq:ctw-sequential}
\end{equation}%
where $c$ is the child of $s$ along the current context. Applying
\eqref{eq:ctw-sequential} from the root down to the active leaf
reproduces the root predictor~\eqref{eq:ctw-ratio}.
Since a new symbol affects only the $L+1$ nodes on its
context path, a block of length $n$ requires $O(nL)$ node updates, a
cost linear in $n$ for a fixed working depth.

The recursion~\eqref{eq:ctw-sequential} also yields the positivity
bound stated in the main text. %
The KT predictor~\eqref{eq:kt-predictor} satisfies
\begin{equation}
P_e^s\bigl(Y_{t+1}=a\mid y_{1:t}\bigr)
\;\ge\;\frac{\tfrac12}{t+\tfrac{|\mathcal{A}|}{2}}
\;=\;\frac{1}{2t+|\mathcal{A}|},
\label{eq:kt-floor}
\end{equation}%
since the count $b_{s,a}(t)$ is non-negative and $n_s(t)\le t$. Because
the root predictor is, by~\eqref{eq:ctw-sequential}, a convex
combination along the active path in which the KT term always appears,
this lower bound propagates to the root, \ie,
\begin{equation}
Q^{\mathrm{CTW}}_{L}\bigl(a\mid y_{1:t}\bigr)
\;\ge\;\frac{1}{2t+|\mathcal{A}|},
\label{eq:ctw-positivity-app}
\end{equation}%
establishing~\eqref{eq:ctw-positivity-method}.

The same argument bounds the frozen scorer of
Definition~\ref{def:frozen-ctw-scorer}. Its KT predictors are formed
from the counts of the completed construction pass over $x_{1:n}$, so $n_s\le n$
at every node and
\begin{equation}
P_e^s\bigl(a\bigr)\;\ge\;\frac{\tfrac12}{n+\tfrac{|\mathcal{A}|}{2}}
\;=\;\frac{1}{2n+|\mathcal{A}|}.
\label{eq:kt-floor-frozen}
\end{equation}
Since the frozen one-step predictor at a scoring context is again a
convex combination along the corresponding path in which the KT term
appears, the same lower bound propagates,
\begin{equation}
\widehat T^{\mathrm{CTW}}_{x_{1:n},L}\bigl(a\mid s\bigr)
\;\ge\;\frac{1}{2n+|\mathcal{A}|}.
\label{eq:frozen-positivity}
\end{equation}

The
bound~\eqref{eq:ctw-positivity-app} and its frozen
counterpart~\eqref{eq:frozen-positivity} guarantee that the negative log
scores entering the sequential entropy estimate
$-\log Q^{\mathrm{CTW}}_{L}(z_{1:n})$ and the frozen cross-entropy estimate
$-\log\widehat P^{\mathrm{CTW}}_{x_{1:n},L}(z_{1:n})$ are finite on every
sample path.

\section{Deterministic Auxiliary Lemmas}
\label{app:aux}

This appendix collects the deterministic
ingredients of the analysis: the proof of the local KL--FIM expansion
(Lemma~\ref{lem:kl-fim}, stated in
Section~\ref{subsec:kl-fim-equivalence}), and the two recovery
lemmas, on the continuity of the least-squares map and the
non-expansiveness of the PSD projection. They are invoked by the
consistency proof (Appendix~\ref{app:cons-proof}) and by the
finite-sample analysis (Appendix~\ref{app:rate-fim}).

\begin{proof}[Proof of Lemma~\ref{lem:kl-fim}]
We work first with the marginal laws
$P_{\theta_0}^{(n)}$ and $P_{\theta_0+\Delta}^{(n)}$, whose pmfs are
$p_{\theta_0}^{(n)}$ and $p_{\theta_0+\Delta}^{(n)}$, and normalize by $n$
at the end.

\emph{Step 1 (definition).}
By the definition of KL divergence,
\begin{equation}
D\bigl(P_{\theta_0}^{(n)}\,\|\,P_{\theta_0+\Delta}^{(n)}\bigr)
= \mathbb{E}_{Y_{1:n}\sim P_{\theta_0}^{(n)}}
\!\left[\log p_{\theta_0}^{(n)}(Y_{1:n})
       -\log p_{\theta_0+\Delta}^{(n)}(Y_{1:n})\right].
\label{eq:klfim-def}
\end{equation}

\emph{Step 2 (third-order Taylor expansion).}
Fix a realization $y_{1:n}$. By Assumption~\ref{ass:reg}(ii), the map
$\theta\mapsto\log p_\theta^{(n)}(y_{1:n})$ is three times continuously
differentiable on the open neighborhood $\mathcal N(\theta_0)$.
Since $\theta_0$ is an interior point of this neighborhood, choose
$\varepsilon_0>0$ such that the compact set
\[
\mathcal K
:=\overline{B(\theta_0,\varepsilon_0)}
\subset\mathcal N(\theta_0).
\]
For every $\|\Delta\|_2\le\varepsilon_0$, the segment
$\{\theta_0+t\Delta:t\in[0,1]\}$ is contained in $\mathcal K$.
Taylor's theorem with Lagrange remainder therefore gives
\begin{align}
\log p_{\theta_0+\Delta}^{(n)}(y_{1:n})
&= \log p_{\theta_0}^{(n)}(y_{1:n})
 + \Delta^\top g_n(y_{1:n};\theta_0) %
 + \tfrac12\,\Delta^\top H_n(y_{1:n};\theta_0)\,\Delta
 + R_3(y_{1:n};\theta_0,\Delta),
\label{eq:klfim-taylor}
\end{align}
where $g_n(\cdot;\theta_0)
:=\nabla_\theta\log p_\theta^{(n)}(\cdot)|_{\theta=\theta_0}$ is the score,
$H_n(\cdot;\theta_0)
:=\nabla_\theta^2\log p_\theta^{(n)}(\cdot)|_{\theta=\theta_0}$ is the Hessian, and the
remainder satisfies
\begin{equation}
|R_3(y_{1:n};\theta_0,\Delta)|
\le \frac{1}{6}
\sup_{t\in[0,1]}
\bigl\|\nabla_\theta^{\otimes3}
\log p_{\theta_0+t\Delta}^{(n)}(y_{1:n})\bigr\|
\;\|\Delta\|_2^3 ,
\label{eq:klfim-remainder}
\end{equation}
with $\nabla_\theta^{\otimes3}$ the third-derivative tensor and $\|\cdot\|$ its operator norm.

\emph{Step 3 (take expectations; score and Hessian terms).}
Substitute~\eqref{eq:klfim-taylor} into~\eqref{eq:klfim-def}. The
$\log p_{\theta_0}^{(n)}(y_{1:n})$ terms cancel, leaving
\begin{align}
D\bigl(P_{\theta_0}^{(n)}\,\|\,P_{\theta_0+\Delta}^{(n)}\bigr)
&= -\,\Delta^\top \mathbb{E}_{\theta_0}[g_n]
 - \tfrac12\,\Delta^\top \mathbb{E}_{\theta_0}[H_n]\,\Delta
 - \mathbb{E}_{\theta_0}[R_3].
\label{eq:klfim-three}
\end{align}
Since the alphabet is finite, differentiation commutes
with the finite-sum expectations; the score and Fisher information
identities of Assumption~\ref{ass:reg}(ii) give
\[
\mathbb E_{\theta_0}[g_n]=0,
\qquad
-\mathbb E_{\theta_0}[H_n]=J_{\theta_0}^{(n)}.
\]
Consequently,~\eqref{eq:klfim-three} becomes
\begin{equation}
D\bigl(P_{\theta_0}^{(n)}\,\|\,P_{\theta_0+\Delta}^{(n)}\bigr)
= \tfrac12\,\Delta^\top J_{\theta_0}^{(n)}\Delta
 - \mathbb{E}_{\theta_0}[R_3].
\label{eq:klfim-quad}
\end{equation}

\emph{Step 4 (finite-block remainder bound).}
We first control the Taylor remainder before taking the per-symbol
limit. For $n\ge L^\star$, Assumption~\ref{ass:source} gives the
additive log-pmf decomposition
\begin{equation}
\log p_\theta^{(n)}(y_{1:n})
= \log\pi_{y_{1:L^\star}}(\theta)
 + \sum_{t=L^\star+1}^{n}\log T_{s_t,y_t}(\theta),
\qquad s_t:=y_{t-L^\star:t-1},
\label{eq:logli-decom}
\end{equation}
where
\[
\pi_s(\theta):=p_\theta^{(L^\star)}(s),
\qquad
T_{s,a}(\theta)
:=\Pr_\theta(Y_t=a\mid Y_{t-L^\star:t-1}=s).
\]
By stationarity and the Markov property,
\[
T_{s,a}(\theta)
=\frac{p_\theta^{(L^\star+1)}(s,a)}
       {p_\theta^{(L^\star)}(s)},
\qquad
\log T_{s,a}(\theta)
=\log p_\theta^{(L^\star+1)}(s,a)
 -\log p_\theta^{(L^\star)}(s).
\]
For $L^\star=0$, we use the conventions
$p_\theta^{(0)}(\emptyset)=\pi_{\emptyset}(\theta)=1$.
Assumption~\ref{ass:reg}(ii) therefore implies that
$\log\pi_s(\theta)$ and $\log T_{s,a}(\theta)$ are $C^3$ on
$\mathcal K$. Since $\mathcal K$ is compact and the sets of contexts
and symbols are finite, the constants
\[
M_\pi
:=\max_{s\in\mathcal A^{L^\star}}
\sup_{\theta\in\mathcal K}
\bigl\|\nabla_\theta^{\otimes3}\log\pi_s(\theta)\bigr\|,
\qquad
M_T
:=\max_{\substack{s\in\mathcal A^{L^\star}\\a\in\mathcal A}}
\sup_{\theta\in\mathcal K}
\bigl\|\nabla_\theta^{\otimes3}\log T_{s,a}(\theta)\bigr\|
\]
are finite. Differentiating~\eqref{eq:logli-decom} three times gives,
uniformly over $y_{1:n}$ and $\theta\in\mathcal K$,
\[
\bigl\|\nabla_\theta^{\otimes3}
\log p_\theta^{(n)}(y_{1:n})\bigr\|
\le M_\pi+(n-L^\star)M_T
\le M_\pi+nM_T.
\]
Combining this bound with~\eqref{eq:klfim-remainder} yields
\[
\frac1n\bigl|\mathbb E_{\theta_0}[R_3]\bigr|
\le
\frac16\left(M_T+\frac{M_\pi}{n}\right)\|\Delta\|_2^3.
\]

\emph{Step 5 (per-symbol limit).}
Divide~\eqref{eq:klfim-quad} by $n$. By the definitions of the
normalized block KL divergence, the KL divergence rate, and the
Fisher information rate, whose limits exist for the stationary
finite-state Markov family under
Assumptions~\ref{ass:source}--\ref{ass:reg},
\[
\bar D_n\bigl(P_{\theta_0}^{(n)}\,\|\,P_{\theta_0+\Delta}^{(n)}\bigr)
\longrightarrow
\bar D\bigl(P_{\theta_0}\,\|\,P_{\theta_0+\Delta}\bigr),
\qquad
\frac1nJ_{\theta_0}^{(n)}
\longrightarrow I(\theta_0).
\]
Indeed, differentiating the additive
decomposition~\eqref{eq:logli-decom} gives
\begin{equation}
\nabla_\theta\log p_\theta^{(n)}(Y_{1:n})
=\nabla_\theta\log\pi_{Y_{1:L^\star}}(\theta)
+\sum_{t=L^\star+1}^{n}\nabla_\theta\log T_\theta(Y_t\mid S_t),
\label{eq:score-decomp}
\end{equation}
and each summand is a martingale difference: since
$\sum_a T_\theta(a\mid s)=1$ for every $s$, differentiation of this
identity gives
$\mathbb E_\theta[\nabla_\theta\log T_\theta(Y_t\mid S_t)\mid\mathcal F_{t-1}]
=\sum_a\nabla_\theta T_\theta(a\mid S_t)=0$. The cross terms in the
second moment of~\eqref{eq:score-decomp} therefore vanish, and only the
diagonal terms and the boundary term remain. Dividing by $n$ and using
stationarity,
\begin{equation}
\frac1n J_{\theta_0}^{(n)}
\longrightarrow
I(\theta_0)
=\sum_{s\in\mathcal A^{L^\star}}\pi_{\theta_0}(s)
\sum_{a\in\mathcal A}T_{\theta_0}(a\mid s)\,
\nabla_\theta\log T_{\theta_0}(a\mid s)\,
\nabla_\theta\log T_{\theta_0}(a\mid s)^\top,
\label{eq:fim-stationary}
\end{equation}
the boundary contribution being $O(1/n)$. The same decomposition
applied to the log-likelihood ratio gives the first limit, with
\begin{equation}
\bar D\bigl(P_{\theta_0}\,\|\,P_{\theta}\bigr)
=\sum_{s\in\mathcal A^{L^\star}}\pi_{\theta_0}(s)
\sum_{a\in\mathcal A}T_{\theta_0}(a\mid s)\,
\log\frac{T_{\theta_0}(a\mid s)}{T_{\theta}(a\mid s)}
\label{eq:kl-rate-closed}
\end{equation}
for every $\theta\in\mathcal N(\theta_0)$.
It follows from~\eqref{eq:klfim-quad} that
$-n^{-1}\mathbb E_{\theta_0}[R_3]$ also has a limit. Define
\[
r(\Delta)
:=
-\lim_{n\to\infty}
\frac1n\mathbb E_{\theta_0}[R_3].
\]
Taking $n\to\infty$ in the normalized form of
\eqref{eq:klfim-quad} and using the bound from Step 4 gives
\begin{equation}
\bar D\bigl(P_{\theta_0}\,\|\,P_{\theta_0+\Delta}\bigr)
=\tfrac12\,\Delta^\top I(\theta_0)\,\Delta+r(\Delta),
\qquad
|r(\Delta)|
\le \frac{M_T}{6}\,\|\Delta\|_2^3.
\label{eq:klfim-persymbol}
\end{equation}
Thus~\eqref{eq:kl-fim-expansion} holds with $C_3:=M_T/6$.
\end{proof}

The central-difference measurements of
Section~\ref{subsec:ls-recovery} average the two perturbations
$\pm\varepsilon u$, which cancels the odd-order terms of the expansion
and improves the bias by one order. The next lemma makes this precise.
Recall from~\eqref{eq:kl-rate-closed} that the divergence rate is a
finite sum in which only the scoring law depends on the perturbation.

\begin{lemma}[Symmetric KL expansion]
\label{lem:sym-expansion}
Suppose Assumptions~\ref{ass:source}--\ref{ass:reg}
hold and that the maps $\theta\mapsto T_\theta(a\mid s)$ are four
times continuously differentiable on $\mathcal N(\theta_0)$ with
derivatives bounded uniformly in $(s,a)$. For a unit vector $u$ set
$F_u(\varepsilon):=\bar D(P_{\theta_0}\,\|\,P_{\theta_0+\varepsilon u})$.
Then $F_u$ is four times continuously differentiable on
$[-\varepsilon_0,\varepsilon_0]$, with
\begin{equation}
F_u(0)=F_u'(0)=0,
\qquad
F_u''(0)=u^\top I(\theta_0)\,u,
\label{eq:Fu-derivatives}
\end{equation}
and $\sup_{\|u\|_2=1}\sup_{|\varepsilon|\le\varepsilon_0}|F_u^{(4)}(\varepsilon)|<\infty$.
Consequently there is a constant $C_4<\infty$, independent of $u$ and
$\varepsilon$, such that
\begin{equation}
\bigl|F_u(\varepsilon)+F_u(-\varepsilon)
-\varepsilon^2\,u^\top I(\theta_0)\,u\bigr|
\le C_4\,\varepsilon^4 .
\label{eq:sym-expansion}
\end{equation}
\end{lemma}

\begin{proof}
By~\eqref{eq:kl-rate-closed},
\begin{equation}
F_u(\varepsilon)
=\sum_{s}\pi_{\theta_0}(s)\sum_a T_{\theta_0}(a\mid s)\,
\log\frac{T_{\theta_0}(a\mid s)}{T_{\theta_0+\varepsilon u}(a\mid s)},
\label{eq:Fu-closed}
\end{equation}
a finite sum in which the weights $\pi_{\theta_0}(s)T_{\theta_0}(a\mid s)$
do not depend on $\varepsilon$. Each summand is a composition of
$\log$ with $\varepsilon\mapsto T_{\theta_0+\varepsilon u}(a\mid s)$,
which is four times continuously differentiable by assumption and
bounded below by $\delta$ (Assumption~\ref{ass:reg}(i)); the alphabet
and the context set being finite, $F_u$ is four times continuously
differentiable and its fourth derivative is bounded uniformly in $u$
and $\varepsilon$ by the assumed bounds on the derivatives and $\delta$.

For the values at $\varepsilon=0$, write
$\dot T(a\mid s):=u^\top\nabla_\theta T_{\theta_0}(a\mid s)$ and
$\ddot T(a\mid s):=u^\top\nabla_\theta^2T_{\theta_0}(a\mid s)\,u$.
Differentiating~\eqref{eq:Fu-closed} once gives
\[
F_u'(\varepsilon)
=-\sum_{s}\pi_{\theta_0}(s)\sum_a T_{\theta_0}(a\mid s)\,
\frac{\partial_\varepsilon T_{\theta_0+\varepsilon u}(a\mid s)}
{T_{\theta_0+\varepsilon u}(a\mid s)} .
\]
At $\varepsilon=0$ the ratio reduces to
$\dot T(a\mid s)/T_{\theta_0}(a\mid s)$, the factor
$T_{\theta_0}(a\mid s)$ cancels, and
$\sum_a\dot T(a\mid s)=u^\top\nabla_\theta\sum_a T_{\theta_0}(a\mid s)=0$
because the transition probabilities sum to one. Hence
$F_u'(0)=0$, and $F_u(0)=0$ is immediate. Differentiating once more
and evaluating at $\varepsilon=0$,
\[
F_u''(0)
=\sum_{s}\pi_{\theta_0}(s)\sum_a T_{\theta_0}(a\mid s)
\left[
\left(\frac{\dot T(a\mid s)}{T_{\theta_0}(a\mid s)}\right)^{\!2}
-\frac{\ddot T(a\mid s)}{T_{\theta_0}(a\mid s)}
\right],
\]
and the second group of terms vanishes for the same reason,
$\sum_a\ddot T(a\mid s)=0$. Since
$\dot T(a\mid s)/T_{\theta_0}(a\mid s)=u^\top\nabla_\theta\log T_{\theta_0}(a\mid s)$,
comparison with~\eqref{eq:fim-stationary} gives
$F_u''(0)=u^\top I(\theta_0)u$, which
proves~\eqref{eq:Fu-derivatives}.

Finally, Taylor's theorem with Lagrange remainder at
order four gives
\[
F_u(\pm\varepsilon)
=\pm F_u'(0)\,\varepsilon+\tfrac12F_u''(0)\,\varepsilon^2
\pm\tfrac16F_u'''(0)\,\varepsilon^3
+\tfrac1{24}F_u^{(4)}(\xi_\pm)\,\varepsilon^4
\]
for some $\xi_\pm$ between $0$ and $\pm\varepsilon$. Adding the two
expansions cancels the odd-order terms and
leaves~\eqref{eq:sym-expansion} with
$C_4:=\tfrac1{12}\sup_{\|u\|_2=1}\sup_{|\varepsilon|\le\varepsilon_0}|F_u^{(4)}(\varepsilon)|$.
\end{proof}

The next lemma is deterministic, and it states that the OLS map is
Lipschitz in its inputs.

\begin{lemma}[Continuity of OLS recovery]
\label{lem:ls-continuity}
Let $U\in\mathbb{R}^{m\times q}$ have full column rank $q$. For any
$v,\widetilde v\in\mathbb{R}^m$, define the OLS solutions
\[
\widehat f:=(U^\top U)^{-1}U^\top v,
\qquad
\widetilde f:=(U^\top U)^{-1}U^\top\widetilde v.
\]
Then
\begin{equation}
\bigl\|\widehat f-\widetilde f\bigr\|_2
\le \kappa_U\,\bigl\|v-\widetilde v\bigr\|_2,
\qquad
\kappa_U
:=\bigl\|(U^\top U)^{-1}U^\top\bigr\|_{\mathrm{op}}
=\frac{1}{\sigma_{\min}(U)},
\label{eq:ls-lipschitz}
\end{equation}
where $\|\cdot\|_{\mathrm{op}}$ is the operator norm induced by the
$2$-norm and $\sigma_{\min}(U)$ is the smallest singular value of
$U$. Moreover,
\begin{equation}
\left\|
\operatorname{unvec}_\triangle(\widehat f)
-\operatorname{unvec}_\triangle(\widetilde f)
\right\|_F
\le
\sqrt{2}\,\kappa_U\,
\|v-\widetilde v\|_2.
\label{eq:ls-matrix-lipschitz}
\end{equation}
\end{lemma}

\begin{proof}
The OLS map is linear with matrix
$U^\dagger=(U^\top U)^{-1}U^\top$, the Moore--Penrose
pseudo-inverse of $U$ in the full-column-rank case. Hence
\[
\|\widehat f-\widetilde f\|_2
=\bigl\|U^\dagger(v-\widetilde v)\bigr\|_2
\le\|U^\dagger\|_{\mathrm{op}}\|v-\widetilde v\|_2
=\kappa_U\|v-\widetilde v\|_2,
\]
which proves~\eqref{eq:ls-lipschitz}. Since $U$ has full column rank,
the singular values of $U^\dagger$ are the reciprocals of those of
$U$~\cite{horn2012matrix}, so
$\kappa_U=1/\sigma_{\min}(U)<\infty$.

Under the vectorization convention~\eqref{eq:f-vector}, every
off-diagonal coordinate appears once in $f$ but twice in the symmetric
matrix $\operatorname{unvec}_\triangle(f)$. Therefore, for every
$h\in\mathbb R^q$, writing
$H:=\operatorname{unvec}_\triangle(h)$,
\[
\|\operatorname{unvec}_\triangle(h)\|_F^2
=\sum_{k=1}^d H_{kk}^2+2\sum_{k<\ell}H_{k\ell}^2
\le 2\|h\|_2^2.
\]
Applying this inequality to
$h=\widehat f-\widetilde f$ and then using
\eqref{eq:ls-lipschitz} proves~\eqref{eq:ls-matrix-lipschitz}.
\end{proof}

The final lemma controls the post-fit projection.

\begin{lemma}[Non-expansiveness of the PSD projection]
\label{lem:psd-proj}
Let
$\mathcal{S}_+^d:=\{X\in\mathbb{R}^{d\times d}_{\mathrm{sym}}\colon
X\succeq0\}$, and let $\Pi_{\succeq0}$ denote the orthogonal
projection onto $\mathcal{S}_+^d$ with respect to the Frobenius inner
product, computed by the eigenvalue clipping~\eqref{eq:eig-clip}. For
any $A,C\in\mathbb{R}^{d\times d}_{\mathrm{sym}}$,
\begin{equation}
\bigl\|\Pi_{\succeq0}(A)-\Pi_{\succeq0}(C)\bigr\|_F
\le \|A-C\|_F.
\label{eq:psd-projection-lipschitz}
\end{equation}
Consequently, for any
$B\in\mathbb{R}^{d\times d}_{\mathrm{sym}}$ with $B\succeq0$,
\begin{equation}
\bigl\|\Pi_{\succeq0}(A)-B\bigr\|_F \le \bigl\|A-B\bigr\|_F .
\label{eq:psd-nonexpansive}
\end{equation}
\end{lemma}

\begin{proof}
For a symmetric matrix $A=V\Lambda V^\top$, minimizing
$\|X-A\|_F$ over $X\succeq0$ amounts to replacing each negative
eigenvalue of $A$ by zero, so the eigenvalue
clipping~\eqref{eq:eig-clip} is precisely the metric projection onto
$\mathcal S_+^d$.
The cone $\mathcal{S}_+^d$ is a closed convex subset of the Hilbert
space
$(\mathbb{R}^{d\times d}_{\mathrm{sym}},
\langle\cdot,\cdot\rangle_F)$. The metric projection onto a closed
convex set in a Hilbert space is firmly non-expansive, and therefore
$1$-Lipschitz~\cite[Prop.~4.16]{bauschke2017convex}. This gives
\eqref{eq:psd-projection-lipschitz}. If $B\succeq0$, then
$B\in\mathcal{S}_+^d$ and $\Pi_{\succeq0}(B)=B$. Taking $C=B$ in
\eqref{eq:psd-projection-lipschitz} gives
\eqref{eq:psd-nonexpansive}.
\end{proof}

\section{Proof of the Consistency Theorem}
\label{app:cons-proof}

This appendix proves Lemma~\ref{lem:ctw-kl} and
Theorem~\ref{thm:consistency}.

\begin{proof}[Proof of Lemma~\ref{lem:ctw-kl}]
By Assumption~\ref{ass:source}, both $P_{\theta_0}$ and
$P_{\theta_0+\Delta}$ are stationary finite-alphabet Markov sources of
order at most $L^\star$ over $\mathcal{A}$, and their context chains
are irreducible and aperiodic. Assumption~\ref{ass:reg}(i) gives
$\Pr_\vartheta
\bigl(
Y_t=a\mid Y_{t-L^\star:t-1}=s
\bigr)
\in[\delta,1-\delta]$
for every
$\vartheta\in\{\theta_0,\theta_0+\Delta\}$,
$s\in\mathcal A^{L^\star}$, and $a\in\mathcal A$.
Hence, in particular every conditional probability is strictly less
than $1$; together with $\bar{D}<\infty$, this verifies the assumptions
of \cite[Theorem~2]{cai2006universal} for both sources with order
$L^\star$ and CTW working depth $L\ge L^\star$.

We expand the estimator to make the application explicit. Let
$Z_{1:n}\sim P_{\theta_0}^{(n)}$ be the reference path and
$X^{(\Delta)}_{1:n}\sim P_{\theta_0+\Delta}^{(n)}$ be the independent
auxiliary path for the perturbed source. By construction~\eqref{eq:kl-est},

\[
\widehat{\bar{D}}_n
= \widehat{\bar{H}}_n(P_{\theta_0},P_{\theta_0+\Delta})
- \widehat{\bar{H}}_n(P_{\theta_0})
= -\tfrac1n\log
\widehat P^{\mathrm{CTW}}_{X^{(\Delta)}_{1:n},L}(Z_{1:n})
  + \tfrac1n\log
Q^{\mathrm{CTW}}_{L}(Z_{1:n}),
\]

where the entropy term uses the sequential CTW code
$Q^{\mathrm{CTW}}_{L}$ of the reference path, whereas the cross term uses a
frozen CTW scorer built from the independent auxiliary sequence.
The two terms are treated separately.

\emph{Cross term.} The frozen scorer built from an
independent auxiliary sequence and applied to $Z_{1:n}$ is the
cross-term estimator of \cite[Theorem~2]{cai2006universal}
up to the convention on the first $L$ construction
positions, which does not affect the almost-sure limits of the node
frequencies, with
$Z_{1:n}$ playing the role of the sequence to be scored and
$X^{(\Delta)}_{1:n}$ that of the independent auxiliary sequence. That
theorem gives
\[
-\tfrac1n\log
\widehat P^{\mathrm{CTW}}_{X^{(\Delta)}_{1:n},L}(Z_{1:n})
\;\xrightarrow[n\to\infty]{\textnormal{a.s.}}\;
\bar{H}(P_{\theta_0},P_{\theta_0+\Delta}).
\]

\emph{Entropy term.} Here the estimator departs from
\cite{cai2006universal}, which freezes a model built from $Z_{1:n}$ and
then scores $Z_{1:n}$ with it; we use instead the sequential code
$Q^{\mathrm{CTW}}_{L}$, whose universality gives the limit directly.
Write
$R_n:=\log\bigl(p_{\theta_0}^{(n)}(Z_{1:n})/Q^{\mathrm{CTW}}_{L}(Z_{1:n})\bigr)$,
so that
\[
-\tfrac1n\log Q^{\mathrm{CTW}}_{L}(Z_{1:n})
=-\tfrac1n\log p_{\theta_0}^{(n)}(Z_{1:n})+\tfrac1nR_n .
\]
The deterministic pointwise redundancy bound
$R_n\le\gamma\log n+c_0$ of Lemma~\ref{lem:self-red} gives
$\limsup_n\tfrac1nR_n\le0$. In the other direction,
$Q^{\mathrm{CTW}}_{L}$ is a pmf on
$\mathcal A^n$, so
$\mathbb E_{\theta_0}[Q^{\mathrm{CTW}}_{L}(Z_{1:n})/p_{\theta_0}^{(n)}(Z_{1:n})]\le1$
and Markov's inequality gives
$\mathbb P(R_n\le-2\log n)\le n^{-2}$; these probabilities are
summable, so $\liminf_n\tfrac1nR_n\ge0$ almost surely by the
Borel--Cantelli lemma. Hence $\tfrac1nR_n\to0$ almost surely, and the
Shannon--McMillan--Breiman theorem
\cite[Thm.~16.8.1]{cover1999elements} applied to the stationary
ergodic source $P_{\theta_0}$ gives
$-\tfrac1n\log p_{\theta_0}^{(n)}(Z_{1:n})\to\bar H(P_{\theta_0})$
almost surely. Therefore
\[
-\tfrac1n\log
Q^{\mathrm{CTW}}_{L}(Z_{1:n})
\;\xrightarrow[n\to\infty]{\textnormal{a.s.}}\;
\bar{H}(P_{\theta_0}).
\]

Subtracting the two almost-sure limits and using
\eqref{eq:kl-rate-entropy} gives~\eqref{eq:ctw-kl-as}.
\end{proof}

\begin{proof}[Proof of Theorem~\ref{thm:consistency}]
Fix $\varepsilon\in(0,\varepsilon_0]$. For each direction $u_j$, define three scalars:
\[
v_j^\star := u_j^\top I(\theta_0)\,u_j,
\quad
v_j^{(\varepsilon)} := \tfrac{2}{\varepsilon^2}\,
   \bar{D}\bigl(P_{\theta_0}\,\|\,P_{\theta_0+\varepsilon u_j}\bigr),
\quad
v_{n,j}^{(\varepsilon)} := \tfrac{2}{\varepsilon^2}\,
   \widehat{\bar{D}}_n\bigl(P_{\theta_0}\,\|\,P_{\theta_0+\varepsilon u_j}\bigr),
\]
which represent the ground-truth target, its population (Taylor-biased) value, and its sample estimate, respectively. We stack them into vectors $v^\star,v^{(\varepsilon)},v_n^{(\varepsilon)}\in\mathbb{R}^m$.

\emph{Step 1 (Taylor bias is deterministic and $O(\varepsilon)$).}
By Lemma~\ref{lem:kl-fim} with $\Delta=\varepsilon u_j$, $\bar{D}(P_{\theta_0}\,\|\,P_{\theta_0+\varepsilon u_j}) =\tfrac12\varepsilon^2 v_j^\star + r(\varepsilon u_j)$ with $|r(\varepsilon u_j)|\le C_3\varepsilon^3$. Hence, the deterministic truncation bias satisfies
\begin{equation}
\bigl|v_j^{(\varepsilon)}-v_j^\star\bigr|
= \tfrac{2}{\varepsilon^2}\,|r(\varepsilon u_j)|
\le 2C_3\,\varepsilon,
\qquad j=1,\dots,m.
\label{eq:cons-bias}
\end{equation}

\emph{Step 2 (estimation error vanishes a.s.).}
By Lemma~\ref{lem:ctw-kl} applied to the pair $(P_{\theta_0},P_{\theta_0+\varepsilon u_j})$, $\widehat{\bar{D}}_n\to\bar{D}$ a.s., so rescaling by the fixed factor $2/\varepsilon^2$ yields
\begin{equation}
v_{n,j}^{(\varepsilon)}\;\xrightarrow[n\to\infty]{\textnormal{a.s.}}\;v_j^{(\varepsilon)},
\qquad j=1,\dots,m.
\label{eq:cons-est}
\end{equation}
Since the direction index set $\{1,\dots,m\}$ is finite, the intersection of these $m$ probability-one events on which~\eqref{eq:cons-est} holds still has probability one. Thus, we have vector-valued convergence $v_n^{(\varepsilon)}\to v^{(\varepsilon)}$ a.s.\ in $(\mathbb{R}^m,\|\cdot\|_2)$.

\emph{Step 3 (propagate through OLS).}
By the construction of $U$ and the coordinate mapping~\eqref{eq:f-vector}, the target matrix vectorizes exactly to $f^\star$, which satisfies $U f^\star = v^\star$ (and thus $f^\star = (U^\top U)^{-1}U^\top v^\star$). Let $\widehat{f}_n^{(\varepsilon)}:=(U^\top U)^{-1}U^\top v_n^{(\varepsilon)}$ and $f^{(\varepsilon)}:=(U^\top U)^{-1}U^\top v^{(\varepsilon)}$ be the OLS solutions. By the linearity highlighted in Lemma~\ref{lem:ls-continuity} and the a.s. convergence in~\eqref{eq:cons-est}, it follows that
\begin{equation}
\widehat{f}_n^{(\varepsilon)}\;\xrightarrow[n\to\infty]{\textnormal{a.s.}}\;f^{(\varepsilon)}.
\label{eq:cons-f}
\end{equation}
Applying the Lipschitz bound~\eqref{eq:ls-lipschitz} of the OLS recovery from Lemma~\ref{lem:ls-continuity} together with the bias bound~\eqref{eq:cons-bias}, we obtain
\begin{equation}
\|f^{(\varepsilon)}-f^\star\|_2
\le \kappa_U\|v^{(\varepsilon)}-v^\star\|_2
\le \kappa_U\sqrt{m}\,\max_j|v_j^{(\varepsilon)}-v_j^\star|
\le 2\kappa_U\sqrt{m}\,C_3\,\varepsilon.
\label{eq:cons-fbias}
\end{equation}

\emph{Step 4 (reshape and define the bias matrix).}
We now map the vectors back to the symmetric matrix space. Let $\widehat{I}_{n,\varepsilon}^{\mathrm{LS}}:=\operatorname{unvec}_\triangle(\widehat{f}_n^{(\varepsilon)})$, $I^{(\varepsilon)}:=\operatorname{unvec}_\triangle(f^{(\varepsilon)})$, and note that $\operatorname{unvec}_\triangle(f^\star)=I(\theta_0)$. Define the deterministic error matrix as $E_U(\varepsilon):=I^{(\varepsilon)}-I(\theta_0)$, where the subscript records its dependence on the chosen direction design. Applying the matrix Lipschitz bound~\eqref{eq:ls-matrix-lipschitz} and the directional bias bound~\eqref{eq:cons-bias} yields
\[
\|E_U(\varepsilon)\|_F
\le \sqrt{2}\,\kappa_U\|v^{(\varepsilon)}-v^\star\|_2
\le 2\sqrt{2}\kappa_U\sqrt{m}\,C_3\,\varepsilon.
\]
Furthermore, applying $\operatorname{unvec}_\triangle$ to the a.s. convergence~\eqref{eq:cons-f} results in $\widehat{I}_{n,\varepsilon}^{\mathrm{LS}}\to I^{(\varepsilon)}=I(\theta_0)+E_U(\varepsilon)$ a.s.

\emph{Step 5 (apply the projection; prove (i)).}
The final estimator is defined as $\widehat{I}_{n,\varepsilon}(\theta_0) =\Pi_{\succeq0}(\widehat{I}_{n,\varepsilon}^{\mathrm{LS}})$. By the non-expansiveness~\eqref{eq:psd-projection-lipschitz} of Lemma~\ref{lem:psd-proj},
\[
\bigl\|\Pi_{\succeq0}(\widehat{I}_{n,\varepsilon}^{\mathrm{LS}})
-\Pi_{\succeq0}\bigl(I(\theta_0)+E_U(\varepsilon)\bigr)\bigr\|_F
\le\bigl\|\widehat{I}_{n,\varepsilon}^{\mathrm{LS}}
-I(\theta_0)-E_U(\varepsilon)\bigr\|_F
\;\xrightarrow[n\to\infty]{\textnormal{a.s.}}\;0,
\]
which gives
\[
\widehat{I}_{n,\varepsilon}(\theta_0)
\;\xrightarrow[n\to\infty]{\textnormal{a.s.}}\;
\Pi_{\succeq0}\bigl(I(\theta_0)+E_U(\varepsilon)\bigr),
\]
i.e.,~\eqref{eq:thm-step1}.

\emph{Step 6 (take the outer limit; prove (ii)).}
By part (i), the inner limit equals the deterministic
matrix $\Pi_{\succeq0}(I(\theta_0)+E_U(\varepsilon))$ for every
$\varepsilon\in(0,\varepsilon_0]$, so the outer limit is a
deterministic one. The Fisher information rate satisfies
$I(\theta_0)\succeq0$, so $I(\theta_0)\in\mathcal S_+^d$. Applying
Lemma~\ref{lem:psd-proj} with
$A=I(\theta_0)+E_U(\varepsilon)$ and $B=I(\theta_0)$ gives
\[
\bigl\|\Pi_{\succeq0}(I(\theta_0)+E_U(\varepsilon))
-I(\theta_0)\bigr\|_F
\le \|E_U(\varepsilon)\|_F
\le 2\sqrt{2}\kappa_U\sqrt{m}\,C_3\,\varepsilon
\;\xrightarrow[\varepsilon\downarrow0]{}\;0,
\]
which is~\eqref{eq:thm-step2}.
\end{proof}

\section{Proof of the Local KL Rate}
\label{app:rate}

This appendix proves Proposition~\ref{prop:kl-rate} together with all
the auxiliary lemmas it rests on. Throughout, $\theta_0$ is the target
parameter, $Z_{1:n}\sim P_{\theta_0}^{(n)}$ the reference path,
$S_t=Z_{t-L:t-1}$, and constants may depend on
$(|\mathcal A|,L,\delta,\pi_{\min},\kappa)$ and on
the first-derivative bound
$C_T:=\sup_{\|\theta-\theta_0\|_2\le\varepsilon_0}\max_{s,a}
\|\nabla_\theta T_\theta(a\mid s)\|_2$, finite by
Assumption~\ref{ass:reg}(ii) and compactness, but never on $n$ or
$\varepsilon$. The analysis is carried out for a \emph{clipped}
modification of the cross scorer, introduced below as a purely technical
device; the final subsection transfers the bound to the unclipped
estimator of Algorithm~\ref{alg:ctw-kl}, so that no clipping appears in
the statement of Proposition~\ref{prop:kl-rate}. Fix the clipping level
$\alpha:=\delta/4$ once and for all; any fixed value in $(0,\delta/2)$
would do.

\subsection{Estimator objects and error anatomy}
\label{app:rate-anatomy}

Fix a unit direction $u$ and abbreviate the output of
Algorithm~\ref{alg:ctw-kl} run with the perturbation
$\Delta=\varepsilon u$ by
\begin{equation}
\widehat{\bar D}_{n,\varepsilon}(u)
:=\widehat{\bar D}_n\bigl(P_{\theta_0}\,\|\,P_{\theta_0+\varepsilon u}\bigr).
\label{eq:dir-shorthand}
\end{equation}
Let $Z_{1:n}\sim P_{\theta_0}^{(n)}$
be the reference path, $S_t:=Z_{t-L:t-1}\in\mathcal A^L$ the length-$L$
scoring context, and, for $s\in\mathcal A^L$, write
$T_\theta(a\mid s):=T_{s,a}(\theta)$ and $\pi_\theta(s)$ for the
transition probability and the stationary context law of
Section~\ref{sec:problem} (the conditional form is used because the
analysis treats $T_\theta(\cdot\mid s)$ as a distribution over $a$).
The estimator is the difference of two code lengths of \emph{different
nature}, and the distinction drives the analysis: the entropy term uses
the sequential CTW code $Q^{\mathrm{CTW}}_{L}$, an online pmf on
$\mathcal A^n$, while the cross-entropy term uses a
frozen one-step predictor
$\widehat T^{\mathrm{CTW},\alpha}_{\theta_0+\varepsilon u,n}$ built from
an independent auxiliary path
$X_{1:n}\sim P_{\theta_0+\varepsilon u}^{(n)}$ and clipped away from zero
(the clipping and warm-up conventions are collected at the end of this
subsection; they are bookkeeping and play no conceptual role). Write
the prediction error of the clipped cross predictor built from a path
$\sim P_\theta^{(n)}$ as
\begin{equation}
e_{\theta,n}(a\mid s)
:=\widehat T^{\mathrm{CTW},\alpha}_{\theta,n}(a\mid s)-T_\theta(a\mid s),
\qquad
\sum_{a\in\mathcal A}e_{\theta,n}(a\mid s)=0,
\label{eq:pred-error-sumzero}
\end{equation}
the normalization identity holding because both are probability vectors
on $\mathcal A$.

Let $\widehat{\bar D}^{\,\alpha}_{n,\varepsilon}(u)$ denote the clipped
variant of the estimator, made precise in~\eqref{eq:dir-obj} below.
Inserting the true log-pmfs
$\pm\tfrac1n\log p_{\theta_0}^{(n)}(Z_{1:n})$ and
$\pm\tfrac1n\log p_{\theta_0+\varepsilon u}^{(n)}(Z_{1:n})$ into the
estimator and expanding them by the Markov chain rule ($L\ge L^\star$)
splits the estimation error into three terms of entirely different
character,
\begin{align}
\widehat{\bar D}^{\,\alpha}_{n,\varepsilon}(u)
-\bar D(P_{\theta_0}\,\|\,P_{\theta_0+\varepsilon u})
&=\underbrace{\frac{1}{n}\sum_{t=L+1}^{n}
\log\frac{T_{\theta_0}(Z_t\mid S_t)}{T_{\theta_0+\varepsilon u}(Z_t\mid S_t)}
-\bar D}_{A_n\ \text{(oracle)}}
\nonumber\\
&\quad+\underbrace{\tfrac1n\log Q^{\mathrm{CTW}}_{L}(Z_{1:n})-\tfrac1n\log p_{\theta_0}^{(n)}(Z_{1:n})}_{B^{\mathrm{self}}_n\ \text{(self)}}
\nonumber\\
&\quad-\underbrace{\frac{1}{n}\sum_{t=L+1}^{n}
\log\frac{\widehat T^{\mathrm{CTW},\alpha}_{\theta_0+\varepsilon u,n}(Z_t\mid S_t)}{T_{\theta_0+\varepsilon u}(Z_t\mid S_t)}}_{B^{\mathrm{cross}}_n\ \text{(cross)}}
+R_{\mathrm{bd}},
\label{eq:master-decomp}
\end{align}
where $R_{\mathrm{bd}}$ is a boundary remainder obeying
a pathwise bound of size $O(\log n/n)$, quantified below. Each term is small for a different
reason:
\begin{itemize}
\item the \emph{oracle} term $A_n$ is the empirical average of
the log-ratio of two \emph{true} models that differ by $O(\varepsilon)$.
By stationarity each summand has mean $\bar D$, so the
average over the $n-L$ scoring positions is exactly centred. The term
$A_n$ divides that sum by $n$ rather than by $n-L$, which shifts it by a
deterministic $O(L\varepsilon/n)$.
Its integrand, and hence (by mixing) its fluctuation, carries an
$\varepsilon$ factor: $\mathbb E|A_n|=O(\varepsilon/\sqrt n)$
(Lemma~\ref{lem:oracle}).
\item the \emph{self} term $B^{\mathrm{self}}_n$ is the coding
\emph{redundancy} of the sequential CTW code, controlled by a
deterministic pointwise bound together with probability conservation,
not by the central limit theorem:
$\mathbb E|B^{\mathrm{self}}_n|=O(\log n/n)$
(Lemma~\ref{lem:self-red}).
\item the \emph{cross} term $B^{\mathrm{cross}}_n$ is the log-loss of the
frozen predictor; the normalization
identity~\eqref{eq:pred-error-sumzero} kills its first-order part up to
the model gap $O(\varepsilon)$, leaving
$\mathbb E|B^{\mathrm{cross}}_n|=O(\varepsilon/\sqrt n+\log n/n)$
(Lemma~\ref{lem:cross-pred}).
\end{itemize}
The sum is $O(\varepsilon/\sqrt n+\log n/n)$
(Proposition~\ref{prop:kl-rate}): an estimation error that \emph{scales
with} $\varepsilon$, which is exactly what the $\varepsilon^{-2}$
division requires.

\paragraph{Technical conventions.}
The following conventions make the objects in~\eqref{eq:master-decomp}
precise; they are bookkeeping and contain no ideas. For a probability
vector $q$ on $\mathcal A$ and a level $\alpha\in(0,1/|\mathcal A|)$
define the renormalized clip
\begin{equation}
\operatorname{clip}_\alpha(q)(a)
:=\frac{q(a)\vee\alpha}{\sum_{b\in\mathcal A}(q(b)\vee\alpha)},
\qquad
\operatorname{clip}_\alpha(q)(a)\ge\alpha':=\frac{\alpha}{|\mathcal A|}>0,
\label{eq:clip}
\end{equation}
the lower bound following from $q(a)\vee\alpha\ge\alpha$ in the numerator
and $\sum_b(q(b)\vee\alpha)\le|\mathcal A|$ (each summand is $\le1$) in the
denominator. The clipped frozen one-step predictor built from an
independent path $\sim P_\theta^{(n)}$ is
$\widehat T^{\mathrm{CTW},\alpha}_{\theta,n}(\cdot\mid s)
:=\operatorname{clip}_\alpha\bigl(\widehat T^{\mathrm{CTW}}_{\theta,n}(\cdot\mid s)\bigr)$,
with block scorer
$\widehat P^{\mathrm{CTW},\alpha}_{\theta,n}(z_{1:n})
=C^\alpha(z_{1:L})\prod_{t=L+1}^n\widehat T^{\mathrm{CTW},\alpha}_{\theta,n}(z_t\mid z_{t-L:t-1})$,
where $C^\alpha(z_{1:L})$ denotes the clipped CTW warm-up probability
assigned to the first $L$ scoring symbols under the standard
shorter-context initialization (for $t\le L$ the active context is the
available prefix $z_{1:t-1}$, whose length is below the working depth).
The analyzed (clipped) directional estimator is
\begin{equation}
\widehat{\bar D}^{\,\alpha}_{n,\varepsilon}(u)
=-\tfrac1n\log\widehat P^{\mathrm{CTW},\alpha}_{\theta_0+\varepsilon u,n}(Z_{1:n})
+\tfrac1n\log Q^{\mathrm{CTW}}_{L}(Z_{1:n}),
\label{eq:dir-obj}
\end{equation}
i.e.\ the output of Algorithm~\ref{alg:ctw-kl} with the cross scorer
clipped; the modification is inactive except on an event of probability
$O(n^{-1})$, on which the unclipped frozen predictor is inaccurate
(Step~0 of the proof of Lemma~\ref{lem:pmc} below), and the final
subsection transfers the bound back to the unclipped
Algorithm~\ref{alg:ctw-kl}. Finally, the
remainder in~\eqref{eq:master-decomp} is
\begin{equation}
R_{\mathrm{bd}}
:=\frac1n\log
\frac{p^{(L)}_{\theta_0+\varepsilon u}(Z_{1:L})}{C^\alpha(Z_{1:L})}
+\frac1n\log
\frac{p^{(L)}_{\theta_0}(Z_{1:L})}{p^{(L)}_{\theta_0+\varepsilon u}(Z_{1:L})},
\label{eq:rbd}
\end{equation}
the first term being the warm-up mismatch between the cross scorer and
the perturbed law, and the second the boundary term left by the
chain-rule expansion of the two true log-likelihoods. The self code
contributes no warm-up term, because $B^{\mathrm{self}}_n$ compares
the two block quantities directly. Since every factor is bounded from below,
by $\delta^{L}$ for the true laws and by $(\alpha')^{L}$ for the
clipped scorer, $|R_{\mathrm{bd}}|\le CL/n$ pathwise. The
$\tfrac1n$-versus-$\tfrac1{n-L}$ normalization discrepancies of the
three sums are treated separately in
Appendix~\ref{app:kl-assembly}.

\subsection{Mixing, oracle, and self terms}

Three of the four error sources are quantified by the following lemmas.
The first is the engine of the analysis: a variance bound for additive
functionals of the geometrically mixing context chain, uniform in the
underlying parameter.

\begin{lemma}[Uniform mixing bound]
\label{lem:mixing}
Under Assumptions~\ref{ass:source}--\ref{ass:reg}, there is a constant
$C_{\mathrm{add}}<\infty$, uniform in $\theta\in\mathcal N(\theta_0)$,
such that for every $\theta\in\mathcal N(\theta_0)$ and every
deterministic $f\colon\mathcal A^L\times\mathcal A\to\mathbb R$,
\begin{equation}
\mathbb E_\theta\!\left|
\frac{1}{n-L}\sum_{t=L+1}^n
\bigl(f(S_t,Y_t)-\mathbb E_\theta f(S_t,Y_t)\bigr)
\right|
\le
\frac{C_{\mathrm{add}}}{\sqrt n}
\bigl(\mathbb E_\theta f(S_t,Y_t)^2\bigr)^{1/2},
\label{eq:mixing-bound}
\end{equation}
where $(Y_t)$ is a trajectory of $P_\theta$ and
$S_t=Y_{t-L:t-1}$. The pair $(S_t,Y_t)$ determines $S_{t+1}$ and is
determined by $(S_t,S_{t+1})$, so $f(S_t,Y_t)$ is a functional of two
consecutive states; the bound is applied to the chain of consecutive
pairs, which is again finite-state and inherits a uniform Doeblin
minorization, with $L+1$ steps in place of $L$.
The right-hand factor $(\mathbb E_\theta f^2)^{1/2}$ is what later lets a
small integrand produce a small fluctuation.
\end{lemma}

\begin{proof}
For the context chain $S_t=Y_{t-L:t-1}$ on $\mathcal A^L$, after $L$ steps
the state equals the $L$ freshly emitted symbols, whose conditional
probabilities are each $\ge\delta$ (Assumption~\ref{ass:reg}); hence
$P_\theta^{L}(s,s')\ge\delta^{L}$ for every $s,s'\in\mathcal A^L$,
uniformly in $\theta\in\mathcal N(\theta_0)$. Writing $\nu$ for the
uniform law on $\mathcal A^L$, this is a Doeblin
minorization~\cite{meyn2009markov}
$P_\theta^{L}(s,\cdot)\ge\eta\,\nu(\cdot)$ with a coefficient
$\eta=\eta(\delta,|\mathcal A|,L)=(|\mathcal A|\delta)^{L}\in(0,1]$
(using $\delta\le|\mathcal A|^{-1}$, which holds since $|\mathcal A|$
conditional probabilities that are each $\ge\delta$ sum to one). The chain
is therefore uniformly ergodic with a contraction coefficient $1-\eta$
independent of $\theta$, and mixes geometrically \cite[Sec.~3]{Pau15}.
Consequently the covariance series of any bounded additive functional is
absolutely summable with a constant $C=C(\delta,|\mathcal A|,L)$, and
\[
\operatorname{Var}_\theta\!\Bigl(\tfrac{1}{n-L}\textstyle\sum_{t=L+1}^n f(S_t,Y_t)\Bigr)
\le\frac{C}{n}\,\mathbb E_\theta f(S_t,Y_t)^2 ,
\]
uniformly in $\theta\in\mathcal N(\theta_0)$. Jensen's inequality
$\mathbb E|X|\le(\operatorname{Var}X)^{1/2}$ applied to the centred
average gives~\eqref{eq:mixing-bound}.
\end{proof}

\begin{lemma}[Oracle fluctuation]
\label{lem:oracle}
Under Assumptions~\ref{ass:source}--\ref{ass:reg} there is a constant
$C_{\mathrm{or}}<\infty$ such that
\begin{equation}
\mathbb E\!\left|
\frac{1}{n-L}\sum_{t=L+1}^n
\log\frac{T_{\theta_0}(Z_t\mid S_t)}{T_{\theta_0+\varepsilon u}(Z_t\mid S_t)}
-\bar D\bigl(P_{\theta_0}\,\|\,P_{\theta_0+\varepsilon u}\bigr)
\right|
\le C_{\mathrm{or}}\,\frac{\varepsilon}{\sqrt n}.
\label{eq:oracle-rate}
\end{equation}
\end{lemma}

\begin{proof}
Set $\ell_\varepsilon(s,a):=\log\frac{T_{\theta_0}(a\mid s)}{T_{\theta_0+\varepsilon u}(a\mid s)}$.
Since $T_\theta(a\mid s)\ge\delta$ and, by
Assumption~\ref{ass:reg},
$|T_{\theta_0+\varepsilon u}(a\mid s)-T_{\theta_0}(a\mid s)|\le C_T\varepsilon$,
the mean-value form of $\log$ gives
$|\ell_\varepsilon(s,a)|\le\delta^{-1}C_T\varepsilon$, hence
$(\mathbb E_{\theta_0}\ell_\varepsilon^2)^{1/2}\le C\varepsilon$. Because
$\mathbb E_{\theta_0}\ell_\varepsilon(S_t,Z_t)
=\sum_s\pi_{\theta_0}(s)\sum_a T_{\theta_0}(a\mid s)\ell_\varepsilon(s,a)
=\bar D(P_{\theta_0}\|P_{\theta_0+\varepsilon u})$, the left-hand side
of~\eqref{eq:oracle-rate} is the centred average of
$f=\ell_\varepsilon$; Lemma~\ref{lem:mixing} at $\theta=\theta_0$ bounds
it by $C_{\mathrm{add}}\varepsilon/\sqrt n$.
\end{proof}

The self term uses the online entropy code $Q^{\mathrm{CTW}}_{L}$. Since
$Q^{\mathrm{CTW}}_{L}$ is built from the reference path itself, it is
\emph{not} independent of that path, so the conditioning argument used for
the cross term does not apply; instead the term is the deterministic
coding \emph{redundancy} of $Q^{\mathrm{CTW}}_{L}$ relative to the true
law, controlled by the CTW redundancy bound.

\begin{lemma}[In-sample self-redundancy]
\label{lem:self-red}
Under Assumptions~\ref{ass:source}--\ref{ass:reg} with $L\ge L^\star$,
\begin{equation}
\mathbb E\!\left|
\tfrac1n\log Q^{\mathrm{CTW}}_{L}(Z_{1:n})
-\tfrac1n\log p_{\theta_0}^{(n)}(Z_{1:n})
\right|
\le C_{\mathrm{self}}\,\frac{\log n}{n}.
\label{eq:self-red-rate}
\end{equation}
\end{lemma}

\begin{proof}
The self term is
\[
B^{\mathrm{self}}_n
=\tfrac1n\log Q^{\mathrm{CTW}}_{L}(Z_{1:n})-\tfrac1n\log p_{\theta_0}^{(n)}(Z_{1:n})
=-\tfrac1n R_n,
\qquad
R_n:=\log\frac{p_{\theta_0}^{(n)}(Z_{1:n})}{Q^{\mathrm{CTW}}_{L}(Z_{1:n})},
\]
where $R_n$ is the individual redundancy of the online CTW code. Two facts
bound $\mathbb E|R_n|$.
\emph{(i)} For a source of order $\le L^\star\le L$ the online CTW code
satisfies the deterministic pointwise redundancy bound
$R_n\le\gamma\log n+c_0$, with $\gamma,c_0$ depending only on
$|\mathcal A|$ and $L$. Indeed, the CTW code is a
Bayesian mixture over the tree models of depth at most $L$, so it
dominates any single term of that mixture,
$Q^{\mathrm{CTW}}_{L}(z_{1:n})\ge\pi(T^\star)\prod_{u\in T^\star}P_e^u$,
where $\pi(T^\star)$ is the prior weight of the minimal tree and
$-\log\pi(T^\star)$ is a constant for fixed $L$. Bounding each leaf
factor by the $|\mathcal A|$-ary Krichevsky--Trofimov redundancy bound
$-\log P_e^u\le n_u\widehat H(Y\mid u)+\tfrac{|\mathcal A|-1}{2}\log(n_u+1)+O(1)$
\cite[Lemma~4.2]{Kon24}, and summing over the at most
$|\mathcal A|^{L}$ leaves, bounds $-\log Q^{\mathrm{CTW}}_{L}(z_{1:n})$
by $\sum_{u\in T^\star}n_u\widehat H(Y\mid u)+\gamma\log n+c_0'$. In
the opposite direction, the true log-likelihood factorizes over the
same leaves and the empirical frequencies maximize it, so
$\log p_{\theta_0}^{(n)}(z_{1:n})\le-\sum_{u\in T^\star}n_u\widehat H(Y\mid u)$
up to the boundary term of the first $L^\star$ symbols, which is at
most $L^\star\log(1/\delta)$ by Assumption~\ref{ass:reg}(i). Adding the
two bounds gives the stated redundancy bound.
\emph{(ii)} The online CTW code is a genuine pmf on
$\mathcal A^n$ ($\sum_{z}Q^{\mathrm{CTW}}_{L}(z)=1$), so
$\mathbb E R_n=D\bigl(P_{\theta_0}^{(n)}\,\|\,Q^{\mathrm{CTW}}_{L}\bigr)\ge0$.
From (i), $\mathbb E R_n^+\le\gamma\log n+c_0$; combined with
$\mathbb E R_n\ge0$ this gives $\mathbb E R_n^-\le\mathbb E R_n^+$, hence
$\mathbb E|R_n|=\mathbb E R_n^++\mathbb E R_n^-\le2(\gamma\log n+c_0)$.
Therefore $\mathbb E|B^{\mathrm{self}}_n|\le 2(\gamma\log n+c_0)/n=O(\log n/n)$.
(All warm-up and normalization bookkeeping sits in the remainder
$R_{\mathrm{bd}}$ of~\eqref{eq:master-decomp}, not in this lemma:
$B^{\mathrm{self}}_n$ there is exactly the block quantity bounded here.)
\end{proof}

\subsection{The prediction-moment lemma and the cross term}
\label{app:rate-pmc}

The cross term is controlled through a single quantitative property of
the frozen CTW predictor, stated once and reused: the first and second
moments of its prediction error, averaged over the scoring contexts,
decay at the parametric rate. This is the only place where the
structural
Assumptions~\ref{ass:topology}--\ref{ass:gap} enter.

\begin{lemma}[Prediction-moment lemma]
\label{lem:pmc}
Suppose Assumptions~\ref{ass:source}--\ref{ass:reg} and
\ref{ass:topology}--\ref{ass:gap} hold with $L\ge L^\star$. The clipped
frozen CTW predictors
$\{\widehat T^{\mathrm{CTW},\alpha}_{\theta,n}\}_{\theta\in\mathcal N(\theta_0)}$
of~\eqref{eq:clip}, each built from an independent path
$\sim P_\theta^{(n)}$ (so that
$\widehat T^{\mathrm{CTW},\alpha}_{\theta,n}(a\mid s)\ge\alpha'$),
satisfy, with $e_{\theta,n}$ as in~\eqref{eq:pred-error-sumzero},
uniformly in $\theta\in\mathcal N(\theta_0)$,
\begin{align}
\mathbb E\sum_{s\in\mathcal A^L}\pi_{\theta_0}(s)\,
\bigl\|e_{\theta,n}(\cdot\mid s)\bigr\|_1
&\le C_1\,n^{-1/2},
\label{eq:pmc-l1}\\
\mathbb E\sum_{s\in\mathcal A^L}\pi_{\theta_0}(s)\,
\bigl\|e_{\theta,n}(\cdot\mid s)\bigr\|_2^2
&\le C_2\,\frac{\log n}{n},
\label{eq:pmc-l2}
\end{align}
with $C_1,C_2$ depending only on
$(|\mathcal A|,L,\delta,\alpha,\pi_{\min},\kappa,C_{\mathrm{add}})$. In
fact~\eqref{eq:pmc-l2} holds in the sharper form
$\mathbb E\sum_s\pi_{\theta_0}(s)\|e_{\theta,n}(\cdot\mid s)\|_2^2\le C_2\,n^{-1}$.
\end{lemma}

The proof is given in the next subsection. The cross-term lemma below
uses only the two displayed moment bounds, never the CTW structure
itself: any frozen one-step predictor
satisfying~\eqref{eq:pmc-l1}--\eqref{eq:pmc-l2} supports the remainder
of the argument verbatim.

The cross term is the one that carries the $\varepsilon$ factor,
obtained from the cancellation~\eqref{eq:pred-error-sumzero}.

\begin{lemma}[Frozen cross-predictor log-loss]
\label{lem:cross-pred}
Let $\theta\in\mathcal N(\theta_0)$ and let the clipped frozen predictor
$\widehat T^{\mathrm{CTW},\alpha}_{\theta,n}$ satisfy the prediction-moment
bounds~\eqref{eq:pmc-l1}--\eqref{eq:pmc-l2}. With
$\Delta_\theta:=\max_{s,a}\lvert T_{\theta_0}(a\mid s)-T_\theta(a\mid s)\rvert$,
\begin{equation}
\mathbb E\!\left|
\frac{1}{n-L}\sum_{t=L+1}^n
\log\frac{\widehat T^{\mathrm{CTW},\alpha}_{\theta,n}(Z_t\mid S_t)}{T_\theta(Z_t\mid S_t)}
\right|
\le C\!\left(\Delta_\theta\,n^{-1/2}+\frac{\log n}{n}\right).
\label{eq:cross-pred-rate}
\end{equation}
In particular, for $\theta=\theta_0+\varepsilon u$ one has
$\Delta_\theta\le C_T\varepsilon$ (Assumption~\ref{ass:reg}), giving the
bound $C(\varepsilon/\sqrt n+\log n/n)$.
\end{lemma}

The $\Delta_\theta$ factor is produced by a two-line cancellation that is
the crux of the whole rate. Conditionally on the construction path, the
first-order part of the mean cross log-loss is
$\sum_s\pi_{\theta_0}(s)\sum_a T_{\theta_0}(a\mid s)\,
e_{\theta,n}(a\mid s)/T_\theta(a\mid s)$; inserting
$T_{\theta_0}=T_\theta+(T_{\theta_0}-T_\theta)$ in the numerator,
\begin{equation}
\sum_a T_{\theta_0}(a\mid s)\,\frac{e_{\theta,n}(a\mid s)}{T_\theta(a\mid s)}
=\underbrace{\sum_a e_{\theta,n}(a\mid s)}_{=\,0
\ \text{by}~\eqref{eq:pred-error-sumzero}}
+\sum_a\bigl(T_{\theta_0}(a\mid s)-T_\theta(a\mid s)\bigr)
\frac{e_{\theta,n}(a\mid s)}{T_\theta(a\mid s)}
\;\le\;\frac{\Delta_\theta}{\delta}\,
\bigl\|e_{\theta,n}(\cdot\mid s)\bigr\|_1 .
\label{eq:crux-cancellation}
\end{equation}
Probability conservation thus makes the first-order log-loss sensitive
only to the \emph{gap between the scoring and construction laws}, not to
the prediction error itself; the full proof, including the second-order
term, follows.

\begin{proof}[Proof of Lemma~\ref{lem:cross-pred}]
The clipped frozen predictor $\widehat T^{\mathrm{CTW},\alpha}_{\theta,n}$
is a function of its independent construction path and is therefore
independent of $Z_{1:n}$. Condition on the construction path; with
$g_{\theta,n}(s,a):=\log\frac{\widehat T^{\mathrm{CTW},\alpha}_{\theta,n}(a\mid s)}{T_\theta(a\mid s)}$,
Lemma~\ref{lem:mixing} at $\theta_0$ gives
\begin{equation}
\mathbb E\!\left[\Bigl|\tfrac1{n-L}\textstyle\sum_t g_{\theta,n}(S_t,Z_t)\Bigr|
\,\Big|\,\widehat T^{\mathrm{CTW},\alpha}_{\theta,n}\right]
\le|m_{\theta,n}|
+\frac{C_{\mathrm{add}}}{\sqrt n}\bigl(\mathbb E_{\theta_0}g_{\theta,n}^2\bigr)^{1/2},
\label{eq:cross-cond}
\end{equation}
where $m_{\theta,n}=\sum_s\pi_{\theta_0}(s)\sum_a T_{\theta_0}(a\mid s)\,g_{\theta,n}(s,a)$.
By $T_\theta\ge\delta$ and $\widehat T^{\mathrm{CTW},\alpha}_{\theta,n}\ge\alpha'$,
Taylor's bound for $\log(1+x)$ gives
$|g_{\theta,n}(s,a)-e_{\theta,n}(a\mid s)/T_\theta(a\mid s)|\le C\,e_{\theta,n}(a\mid s)^2$,
so
\[
|m_{\theta,n}|\le\Bigl|\sum_s\pi_{\theta_0}(s)\sum_a T_{\theta_0}(a\mid s)\tfrac{e_{\theta,n}(a\mid s)}{T_\theta(a\mid s)}\Bigr|
+C\sum_s\pi_{\theta_0}(s)\|e_{\theta,n}(\cdot\mid s)\|_2^2 .
\]
In the first term insert $T_{\theta_0}=T_\theta+(T_{\theta_0}-T_\theta)$;
the part weighted by $T_\theta$ vanishes by the normalization
$\sum_a e_{\theta,n}(a\mid s)=0$ of~\eqref{eq:pred-error-sumzero}, leaving
$|\sum_a T_{\theta_0}(a\mid s)e_{\theta,n}(a\mid s)/T_\theta(a\mid s)|
\le C\Delta_\theta\|e_{\theta,n}(\cdot\mid s)\|_1$. Taking expectations
and using the prediction-moment bounds~\eqref{eq:pmc-l1}--\eqref{eq:pmc-l2},
$\mathbb E|m_{\theta,n}|\le C(\Delta_\theta n^{-1/2}+\log n/n)$. For the
fluctuation term, $\mathbb E_{\theta_0}g_{\theta,n}^2\le C\sum_s\pi_{\theta_0}(s)\|e_{\theta,n}(\cdot\mid s)\|_2^2$,
so by Jensen and~\eqref{eq:pmc-l2},
$\mathbb E[(\mathbb E_{\theta_0}g_{\theta,n}^2)^{1/2}]\le C(\log n/n)^{1/2}$;
dividing by $\sqrt n$ contributes $\sqrt{\log n}/n\le\log n/n$. Taking the
outer expectation in~\eqref{eq:cross-cond} and combining the two bounds
proves~\eqref{eq:cross-pred-rate}.
\end{proof}

\subsection{Proof of the prediction-moment lemma}
\label{app:pmc-proof}

The proof decomposes the predictor error via the BCT
posterior-predictive representation of \cite{Kon24} into a
correct/overfit block, controlled by finite-state concentration
(Lemma~\ref{lem:mixing}), and an underfit block, whose posterior weight
decays exponentially by the separation gap~$\kappa$; clipping is shown
to be inactive on a high-probability event and to cost only $O(n^{-1})$.
The underfit step rests on the concentration of an empirical conditional
mutual information, which we isolate as a lemma. For an
internal node $s$ of length $\ell$ with children $cs$ ($c\in\mathcal A$
the child-selecting symbol), let $\widehat H(Y\mid u):=-\sum_a\tfrac{b_{u,a}}{n_u}\log\tfrac{b_{u,a}}{n_u}$
be the empirical conditional entropy at context $u$,
with the conventions $0\log0=0$ and $\widehat H(Y\mid u):=0$ when
$n_u=0$, and
\begin{equation}
\widehat J_s
:=\frac{n_s}{N}\Bigl(\widehat H(Y\mid s)-\sum_{c\in\mathcal A}\frac{n_{cs}}{n_s}\widehat H(Y\mid cs)\Bigr),
\label{eq:emp-cmi}
\end{equation}
set to zero when $n_s=0$, where $N=n-L$ is the number
of scoring positions of the proof of Lemma~\ref{lem:pmc},
the empirical version of $\pi_\theta(s)\,I_\theta(Y_t;C_t\mid S^{(\ell)}_t=s)$.

\begin{lemma}[Empirical conditional MI concentration]
\label{lem:emp-cmi}
Under Assumptions~\ref{ass:source}--\ref{ass:reg} and
\ref{ass:occupancy}, there is a constant $C<\infty$
such that, uniformly in $s\in\mathcal I(T^\star)$ and
$\theta\in\mathcal N(\theta_0)$,
\begin{equation}
\mathbb E\bigl|\widehat J_s-\pi_\theta(s)\,I_\theta(Y_t;C_t\mid s)\bigr|^2
\le\frac{C}{n}.
\label{eq:emp-cmi-mse}
\end{equation}
\end{lemma}

\begin{proof}
Write $J_s:=\pi_\theta(s)\,I_\theta(Y_t;C_t\mid s)$ and
let
$G_s':=\{n_s\wedge\min_c n_{cs}\ge\tfrac12 N\pi_{\min}\}\cap\{\min_{u\in\{s\}\cup\{cs\}_c,\,a}\widehat T(a\mid u)\ge\tfrac\delta2\}$,
where $\widehat T(a\mid u)=b_{u,a}/n_u$.
Every empirical input of $\widehat J_s$ is built from
occupation counts of the finite-state chain. With the counts of the
proof of Lemma~\ref{lem:pmc} below, $n_u/N$ and $b_{u,a}/N$ are
empirical averages of the deterministic functions
$\mathbf 1\{u\preceq s\}$ and
$\mathbf 1\{u\preceq s\}\mathbf 1\{a'=a\}$ of the pair $(s,a')$, so
Lemma~\ref{lem:mixing} applies to each of them and gives
\[
\operatorname{Var}\bigl(n_u/N\bigr)=O(n^{-1}),
\qquad
\operatorname{Var}\bigl(b_{u,a}/N\bigr)=O(n^{-1}),
\]
with means $\pi_\theta(u)$ and $\pi_\theta(u)T_\theta(a\mid u)$, where
$T_\theta(a\mid u):=\Pr_\theta(X_t=a\mid u\preceq S_t)$ is the
stationary conditional probability given $u$. We work with occupation
measures rather than with martingale increments because, for an
insufficient context $u$, the one-step conditional law of the chain
given the whole past is $T_\theta(\cdot\mid S_t)$ and not
$T_\theta(\cdot\mid u)$.

By Chebyshev's inequality, each of the events
$\{n_u\ge\tfrac12N\pi_\theta(u)\}$, for the finitely many contexts
$u\in\{s\}\cup\{cs\}_c$, fails with probability $O(n^{-1})$; on their
intersection every count is at least $\tfrac12N\pi_{\min}$. There
$\widehat T(a\mid u)=(b_{u,a}/N)/(n_u/N)$ is a ratio of two empirical
averages whose denominator is bounded below by $\pi_{\min}/2$, and on
that region the map $(p,q)\mapsto p/q$ is Lipschitz, so
$\mathbb E[(\widehat T(a\mid u)-T_\theta(a\mid u))^2\mathbf 1]=O(n^{-1})$;
since $T_\theta(a\mid u)\ge\delta$, being an average of transition
probabilities that are all $\ge\delta$, this also gives
$\mathbb P(\widehat T(a\mid u)<\tfrac\delta2)=O(n^{-1})$. A union bound
over the finitely many pairs $(u,a)$ gives
$\mathbb P(G_s'^{\,c})=O(n^{-1})$. The same Lipschitz argument
applies to the count ratios $n_{cs}/n_s$ that enter $\widehat J_s$, so
all the empirical inputs have $O(n^{-1})$ second moments about their
true values.

On $G_s'$ the frequencies stay in the compact region
$\{\widehat T(\cdot\mid u)\ge\delta/2\}$, on which
$\widehat J_s=\Phi(\xi)$ is a fixed $C^\infty$ function of the vector
$\xi=(\{\widehat T(a\mid u)\},\{n_{cs}/n_s\},n_s/N)$ of empirical inputs, with
gradient bounded uniformly (the entropies are
smooth away from the simplex boundary, which $\widehat T\ge\delta/2$
avoids, and the shared denominators
$n_s,n_{cs}\ge\tfrac12 N\pi_{\min}$ on
$G_s'$ keep the ratio derivatives bounded). By the
second-moment bounds
above the input vector has total second moment
$\mathbb E\|\xi-\bar\xi\|_2^2\mathbf 1_{G_s'}=O(n^{-1})$ about its true
value $\bar\xi$, and $J_s=\Phi(\bar\xi)$. Since $\Phi$
is Lipschitz on that region,
$|\widehat J_s-J_s|\le\|\nabla\Phi\|_\infty\|\xi-\bar\xi\|_2$ on
$G_s'$, so
\[
\mathbb E\bigl[|\widehat J_s-J_s|^2\mathbf 1_{G_s'}\bigr]
\le\|\nabla\Phi\|_\infty^2\,
\mathbb E\bigl[\|\xi-\bar\xi\|_2^2\mathbf 1_{G_s'}\bigr]
=O(n^{-1}).
\]
On the complement, $\widehat J_s$ and $J_s$ both lie in
$[0,\log|\mathcal A|]$ deterministically, so
$\mathbb E[|\widehat J_s-J_s|^2\mathbf 1_{G_s'^{\,c}}]
\le(\log|\mathcal A|)^2\,\mathbb P(G_s'^{\,c})=O(n^{-1})$. Adding the
two contributions gives~\eqref{eq:emp-cmi-mse}.
\end{proof}

\begin{proof}[Proof of Lemma~\ref{lem:pmc}]
All randomness in $\widehat T^{\mathrm{CTW},\alpha}_{\theta,n}$ comes from
the construction path $X_{1:n}\sim P_\theta^{(n)}$; expectations below are
under $P_\theta$. Throughout, $S_t:=X_{t-L:t-1}$ denotes
the length-$L$ context of the construction path, $u\preceq S_t$ means
that $u$ is a suffix of $S_t$, and $N:=n-L$ is the number of scoring
positions; we assume $n\ge 2L$, so that $N\ge n/2$. For a context
$u\in\mathcal A^{\ell}$, $0\le\ell\le L$,
write the counts
$n_u=\sum_{t=L+1}^{n}\mathbf 1\{u\preceq S_t\}$ and
$b_{u,a}=\sum_{t=L+1}^{n}\mathbf 1\{u\preceq S_t\}\mathbf 1\{X_t=a\}$,
and
the KT predictor $P_e^u(a)=(b_{u,a}+\tfrac12)/(n_u+\tfrac{|\mathcal A|}{2})$.
Its stationary occupancy is $\pi_\theta(u):=\sum_{w\in\mathcal A^{L}:\,u\preceq w}\pi_\theta(w)$,
so by Assumption~\ref{ass:occupancy} $\pi_\theta(u)\ge\pi_{\min}$ for every
context $u$ of length $\le L$ compatible with the chain.
Fix $w\in\mathcal A^L$ and let $s^\star(w)\in T^\star$ be the minimal-tree
leaf that is a suffix of $w$, of depth $\ell^\star(w)\le L^\star\le L$;
under Assumption~\ref{ass:topology}, $T_\theta(\cdot\mid w)=T_\theta(\cdot\mid s^\star(w))=:T$.
The frozen predictor at a scoring context $w$ is a
convex combination of the KT predictors along the suffix path of $w$.
Indeed, unrolling the one-step recursion~\eqref{eq:ctw-sequential} of
Appendix~\ref{app:ctw} from the root down to the leaf $w$, with the
weights $\beta^{w(i)}\in(0,1)$ of~\eqref{eq:beta} evaluated on the
construction path, gives
\begin{equation}
\widehat T^{\mathrm{CTW},\circ}_{\theta,n}(\cdot\mid w)
=\sum_{i=0}^{L}\gamma_i(w)\,P_e^{w(i)},
\qquad
\gamma_i(w):=\beta^{w(i)}\prod_{j<i}\bigl(1-\beta^{w(j)}\bigr),
\label{eq:mixture-rep}
\end{equation}
where $w(i)$ is the length-$i$ suffix of $w$ and the product is over
the ancestors of $w(i)$ on that path, with the convention
$\beta^{w(L)}=1$. The weights are nonnegative and telescope to one,
and they are computed from the construction path alone. Subtracting
$T$ and splitting the sum at the minimal-tree leaf gives
\[
\widehat T^{\mathrm{CTW},\circ}_{\theta,n}(\cdot\mid w)-T
=\underbrace{\sum_{i\in O}\gamma_i(w)\bigl(P_e^{w(i)}-T\bigr)}_{E_O}
+\underbrace{\sum_{i\in U}\gamma_i(w)\bigl(P_e^{w(i)}-T\bigr)}_{E_U},
\quad \rho:=\sum_{i\in U}\gamma_i(w),
\]
where $\widehat T^{\mathrm{CTW},\circ}_{\theta,n}$ is the unclipped
predictor, the underfit set is
$U=\{0\le i<\ell^\star(w)\}$, and the correct/overfit set is
$O=\{\ell^\star(w)\le i\le L\}$.

\emph{Step 0 (clipping is inactive with high probability).}
On $\mathcal G_w:=\{\max_a|\widehat T^{\mathrm{CTW},\circ}_{\theta,n}(a\mid w)-T_\theta(a\mid w)|\le\delta/2\}$,
every coordinate is $\ge\delta/2>\alpha$, so clipping leaves the
conditional unchanged. Both vectors being probability vectors,
$\|\widehat T^{\mathrm{CTW},\alpha}_{\theta,n}(\cdot\mid w)-\widehat T^{\mathrm{CTW},\circ}_{\theta,n}(\cdot\mid w)\|_1\le2\,\mathbf 1_{\mathcal G_w^c}$
and $\le4\,\mathbf 1_{\mathcal G_w^c}$ in $\ell_2^2$. By the second-moment
(Chebyshev) count bound of Step~1 and $\mathbb E\rho=O(n^{-1})$ (Step~3,
applied to the unclipped mixture and hence not circular),
$\mathbb P(\mathcal G_w^c)=O(n^{-1})$; averaging these corrections
against $\pi_{\theta_0}$ adds only $O(n^{-1})$, so it suffices to bound
the unclipped predictor.

\emph{Step 1 (KT error at a correct context).}
Let $u=w(i)$, $i\in O$, so $T_\theta(\cdot\mid u)=T$. Splitting
$P_e^u(a)-T(a)=\frac{b_{u,a}-n_uT(a)}{n_u+|\mathcal A|/2}
+\frac{\tfrac12-\tfrac{|\mathcal A|}{2}T(a)}{n_u+|\mathcal A|/2}$, the
increments $\xi_t(a)=\mathbf 1\{u\preceq X_{t-L:t-1}\}(\mathbf 1\{X_t=a\}-T_\theta(a\mid u))$
are martingale differences (the indicator is predictable and
$\mathbb E[\mathbf 1\{X_t=a\}\mid\mathcal F_{t-1}]=T_\theta(a\mid u)$ on
$\{u\preceq X_{t-L:t-1}\}$), so
$\mathbb E[(b_{u,a}-n_uT(a))^2]=\mathbb E[n_u]T(a)(1-T(a))\le\tfrac14 n$ and
$\mathbb E|b_{u,a}-n_uT(a)|\le\tfrac12\sqrt n$. On the occupancy event
$G_u=\{n_u\ge\tfrac12 N\pi_\theta(u)\}$, which has
$\mathbb P(G_u^c)=O(n^{-1})$ by Lemma~\ref{lem:mixing} and Chebyshev
($n_u$ is an additive functional with mean
$N\pi_\theta(u)\ge N\pi_{\min}$),
every denominator is $\ge\tfrac12 N\pi_{\min}$; bounding the centred and
KT-bias parts there and the complement by $\|P_e^u-T\|_1\le2$
(resp.\ $\le4$),
\[
\mathbb E\|P_e^u-T\|_1\le\frac{|\mathcal A|}{\pi_{\min}\sqrt n}+O(n^{-1})=O(n^{-1/2}),
\qquad
\mathbb E\|P_e^u-T\|_2^2\le\frac{|\mathcal A|}{\pi_{\min}^2\,n}+O(n^{-1})=O(n^{-1}).
\]

\emph{Step 2 (correct/overfit block).}
Each $i\in O$ has $w(i)$ refining $s^\star(w)$, so $T_\theta(\cdot\mid w(i))=T$
and Step~1 applies. Using $\gamma_i\le1$, $|O|\le L+1$, and the convexity
estimate $\|\sum_{i\in O}\gamma_i v_i\|_2^2\le\sum_{i\in O}\gamma_i\|v_i\|_2^2$,
$\mathbb E\|E_O\|_1=O((L{+}1)n^{-1/2})$ and $\mathbb E\|E_O\|_2^2=O((L{+}1)n^{-1})$.

\emph{Step 3 (underfit block is negligible).}
Since $\|E_U\|_1\le2\rho$ and $\|E_U\|_2^2\le4\rho$, it suffices that
$\mathbb E\rho=O(n^{-1})$.

\emph{(3a) Reduction to a log-odds sum.} By the $|\mathcal A|$-ary CTW
recursion $P_w^s=\tfrac12 P_e^s+\tfrac12\prod_{c}P_w^{cs}$
\cite{Kon24} and $P_w^{cs}\ge\tfrac12 P_e^{cs}$,
the posterior probability of treating $s$ as a leaf, given that the suffix
path reaches $s$, is
$\tfrac12 P_e^s/P_w^s\le P_e^s/\prod_c P_w^{cs}\le 2^{|\mathcal A|}\,e^{\Lambda_s}$
with $\Lambda_s:=\log P_e^s-\sum_{c}\log P_e^{cs}$. The unconditional
weight $\gamma_{|s|}(w)$ of stopping at $s$ is this quantity times the
product of the split probabilities strictly above $s$, each in $[0,1]$;
hence $\gamma_{|s|}(w)\le 2^{|\mathcal A|}e^{\Lambda_s}$ and, summing over
the at most $\ell^\star(w)\le L$ underfit ancestors,
$\rho\le 2^{|\mathcal A|}\sum_{s\in\mathcal I(T^\star)}e^{\Lambda_s}$.

\emph{(3b) Deterministic log-odds expansion.} The $|\mathcal A|$-ary KT
redundancy identity \cite[Lemma~4.2]{Kon24},
$-\log P_e^u=n_u\widehat H(Y\mid u)+\tfrac{|\mathcal A|-1}{2}\log(n_u+1)+\epsilon_u$
with $|\epsilon_u|\le c(|\mathcal A|)$ deterministic (the $+1$ keeps the
bound valid at $n_u=0$, where $P_e^u=1$), applied to $u=s$ and
to each child $u=cs$ and subtracted (using $\sum_c n_{cs}=n_s\le n$), gives
$\Lambda_s=-N\widehat J_s+r_s$, where $\widehat J_s\ge0$ is the empirical
conditional mutual information~\eqref{eq:emp-cmi} and the remainder
$|r_s|\le\tfrac{|\mathcal A|-1}{2}\bigl(\log(n_s+1)+\sum_c\log(n_{cs}+1)\bigr)+(|\mathcal A|+1)c(|\mathcal A|)=O(|\mathcal A|^2\log n)$
is deterministic.

\emph{(3c) Concentration of $\widehat J_s$.} By
Lemma~\ref{lem:emp-cmi}, $\mathbb E|\widehat J_s
-\pi_\theta(s)I_\theta(Y_t;C_t\mid s)|^2\le C/n$.

\emph{(3d) Conclusion.} Assumption~\ref{ass:gap} gives
$\pi_\theta(s)I_\theta(Y_t;C_t\mid s)\ge\kappa$, so
$\{\widehat J_s<\tfrac\kappa2\}$ implies
$|\widehat J_s-\pi_\theta(s)I_\theta(Y_t;C_t\mid s)|>\tfrac\kappa2$,
and Markov's inequality applied to~\eqref{eq:emp-cmi-mse} gives
\[
\mathbb P\bigl(\widehat J_s<\tfrac\kappa2\bigr)
\le\frac{4}{\kappa^2}\,
\mathbb E\bigl|\widehat J_s-\pi_\theta(s)I_\theta(Y_t;C_t\mid s)\bigr|^2
\le\frac{4C}{\kappa^2 n}.
\]
Using $\rho\le1$ always
and, on the event $\{\widehat J_s\ge\tfrac\kappa2\ \forall s\}$, the
deterministic bound
$e^{\Lambda_s}\le e^{-N\kappa/2+O(\log n)}$ from (3b),
\[
\mathbb E\rho\le\mathbb P\bigl(\exists s\in\mathcal I(T^\star):\widehat J_s<\tfrac\kappa2\bigr)
+2^{|\mathcal A|}|\mathcal I(T^\star)|\,e^{-N\kappa/2+O(\log n)}=O(n^{-1}).
\]

\emph{Conclusion.} Adding Steps~0--3, the per-context errors are
$\mathbb E\|\widehat T^{\mathrm{CTW},\alpha}_{\theta,n}(\cdot\mid w)-T\|_1=O(n^{-1/2})$
and $\mathbb E\|\widehat T^{\mathrm{CTW},\alpha}_{\theta,n}(\cdot\mid w)-T\|_2^2=O(n^{-1})$,
uniformly in $w$ and in $\theta\in\mathcal N(\theta_0)$, with constants
depending only on $(|\mathcal A|,L,\delta,\alpha,\pi_{\min},\kappa,C_{\mathrm{add}})$.
Averaging against $\pi_{\theta_0}$ gives~\eqref{eq:pmc-l1} and the sharper
$O(n^{-1})$ form of~\eqref{eq:pmc-l2}.
\end{proof}

\subsection{Proof of Proposition~\ref{prop:kl-rate}}
\label{app:kl-assembly}

\begin{proof}
Consider first the clipped estimator~\eqref{eq:dir-obj}. Take
expectations of absolute values in the
decomposition~\eqref{eq:master-decomp}, whose remainder satisfies
$|R_{\mathrm{bd}}|\le CL\log n/n$ deterministically (technical
conventions of Appendix~\ref{app:rate-anatomy}). Lemma~\ref{lem:oracle}
gives $\mathbb E|A_n|\le C_{\mathrm{or}}\varepsilon/\sqrt n$ (the
$\tfrac1n$-versus-$\tfrac1{n-L}$ discrepancy between~\eqref{eq:master-decomp}
and the lemma is a deterministic $O(L\varepsilon/n)$, absorbed into
$R_{\mathrm{bd}}$'s order); Lemma~\ref{lem:self-red} gives
$\mathbb E|B^{\mathrm{self}}_n|\le C_{\mathrm{self}}\log n/n$; and
Lemma~\ref{lem:cross-pred} with $\theta=\theta_0+\varepsilon u$, whose
prediction-moment hypothesis holds by Lemma~\ref{lem:pmc}, together with
$\Delta_\theta\le C_T\varepsilon$ gives
$\mathbb E|B^{\mathrm{cross}}_n|\le C(\varepsilon/\sqrt n+\log n/n)$
(same remark on the normalization). The triangle inequality
proves~\eqref{eq:kl-rate} with
$\widehat{\bar D}^{\,\alpha}_{n,\varepsilon}(u)$ in place of
$\widehat{\bar D}_{n,\varepsilon}(u)$.

\emph{From the clipped to the unclipped estimator.} It remains to
compare $\widehat{\bar D}^{\,\alpha}_{n,\varepsilon}(u)$ with the actual
output $\widehat{\bar D}_{n,\varepsilon}(u)$ of
Algorithm~\ref{alg:ctw-kl}, i.e.~\eqref{eq:dir-obj} with the unclipped
cross scorer $\widehat P^{\mathrm{CTW}}_{X_{1:n},L}$. Both cross scorers
have per-symbol probabilities bounded below, by
$(2n+|\mathcal A|)^{-1}$ for the unclipped scorer
(Appendix~\ref{app:ctw}) and by $\alpha'$ after clipping, so each
per-symbol log-ratio is $O(\log n)$ and
$\bigl|\widehat{\bar D}_{n,\varepsilon}(u)
-\widehat{\bar D}^{\,\alpha}_{n,\varepsilon}(u)\bigr|\le C\log n$
deterministically. On the event
$\bigcap_{w\in\mathcal A^L}\mathcal G_w$ of Step~0 of the proof of
Lemma~\ref{lem:pmc}, clipping is inactive at every scoring context, so
the two estimates can differ only through their warm-up factors, i.e.\
by at most $CL\log n/n$ deterministically. Since
$\mathbb P(\bigcup_w\mathcal G_w^c)=O(n^{-1})$ by Step~0 and a union
bound over the finitely many contexts,
\[
\mathbb E\bigl|\widehat{\bar D}_{n,\varepsilon}(u)
-\widehat{\bar D}^{\,\alpha}_{n,\varepsilon}(u)\bigr|
\le\frac{CL\log n}{n}+C\log n\cdot O(n^{-1})
=O\Bigl(\frac{\log n}{n}\Bigr),
\]
which is absorbed into the right-hand side of~\eqref{eq:kl-rate}. The
proposition therefore holds as stated for the output of
Algorithm~\ref{alg:ctw-kl}.
\end{proof}

\section{Proofs of the Directional and Matrix FIM Rates}
\label{app:rate-fim}

This appendix proves Theorem~\ref{thm:fim-rate}.
The directional rates follow from Proposition~\ref{prop:kl-rate} and
Lemma~\ref{lem:kl-fim}; the least-squares assembly is shared by both
parts of the theorem.

\begin{proof}[Proof of Theorem~\ref{thm:fim-rate}]
\emph{Step 1 (directional rates).} Fix a direction $u_j$
and write $v^\star_j:=u_j^\top I(\theta_0)u_j$. By
Proposition~\ref{prop:kl-rate} and Lemma~\ref{lem:kl-fim} (applied with
$\Delta=\varepsilon u_j$),
\[
\mathbb E\bigl|v_j-v^\star_j\bigr|
\le\tfrac{2}{\varepsilon^2}\,
\mathbb E\bigl|\widehat{\bar D}_{n}(P_{\theta_0}\,\|\,P_{\theta_0+\varepsilon u_j})
-\bar D(P_{\theta_0}\,\|\,P_{\theta_0+\varepsilon u_j})\bigr|
+\tfrac{2}{\varepsilon^2}\,\bigl|r(\varepsilon u_j)\bigr|
\le2C_{\mathrm{KL}}\Bigl(\tfrac{1}{\varepsilon\sqrt n}
+\tfrac{\log n}{n\varepsilon^2}\Bigr)+2C_3\varepsilon
=:\eta^{(1)}_{n,\varepsilon}.
\]
For the central measurement, the $C^4$ smoothness assumed in part (b)
puts Lemma~\ref{lem:sym-expansion} at our disposal, so that
$\bigl|\bar D(P_{\theta_0}\|P_{\theta_0+\varepsilon u_j})
+\bar D(P_{\theta_0}\|P_{\theta_0-\varepsilon u_j})
-\varepsilon^2 u_j^\top I(\theta_0)u_j\bigr|\le C_4\varepsilon^4$.
Combining this bound with Proposition~\ref{prop:kl-rate} applied at
$\theta_0\pm\varepsilon u_j$ (both $u_j$ and $-u_j$ are unit
directions),
\[
\mathbb E\bigl|v^{(2)}_j-v^\star_j\bigr|
\le\tfrac{1}{\varepsilon^2}\sum_{\pm}
\mathbb E\bigl|\widehat{\bar D}_{n}(P_{\theta_0}\,\|\,P_{\theta_0\pm\varepsilon u_j})
-\bar D(P_{\theta_0}\,\|\,P_{\theta_0\pm\varepsilon u_j})\bigr|
+C_4\varepsilon^2
\le2C_{\mathrm{KL}}\Bigl(\tfrac{1}{\varepsilon\sqrt n}
+\tfrac{\log n}{n\varepsilon^2}\Bigr)+C_4\varepsilon^2
=:\eta^{(2)}_{n,\varepsilon}.
\]

\emph{Step 2 (assembly).}
The OLS solution is linear in the measurement vector, and
$\operatorname{unvec}_\triangle(f^\star)=I(\theta_0)$ with
$Uf^\star=v^\star$; since $I(\theta_0)\succeq0$, the PSD projection is the
identity on it, and by non-expansiveness (Lemma~\ref{lem:psd-proj}) and
the matrix Lipschitz bound~\eqref{eq:ls-matrix-lipschitz},
for either family of measurements,
\[
\bigl\|\widehat I_{n,\varepsilon}(\theta_0)-I(\theta_0)\bigr\|_F
\le\sqrt2\,\kappa_U\,\bigl\|v-v^\star\bigr\|_2
\le\sqrt2\,\kappa_U\sum_{j=1}^m\bigl|v_j-v^\star_j\bigr|,
\]
and likewise with $\widehat I^{(2)}_{n,\varepsilon}(\theta_0)$
and $v^{(2)}$. Taking expectations and absorbing $\sqrt2\,\kappa_U m$
into the constant gives~\eqref{eq:fim-matrix-onesided}
and~\eqref{eq:fim-matrix-rate}. Substituting
$\varepsilon_n\asymp n^{-1/4}$ into $\eta^{(1)}_{n,\varepsilon}$
balances its first and third terms at $n^{-1/4}$, the middle term being
of lower order $(\log n)\,n^{-1/2}$, which proves part (a).
Substituting $\varepsilon_n\asymp n^{-1/6}$ into
$\eta^{(2)}_{n,\varepsilon}$ makes the first and third terms $n^{-1/3}$
and the middle term $n^{-2/3}\log n$, which
proves~\eqref{eq:fim-rate-final} and part (b).
\end{proof}

\bibliographystyle{ieeetr}
\bibliography{reference}

\end{document}